\documentclass[twoside, runningheads, oribibl]{llncs}
\newcommand{\qedhere}{\qed}

\usepackage{silence}
\usepackage{amsmath}
\usepackage{amssymb}
\usepackage{cancel}	

\usepackage{cite} %

\usepackage[hyperfootnotes=false]{hyperref}
\usepackage[dvipsnames]{xcolor}

\usepackage{nicefrac}
\usepackage{stmaryrd}
\let\merge\undefined

\usepackage{tikz}
\usetikzlibrary{matrix}
\usetikzlibrary{calc}
\usetikzlibrary{positioning}
\usetikzlibrary{scopes}
\usetikzlibrary{chains}
\usetikzlibrary{matrix}
\usetikzlibrary{decorations.pathmorphing, decorations.pathreplacing, decorations.shapes}
\usetikzlibrary{shapes.geometric}
\usetikzlibrary{intersections}

\usepackage{graphicx}
\usepackage{caption}
\usepackage{subcaption}
\usepackage{wrapfig}

\usepackage{enumitem}

\usepackage{xspace}
\makeatletter
\newcommand*{\eg}{e.g.,\@\xspace}
\newcommand*{\etal}{~\emph{et~al}.\@\xspace}
\newcommand*{\iest}{i.e.,\@\xspace}

\newcommand{\wrt}{w.r.t.\@\xspace}
\makeatother

\usepackage{prettyref}
\newcommand{\rref}[2][]{\prettyref{#2}}
\newrefformat{sec}{Section\,\ref{#1}}
\newrefformat{def}{Def.\,\ref{#1}}
\newrefformat{thm}{Theorem\,\ref{#1}}
\newrefformat{prop}{Proposition\,\ref{#1}}
\newrefformat{lem}{Lemma\,\ref{#1}}
\newrefformat{cor}{Corollary\,\ref{#1}}
\newrefformat{ex}{Example\,\ref{#1}}
\newrefformat{tab}{Table\,\ref{#1}}
\newrefformat{fig}{Fig.\,\ref{#1}}
\newrefformat{app}{Appendix\,\ref{#1}}
\newrefformat{itm}{item \ref{#1}\@\xspace}
\newrefformat{eq}{equation (\ref{#1})}
\newrefformat{rem}{Remark\,\ref{#1}}

\makeatletter
\newcommand\footnoteref[1]{\protected@xdef\@thefnmark{\ref*{#1}}\@footnotemark}
\makeatother

\xspaceaddexceptions{\footnoteref}

\usepackage{ifthen}
\usepackage{xkeyval}
\usepackage{texhelper}

\usepackage{maths}
\usepackage{syntax}
\usepackage{semantics}
\usepackage{calculus}
\usepackage{misc}

\usepackage{xproofs} %

\spnewtheorem*{notation}{Notation}{\itshape}{\rmfamily}

\let\oldsubsubsection\subsubsection
\renewcommand*{\subsubsection}[1]{\oldsubsubsection{#1.}}

\newcommand{\ctrlNotify}{\progtt{noti}}
\newcommand{\ctrlUpdate}{\progtt{updt}}
\newcommand{\ctrlVelocity}{\progtt{velo}}
\newcommand{\ctrlDistance}{\progtt{dist}}

\newcommand{\tarvelo}{v_{tar}}
\newcommand{\maxvelo}{V}
\newcommand{\mespos}{m}
\newcommand{\waitvar}{w}

\newcommand{\safevelo}[1]{\nicefrac{#1}{\epsilon}}
\newcommand{\maxdist}{\epsilon \maxvelo}

\newcommand{\parameters}[1]{(#1)}
\newcommand{\invariant}[1]{\textnormal{\uppercase{#1}}}

\newcommand{\fAssume}{\A}
\newcommand{\lCommit}{\C}

\tikzstyle{disc}=[draw, shape=circle, inner sep=.05cm]
\tikzstyle{ghost}=[shape=circle, inner sep=0.05cm]
\tikzstyle{continuous}=[decorate, decoration={coil, aspect=0, amplitude=.2mm, segment length=4pt}]
\tikzstyle{badtrace}=[dashed, BrickRed]

\tikzstyle{comarrow}=[dashed, -stealth, MidnightBlue]
\tikzstyle{noisy}=[comarrow, solid, decorate, decoration={zigzag, segment length=2mm, amplitude=0.4mm, post length= 1mm}, -stealth]

\newcommand{\notcheckmark}{\checkmark\!\!\!\raisebox{2.2pt}{\textbf{\tiny\textbackslash}}}

\newcommand{\distpicto}{\text{\scriptsize$|\kern-4pt\leftrightarrow\kern-4pt|$}}

\newcommand{\success}{\textcolor{OliveGreen}{$\checkmark$}}
\newcommand{\deny}{\textcolor{BrickRed}{$\notcheckmark$}}
\newcommand{\termination}{\textcolor{BrickRed}{$\times$}}
\newcommand{\measurment}{\textcolor{MidnightBlue}{$\distpicto$}}

\newsavebox\disc
\sbox\disc{\tikz[baseline=-0.5ex]{\node[disc] {};}}

\newsavebox\continuous
\sbox\continuous{\tikz[baseline=-0.5ex]{\draw[continuous] (0,0) -- (.5,0);}}

\newsavebox\communication
\sbox\communication{\tikz[baseline=-0.5ex]{\draw[comarrow] (0,0) -- (.5,0);}}

\newsavebox\lossy
\sbox\lossy{\tikz[baseline=-0.5ex]{\draw[comarrow] (0,0) -- (.5,0) node [right, inner sep = 0] {\termination};}}

\newsavebox\noisy
\sbox\noisy{\tikz[baseline=-0.5ex]{\draw[noisy] (0,0) -- (.5,0);}}

\newsavebox\crash
\sbox\crash{\tikz[baseline=-0.5ex]{\draw[badtrace] (0,0) -- (.5,0);}}

\newif\ifblind\blindfalse
\blindfalse

\newif\iflongversion\longversiontrue

\begin{document}

\title{Dynamic Logic of\\ Communicating Hybrid Programs}

\ifblind
\author{}\institute{}
\else
\author{Marvin Brieger\inst{1}\orcidID{0000-0001-9656-2830}, Stefan Mitsch\inst{2}\orcidID{0000-0002-3194-9759}, Andr\'e Platzer\inst{2, 3}\orcidID{0000-0001-7238-5710}}

\institute{
  LMU Munich, Germany \\
  \and
  Carnegie Mellon University, Pittsburgh, USA \\
  \and
  Karlsruhe Institute of Technology, Germany \\
  \email{marvin.brieger@sosy.ifi.lmu.de}\quad \email{smitsch@cs.cmu.edu}\quad \email{platzer@kit.edu}
}
\fi

\maketitle

\begin{abstract}
	This paper presents a dynamic logic \dLCHP for compositional deductive verification of communicating hybrid programs (CHPs).
	CHPs go beyond the traditional mixed discrete and continuous dynamics of hybrid systems by adding CSP-style operators for communication and parallelism.
	A compositional proof calculus is presented that modularly verifies CHPs including their parallel compositions from proofs of their subprograms by assumption-commitment reasoning in dynamic logic.	
	Unlike Hoare-style assumption-commitments,
	\dLCHP supports intuitive symbolic execution via explicit recorder variables for communication primitives.
	Since \dLCHP is a conservative extension of differential dynamic logic \dL, it can be used soundly along with the \dL proof calculus and \dL's complete axiomatization for differential equation invariants.
\end{abstract}


\section{Introduction}

Their prevalence in safety-critical applications and ample technical subtleties make both cyber-physical systems (CPS) verification \cite{Alur1993, Henzinger1996, DBLP:journals/jar/Platzer08, Alur2011} and parallel program verification \cite{OwickiGries1976, LevinGries1981, deRoever2001, AptdeBoerOlderog10} important challenges.
CPS verification complexity stems from subtle interactions of their discrete control decisions and differential equations.
Parallel program verification complexity comes from intricate interactions caused by synchronization via state or communication interdependencies between parallel components.
But their combination becomes intrinsically complicated because parallel CPS are \emph{always} interdependent as they \emph{always} synchronize implicitly by sharing the same global time. 
Moreover, many real-world CPS have heterogeneous subsystems
whose controllers do not operate in lock-step, and their communication is potentially unreliable (see \rref{fig:interaction}).
Unlike hybrid systems model checking \cite{Alur1993,Lynch2003,Frehse2004,Benvenuti2014}---which needs to compute products of communicating parallel automata of significant size even when reducing the size of the subautomata---deductive approaches can be truly compositional.
Existing deductive verification approaches, however, fail to properly tackle the above challenges, since they are restricted to homogeneous subsystems \cite{DBLP:conf/csl/Platzer10},
specific component shapes and interaction patterns \cite{Lunel2019,DBLP:journals/sttt/MullerMRSP18, Kamburjan2020},
do not support symbolic execution \cite{Liu2010,Wang2012,Guelev2017},
or are non-compositional \cite{Liu2010,Wang2012,Guelev2017},
where even explicit attempts on compositionality \cite{Wang2012,Guelev2017} turn out to be non-compostional~again.

\begin{wrapfigure}[19]{R}{.43\textwidth}
	\vspace{-.5\baselineskip}
	\centering
	\begin{tikzpicture}[node distance=.8cm]

    \node[disc]	(l-0) 						{};
    \node[disc] (l-1) [right= of l-0] 		{};	
    \node[disc] (l-2) [right= of l-1] 		{};
    \node[disc] (l-3) [right= .2cm of l-2] {};
    \node[disc] (l-4) [right= .2cm of l-3] {};
    \node		(l-5) [right= .2cm of l-4]		{};

    \node[inner sep=0] (terminate) [below=.4cm of l-3] {\termination};

    \draw[continuous] (l-0) -- (l-1); 
    \draw[continuous] (l-1) -- (l-2);
    \draw[continuous] (l-2) -- (l-3);
    \draw[continuous] (l-3) -- (l-4);
    \draw[continuous] (l-4) -- (l-5.center);


    \node		(f-0) 	[below= of l-0] 		{};
    \node[disc]	(f-0-1) [right= .3cm of f-0] 	{};
    \node[disc] (f-1) 	[right= .3cm of f-0-1] {};
    \node[disc] (f-2) 	[right= .2cm of f-1] 	{};
    \node[disc]	(f-3) 	[below= .85cm of l-4] 		{};
    \node[disc] (f-4) 	[right= .2cm of f-3] 		{};

    \draw[continuous] (f-0.center) -- (f-0-1);
    \draw[continuous] (f-0-1) -- (f-1);
    \draw[continuous] (f-1) -- (f-2);
    \draw[continuous] (f-2) -- (f-3);
    \draw[continuous] (f-3) -- (f-4);


    \node	[left= .5em of l-0] 	{$\progtt{CPS}_1$};
    \node	[left= .5em of f-0] 	{$\progtt{CPS}_2$};


    \draw[comarrow] (l-1) -- (f-1);
    \draw[comarrow] (f-2) -- (l-2);
    \draw[comarrow] (l-3) -- (terminate);
    \draw[noisy] (l-4) -- (f-3);
    

    \coordinate	[below= .3cm of f-0] (axisstart);
    \coordinate [below= .34cm of f-4] (axisend);
    
    \draw[-stealth] ([shift={(.3, 0)}] axisstart) -- ([shift={(-.3, 0)}] axisend) node[right] {time};


    \coordinate [above right=.5 and .3 of l-1] (origin);

    \node[ellipse, draw, gray, minimum width=1.3cm, minimum height=.8cm] at ($(l-0)!0.5!(l-1)$) (e) {};

    \draw[gray] (e.north east) -- (origin);

    \draw[black, ->] (origin) -- +(0, .8) node[left, shift={(0,-.1)}] {\small state};
    \draw[black, ->] (origin) -- +(1.5, 0) node[right] {\small time};

    \node[disc, inner sep=.1em] [above right=.1 and .2 of origin] (d-0) {};
    \node[disc, inner sep=.1em] [above right=.6 and .2 of origin] (d-1) {};
    \node[disc, inner sep=.1em] [above right=.1 and 1.2 of origin] (d-2) {};
    \node[disc, inner sep=.1em] [above right=.4 and 1.2 of origin] (d-3) {};

    \draw[dotted, semithick] (d-0) -- (d-1);
    \draw[dotted, semithick] (d-2) -- (d-3);
    \draw[black] (d-1) to [bend right=30] (d-2);
\end{tikzpicture}
	\vspace*{-1.5em}
	\caption{
		The interaction of $\progtt{CPS}_1$ and $\progtt{CPS}_2$ has (delayed) (\usebox{\communication}),
		lossy (\usebox{\lossy}), and
		noisy (\usebox{\noisy}) communication.
		Discrete change (\usebox{\disc}) is independent of discrete change in parallel unless synchronization takes place. 
		Continuous evolution (\usebox{\continuous}) is not interrupted by parallel discrete behavior.
 	}
	\label{fig:interaction}
\end{wrapfigure}

Neither compositionality to tame complexity by reasoning separately about discrete, continuous, and communication behavior, 
nor generality in expressing models, nor symbolic execution to boost the feasibility of deductive reasoning are dispensable for a practical approach.
Thus, to tackle all three of these challenges, this paper presents $\dLCHP$, 
a \emph{dynamic logic for communicating hybrid programs} (\emph{CHPs}),
that extends differential dynamic logic \dL for hybrid programs of differential equations 
\ifblind
\cite{DBLP:journals/jar/Platzer17}
\else
\cite{DBLP:journals/jar/Platzer08,Platzer10,DBLP:journals/jar/Platzer17,Platzer18} 
\fi
with CSP-style operators for parallel composition and message passing communication \cite{Hoare1978, Hoare85} as well as assumption-commitment (ac) reasoning \cite{Misra1981, AcHoare_Zwiers}.
Parallel CHPs cannot share state but can communicate.
They run synchronously in global time and their communication is instantaneous.

There are two fundamental approaches to reasoning about parallelism.
Parallelism can either be understood via unfolding its explicit operational semantics,
or via its denotational semantics implicitly characterizing parallel behavior by matching behavior of the subprograms on their respective part of the state.
Verification based on the operational semantics unrolls the exponentially many possible interleavings,
as in hybrid automata \cite{Alur1993} or hybrid CSP \cite{Wang2012,Guelev2017}.
Such approaches admit superficial mitigation of the state space explosion but always resorts to product automata after reducing the size of the subautomata \cite{Frehse2004,Henzinger2001,Benvenuti2014},
or merely postpone the burden of reasoning about exponentially many trace interleavings \cite{Wang2012,Guelev2017} since it reveals the internal structure of subprograms \cite{Wang2012}.
In contrast, verification based on the denotational semantics is compositional for discrete programs with ac-reasoning \cite{AcHoare_Zwiers, AcSemantics_Zwiers}, \emph{if} the semantics is compositional and aligns well with the intended reasoning.
This paper generalizes ac-reasoning for the purpose of integration with a compositional hybrid systems logic
\ifblind
\cite{DBLP:journals/jar/Platzer17}.
\else
\cite{DBLP:journals/jar/Platzer08,Platzer10,DBLP:journals/jar/Platzer17,Platzer18}.
\fi

Our central contribution is a sound \emph{compositional} proof calculus for \dLCHP. 
For compositional verification of parallel communication,
it embeds ac-reasoning into dynamic logic via the explicit ac-modality $[ \alpha ] \ac \psi$.
The ac-modality expresses that for all runs of~$\alpha$ whose incoming communication meets assumption $\A$ the outgoing communication fulfills commitment $\C$ and that $\psi$ holds in the final state.
Formulas $\A$ and $\C$ specify the communication behavior of $\alpha$, 
and enable parallel composition if they mutually agree for the contracts of subprograms.
Crucially, these formulas directly interface the communication history via history variables such that communication can remain implicit in the parallel composition axiom just as in the underlying denotational semantics.
Since we prove that \dLCHP is a conservative extension of \dL,
it inherits \dL's complete axiomatization of differential equation invariants \cite{DBLP:conf/lics/PlatzerT18}.

Unlike approaches built on Hoare-logics \cite{Liu2010, Wang2012,Guelev2017},
our calculus supports statement-by-statement symbolic execution
instead of executing them backward.
Since executing communication primitives extends the former history,
the new history state requires a fresh name. 
As a consequence, it is unsound to adopt Hoare-style ac-reasoning \cite{AcHoare_Zwiers} verbatim with a distinguished variable for the communication history.
Instead, \dLCHP reconciles ac-reasoning and symbolic execution via recorder variables whose evolution is maintainable, as done by our communication axiom along the way.

In summary, we provide the first truly \emph{compositional} approach to the verification of communicating parallel hybrid system models and prove its soundness.
Our logical treatment separates the essentials of discrete, continuous, and communication reasoning into axioms that are simple and modular compared to earlier approaches \cite{Liu2010, Wang2012}.
Even though the technical development is challenging because of a prefix-closed dynamic semantics, 
a subtle static semantics due to the global time and recorder variables,
and the mutual dependency of formulas and programs in dynamic logic 
it remains finally under the hood.
We demonstrate the flexibility of \dLCHP and its proof calculus with an example in autonomous driving considering the challenge of lossy communication, 
where a $\progtt{follower}$ and $\progtt{leader}$ car communicate to form a convoy.

\section{Dynamic Logic of Communicating Hybrid Programs}

We introduce \dLCHP, 
a \emph{dynamic logic} to reason about \emph{communicating hybrid programs} (CHPs).
CHPs combine \dL's hybrid programs 
\ifblind
\cite{DBLP:journals/jar/Platzer08,DBLP:journals/jar/Platzer17}
\else
\cite{DBLP:journals/jar/Platzer08,DBLP:journals/jar/Platzer17,Platzer18} 
\fi
with communication primitives and a parallel operator similar to CSP \cite{Hoare1978,Hoare85}. 
On the logical side, 
$\dLCHP$ introduces ac-reasoning \cite{Misra1981,AcHoare_Zwiers,AcSemantics_Zwiers} into the dynamic logic setup \cite{Harel1979} of \dL, 
allowing compositional reasoning about communication in a way that preserves symbolic execution as an intuitive reasoning principle.

\subsection{Syntax}

The syntax uses channel names $\Chan$ and variables $\V = \RVar \cup \NVar \cup \TVar$ with pairwise disjoint sets of real variables $\RVar$, integer variables $\NVar$, and trace variables $\TVar$. 
The variable $\globalTime \in \RVar$ is designated to reflect the global time.
By convention $x, y, t \in \RVar$, $\intVar, \intVar_i \in \NVar$, $\historyVar, \historyVar_i \in \TVar$, $ch, ch_i \in \Chan$, and $\arbitraryVar, \arbitraryVar_i \in \V$.
Notions $\FV(\cdot)$ of free and $\BV(\cdot)$ of bound variables in formulas and programs are defined as usual by simultaneous induction with the syntax (see \rref{app:staticSemantics}).
$\V(\cdot)=\FV(\cdot)\cup\BV(\cdot)$ is the set of all variables, 
whether read or written. 

All real variables $\RVar$ can be read \emph{and} written in programs but the global time~$\globalTime$ is \emph{not meant} to be written manually.%
\footnote{Programs writing the global time $\globalTime$ manually are likely to be meaningless such as $\globalTime \ceq 0 \parOp \globalTime \ceq 1$ that sets $\globalTime$ to $0$ and $1$ in parallel, which fortunately has no runs.} 
Instead, the built-in evolution of~$\globalTime$ with every continuous behavior makes it represent the global flow of time. 
Trace variables are bound by programs when they record communication.
For a program $\alpha$,  
we call the remaining set $(\BV(\alpha) \cap \RVar) \setminus \{ \globalTime \}$ the state of $\alpha$ and say that $\alpha$ operates over a real state.
In parallel compositions $\alpha \parOp \beta$,
the programs $\alpha$ and $\beta$ may communicate explicitly but do \emph{not} share state.

The logical language of $\dLCHP$ features trace algebra to reason about communication behavior following the approach of Zwiers\etal \cite{AcSemantics_Zwiers, Zwiers_Phd}.
During symbolic program execution, 
communication events are collected syntactically in logical trace variables that were explicitly designated to record the history.
This way, the communication behavior of a program can be specified using recorder variables as interface.
In analogy to a distinguished history variable, this interface is crucial to obtain a compositional proof rule for parallel composition \cite{AcSemantics_Zwiers, Hooman1992, deRoever2001}.

\begin{definition}[Terms] \label{def:syntax_terms}
	Real terms $\Rtrm$, integer terms $\Itrm$, channel terms $\Ctrm$, and trace terms $\Ttrm$ are defined by the grammar below, 
	where $c \in \rationals$ is a rational constant, $\ch{} \in \Chan$ a channel name, $\rp_1, \rp_2 \in \polynoms{\rationals}{\RVar}$ are polynomials in $\RVar$, and $\cset \subseteq \Chan$ is a finite set of channel names.
	The set of all terms is denoted by $\Trm$.
	\begin{align*}
		\Rtrm: &&\re_1, \re_2 & \cceq x \mid \rationalConst \mid \re_1 + \re_2 \mid \re_1 \cdot \re_2 \mid \val{\at{\te}{\ie}} \mid \stamp{\at{\te}{\ie}} \\
		\Itrm: &&\ie_1, \ie_2 & \cceq \intVar \mid 0 \mid 1 \mid \ie_1 + \ie_2 \mid \len{\te}\\
		\Ctrm: &&\ce_1, \ce_2 &  \cceq \ch{} \mid \chan{\at{\te}{\ie}}\\
		\Ttrm: &&\te_1, \te_2 & \cceq \historyVar \mid \epsilon \mid \comItem{ch, \rp_1, \rp_2} \mid \te_1 \cdot \te_2 \mid \te \downarrow \cset
	\end{align*}
\end{definition}

Real terms $\Rtrm$ are formed by arithmetic operators, variables $x$ (including $\globalTime$), and rational constants $\rationalConst$. 
Additionally, $\val{\at{\te}{\ie}}$ accesses the value and $\stamp{\at{\te}{\ie}}$ the timestamp of the $\ie$-th communication in trace $\te$. 
In CHPs only polynomials $\rp \in \polynoms{\rationals}{\RVar} \subset \Rtrm$ in $\RVar$ over rational coefficients occur, 
\iest without trace terms, 
since CHPs operate over a real state.
By convention,~$\rp, \rp_i$ denote terms from $\polynoms{\rationals}{\RVar}$.

For integers, we use Presburger arithmetic (no multiplication) since it is decidable \cite{Presburger1931} and sufficient for our purposes.%
\footnote{Presburger arithmetic is the subset of $\Itrm$ without length computations $\len{\te}$.}
The integer term $\len{\te}$ denotes the length of trace $\te$.
In analogy to $\val{\at{\te}{\ie}}$, the term $\chan{\at{\te}{\ie}}$ accesses the channel name of the $\ie$-th communication in trace $\te$.
The trace term $\epsilon$ represents empty communication, $\te_1 \cdot \te_2$ is the concatenation of trace terms $\te_1$ and $\te_2$, and $\te \downarrow \cset$ the projection of $\te$ onto the set of channel names $\cset \subseteq \Chan$. 
The tuple $\comItem{\ch{}, \rp_1, \rp_2}$ represents communication of $\rp_1$ along channel $\ch{}$ at time $\rp_2$,
where~$\rp_1, \rp_2 \in \polynoms{\rationals}{\RVar}$ since communication is between programs over a real state.

A trace variable $\historyVar$ refers to a sequence of communication events.
By symbolic execution,
proofs collect communication items in trace variables designated to record the history.
Then, communication behavior is specified against these recorder variables using projections onto the channels of interest (see \rref{ex:historyAssertion} below).
This allows specifications to hide the internal structure of programs leading to compositional reasoning in the presence of communication \cite{AcSemantics_Zwiers, Hooman1992, deRoever2001}.

\begin{notation}
	We write $\val{\te}$ to abbreviate $\val{\at{\te}{\len{\te} - 1}}$, \iest access to the value of the last item on trace $\te$.
	Likewise, we use $\stamp{\te}$ and $\chan{\te}$.
\end{notation}

\begin{example} \label{ex:historyAssertion}
	The formula $\len{\historyVar \downarrow \ch{}} > 0 \rightarrow \val{\historyVar \downarrow \ch{}} > 0$ expresses that the last value sent along channel $\ch{}$ recorded by $\historyVar$ is positive.
	The precondition $\len{\historyVar \downarrow \ch{}} > 0$ ensures that the value is accessed only for a non-empty history.
\end{example}

\begin{definition}[Communicating hybrid programs] \label{def:syntax_chps}
	The set $\Chp$ of \emph{communicating hybrid programs} is defined by the grammar below,
	where $x \in \RVar$ for $x, x'$,
	and $\rp \in \polynoms{\rationals}{\RVar}$ is a polynomial in $\RVar$, 
	and $\chi \in \FolRA$ is a formula of first-order real-arithmetic.
	In the parallel composition $\alpha \parOp \beta$, the constituents must not share state, \iest $\V(\alpha) \cap \BV(\beta) = \V(\beta) \cap \BV(\alpha) \subseteq \{\globalTime\} \cup \TVar$.
	\begin{align*}
		\alpha, \beta \cceq \;
		& x \ceq \rp \mid x \ceq * \mid \evolution{}{} \mid \test{} \mid \alpha \seq \beta \mid \alpha \cup \beta \mid \repetition{\alpha} \mid \tag*{\nonrelevant{(standard \dL)}}\\
		& \send{}{}{} \mid \receive{}{}{} \mid \alpha \parOp \beta \tag*{\nonrelevant{(CSP extension)}}
	\end{align*}
\end{definition}

The statement $x \ceq \rp$ instantly changes $x$ to $\rp$ and nondeterministic assignment $x \ceq *$ sets $x$ to an arbitrary value.
Assignment to the global time $\globalTime$ is only meant to be used by axioms.
As in $\dL$ \cite{DBLP:journals/jar/Platzer08}, 
continuous evolution \mbox{$\evolution{}{}$} follows the differential equation $x' = \rp$ for a nondeterministically chosen duration but only as long as the domain constraint $\chi$ is fulfilled.%

In \dLCHP, the global time $\globalTime$ always evolves during a continuous evolution according to $\globalTime' = 1$ even if it does not occur syntactically.
Hence, an evolution $\globalTime' = \rp$ 
is considered ill-formed if $\rp \not\equiv 1$ just like an ODE $x' = 1, x' = 2$ is considered ill-formed.
Since programs operate over a real state,
terms $\rp \in \polynoms{\rationals}{\RVar}$ are polynomials in $\RVar$ and $\chi \in \FolRA$ is a formula of first-order real-arithmetic.

The test $\test{}$ passes if formula $\chi$ is satisfied. 
Otherwise, execution is aborted. 
The sequential composition $\alpha \seq \beta$ executes $\beta$ after $\alpha$. 
The choice $\alpha \cup \beta$ follows~$\alpha$ or $\beta$ nondeterministically, 
and $\repetition{\alpha}$ repeats $\alpha$ zero or more times.

Communication and parallelism are inspired by CSP \cite{Hoare1978}.
The primitive $\send{}{}{}$ sends the value of term $\rp$ along channel $\ch{}$ and $\receive{}{}{}$ receives a value from $\ch{}$ binding it to variable $x$.
For both statements, $\historyVar$ is the trace variable designated to record the communication.
In an ongoing symbolic execution,
this variable can be renamed to refer to the variable keeping the most recent communication history.
In system models, 
all recorder variables are meant to be the same. 
In this case, 
we also write $\send{}{non}{}$ and $\receive{}{non}{}$ instead of $\send{}{}{}$ and $\receive{}{}{}$.

Finally, $\alpha \parOp \beta$ executes $\alpha$ and $\beta$ in parallel for equal duration,
\iest their final states agree on the value of the global time $\globalTime$.
If $\globalTime$ is not manipulated manually,
its increase equals the duration of continuous behavior.
As in CSP,~$\alpha$ and $\beta$ can perform synchronous message passing but cannot share state.%
\footnote{The global time $\globalTime$ and trace variables $\TVar$ are not considered as program state.}
All programs participating in communication over a channel must agree on all communication along that channel and share the same values and recorder variables at the same time,
\iest message passing does not consume time.
Since shared recorder variables agree on the communication for all subprograms,
they provide the interface that allows for decomposition of parallel behavior.
If the recorders are not the same as in $\send{}{\historyVar_0}{} \parOp \receive{}{\historyVar_1}{}$,
there are no runs.
The need for matching recorders in parallel composition must not be confused with renaming of the history by symbolic execution along the sequential structure of programs.

As usual in CSP \cite{Hoare1978},
the syntax does not enforce unidirectional communication such that several programs may send and receive on the same channel at the same time as long as they agree on the recorder variables and values.
For example, $\receive{}{}{x} \parOp \receive{}{}{y}$ is a well-formed program.
Its semantics has all runs where $x$ and $y$ receive the same values.
Likewise, $\send{}{}{\rp_1} \parOp \dots \parOp \send{}{}{\rp_n}$ has terminating runs if the values of all $\rp_i$ agree.
Consequently, a receive statement $\receive{}{}{}$ can be replaced with the semantical equivalent $x \ceq * \seq \send{}{}{x}$.

\begin{notation}
	As usual $\ifstat{\varphi}{\{ \alpha \}}$ is short for $(\test{\varphi} \seq \alpha) \cup \test{\neg\varphi}$.
\end{notation}

\begin{example} \label{ex:followerLeader}
	\rref{fig:followerLeader} models two cars in a convoy safely adjusting their speed. 
	From time to time, 
	the $\progtt{leader}$ changes its speed $v_l$ in the range $0$ to $\maxvelo$ and notifies the $\progtt{follower}$ about the change.
	The communication, however, is lossy ($\ch{vel} ! v_l \cup \skipProg$).
	Sending position updates by $\ctrlUpdate$ succeeds at least every $\epsilon$ time units.
	On such an update, the $\progtt{follower}$'s controller $\ctrlDistance$ awakes.
	If the distance~$d$ fell below $\maxdist$, the $\progtt{follower}$ slows down to avoid collision before the next position update.
	
	Regularly, the $\progtt{follower}$ adopts the speed update in $\ctrlVelocity$, but crucially refuses to do so 
	if the last known distance was not safe ($d > \maxdist$).
	If the $\progtt{leader}$ could overwrite the $\progtt{follower}$'s speed, it could cause a future collision (see \rref{fig:convoyPlot} below) even though obeying would be perfectly fine at the moment.
	This is because a subsequent notification of the leader slowing down could be lost.
\end{example}

\begin{figure}[h!tb]
	\vspace*{-3em}
	\begin{small}
		\begin{minipage}{.55\textwidth}
			\begin{align*}
				\ctrlVelocity \equiv\;
					& \ch{vel} ? \tarvelo \seq 
					\ifstat{d > \epsilon \maxvelo}{v_f \ceq \tarvelo} \\
				\ctrlDistance \equiv\; 
					& \ch{pos} ? \mespos \seq 
					d \ceq \mespos - x_f \seq \\
					& \ifstat{d \le \epsilon \maxvelo}{%
						\{v_f \ceq * \seq \test{\orange{0}{v_f}{\safevelo{d}}}\}
					} \\
				\Plant_f \equiv\;
					& \evolution{x_f' = v_f}{non} \\
				\progtt{follower} \equiv\;
					& \big( (\ctrlVelocity \cup \ctrlDistance) \seq \Plant_f \big)^*
			\end{align*}
		\end{minipage}%
		\begin{minipage}{.45\textwidth}
			\begin{align*}
				\chi_{vel} & \equiv \range{0}{v_l}{\maxvelo} \\
				\ctrlNotify & \equiv v_l \ceq * \seq \test{\chi_{vel}} \seq (\ch{vel} ! v_l \cup \skipProg) \\
				\ctrlUpdate & \equiv \ch{pos} ! x_l \seq \waitvar \ceq 0 \\ 
				\Plant_l & \equiv \evolution{\waitvar' = 1, x_l' = v_l}{\waitvar \le \epsilon} \\
				\progtt{leader} & \equiv \big( (\ctrlNotify \cup \ctrlUpdate) \seq \Plant_l \big)^*
			\end{align*}
		\end{minipage}
	\end{small}
	\caption{Models of two moving cars ($\progtt{follower}$ and $\progtt{leader}$) intended to form the convoy $\progtt{follower} \parOp \progtt{leader}$ by parallel composition, communicating target speed.}
	\label{fig:followerLeader}
	\vspace*{-1.5em}
\end{figure}

\begin{wrapfigure}[24]{R}{.45\textwidth}
	\vspace*{-.5\baselineskip}
	\centering
	\begin{tikzpicture}[node distance= 2em]
    \coordinate (origin) at (0,0);
    \coordinate (coorigin) at ([shift={(-.3, -.4)}] origin);
    \coordinate (epsorigin) at ([shift={(0, -.4)}] origin);


    \draw[black, ->] (coorigin) --++ (4, 0) node[right] {\small time};
    \draw[black, ->] (coorigin) --++ (0, 3) node[above] {\small position};


    \node[disc] [above=1.25 of origin] (l-0) {};
    \node[disc] [above right=2 and 1 of origin] (l-0-1) {};
    \node[disc] [above right=2 and 1.66 of origin] (l-1) {};
    \node[disc] [above right=2 and 2 of origin] (l-2) {};
    \node[disc] [above right=2.6 and 2.5 of origin] (l-2-1) {};
    \node[ghost] [above right=2.75 and 3.66 of origin] (l-3) {};

    \node[ghost][above right=2 and 4 of origin] (l-g) {};


    \node[inner sep=0] (terminate-1) [below=.2cm of l-0-1] {\termination};
    \draw[comarrow] (l-0-1) -- (terminate-1);

    \node[inner sep=0] (terminate-2) [below=.2cm of l-2-1] {\termination};
    \draw[comarrow] (l-2-1) -- (terminate-2);


    \node[disc] (f-0) at (origin) {};
    \node[disc] [above right=1.25 and 1.66 of origin] (f-1) {};
    \node[disc] [above right=1.35 and 2 of origin] (f-2) {};
    \node[ghost] [above right=1.75 and 3.33 of origin] (f-3) {};

    \node[ghost] [above right=3.15 and 3.5 of origin] (f-g) {};


    \draw[black] (l-0) -- (l-0-1);
    \draw[black] (l-0-1) -- (l-1);
    \draw[black] (l-1) -- (l-2);
    \draw[black] (l-2) -- (l-2-1);
    \draw[black, name path=slow] (l-2-1) -- (l-3);


    \draw[black] (f-0) -- (f-1);
    \draw[black] (f-1) -- (f-2);
    \draw[black] (f-2) -- (f-3);

    \draw[dotted, gray] (l-2) -- (l-g);
    \draw[badtrace, name path=bad] (f-2) -- (f-g);


    \draw[comarrow] (l-0) -- (f-0);
    \draw[comarrow] (l-1) -- (f-1);
    \draw[comarrow] (l-2) -- (f-2);


    \path[name intersections={of=slow and bad,by=crash}];
    \node[draw, lightgray, circle] (collision) at (crash) {};
    \node [above left=-.2 and -.1 of collision] {$\lightning$};



    \node [below right=-.15 and -.1 of f-0] {{\scriptsize$\ctrlVelocity$} \!\success};
    \node [below=.1 of f-1] {\scriptsize$\ctrlDistance$ \!\measurment};
    \node [below right=-.15 and -.1 of f-2] {{\scriptsize$\ctrlVelocity$} \!\deny};

    \node [above right=0 and .1 of l-0, rotate=90] {\scriptsize$\ctrlNotify$};
    \node [above right=0 and .1 of l-0-1, rotate=90] {\scriptsize$\ctrlNotify$};
    \node [above right=0 and .15 of l-1, rotate=90] {\scriptsize$\ctrlUpdate$};
    \node [above right=0 and .1 of l-2, rotate=90] {\scriptsize$\ctrlNotify$};
    \node [above right=0 and .1 of l-2-1, rotate=90] {\scriptsize$\ctrlNotify$};


    \node [left=.3 of f-0] {$x_f$};
    \node [left=.3 of l-0] {$x_l$};


    \draw[decoration={brace, mirror, raise=3pt}, decorate] (epsorigin) -- node[below=.5em] {$\epsilon$} ++(1.71,0);

    \draw[decoration={brace, mirror, raise=3pt}, decorate] ([shift={(1.71, 0)}] epsorigin) -- node[below=.5em] {$\epsilon$} ++(1.66, 0);

    \draw[dotted, gray] ([shift={(1.71, -.3)}] origin) -- (f-1.south);
    \draw[dotted, gray] ([shift={(3.37, -.3)}] origin) -- ([shift={(-.33, 0)}] l-3.north);


    \coordinate (cont) at ([shift={(3.3, .6)}] coorigin);
    \node [right=0 of cont] {\scriptsize controller};
    \node[disc] [left=0 of cont] {};

    \coordinate (com) at ([shift={(0, -.3)}] cont);
    \node [right=0 of com] {\scriptsize comm.};
    \draw[comarrow] ([shift={(-.5, 0)}] com) -- (com);
\end{tikzpicture}
	\vspace*{-1.5em}
	\caption{
		Plot of example positions $x_f$ and $x_l$ of the cars over time.
		First, speed update is accepted (\success).
		The next update is lost (\usebox{\lossy}).
		After position update (\measurment),
		the $\progtt{follower}$ adjusts its speed.
		Crucially, it conservatively rejects the speed update (\deny) 
		when a crash ($\lightning$) with a slowing $\progtt{leader}$ is possible since speed communication may fail (\usebox{\lossy}) until the next reliable position update is expected, see dashed trajectory (\usebox{\crash}).
		}
	\label{fig:convoyPlot}
\end{wrapfigure}

\begin{definition}[Formulas] \label{def:syntax_formulas}
	The set of \emph{$\dLCHP$ formulas} $\Fml$ is defined by the following grammar, 
	where $\arbitraryVar \in \V$, terms $\expr_1, \expr_2 \in \Trm$ are of equal sort, $\re_i \in \Rtrm$, $\ie_i \in \Itrm$, $\te_i \in \Ttrm$, and the ac-formulas $\A$ and $\C$ do not refer to state and time of~$\alpha$,
	\iest $(\FV(\A) \cup \FV(\C)) \cap \BV(\alpha) \subseteq \TVar$.
	\begin{align*}
		\varphi, \psi, \A, \C \cceq \; 
			& \expr_1 = \expr_2 \mid \re_1 \ge \re_2 \mid \ie_1 \ge \ie_2 \mid \te_1 \preceq \te_2 \mid \neg \varphi \mid \varphi \wedge \psi \mid \varphi \vee \psi \mid \\
			& \varphi \rightarrow \psi \mid \fa{\arbitraryVar} \varphi \mid \ex{\arbitraryVar} \varphi \mid [ \alpha ] \psi \mid [\alpha] \ac \psi
	\end{align*}
\end{definition}

For specifying the (communication) behavior of CHPs, we combine first-order dynamic logic \cite{Harel1979} with ac-reasoning \cite{AcHoare_Zwiers}. 
Equality is defined on each sort of terms. 
On real and integer terms, $\ge$ has the usual meaning.
On trace terms, $\te_1 \preceq \te_2$ means that $\te_1$ is a prefix of $\te_2$. 
There is no order on channel terms.
Quantified variables $\arbitraryVar \in \V$ are of arbitrary sort.
Since our primary interest is safety, we omit the dynamic modality $\langle \alpha \rangle \psi$ and give no dual $\langle \alpha \rangle \ac \psi$ for $[ \alpha ] \ac \psi$. 

Besides the dynamic modality $[ \alpha ] \psi$, $\dLCHP$ prominently features the ac-box $[ \alpha ] \ac \psi$ 
that reshapes Hoare-style ac-reasoning \cite{AcHoare_Zwiers} into the modal approach of dynamic logic.
In an ac-contract $\varphi \rightarrow [ \alpha ] \ac \psi$,
assumption $\A$ and commitment~$\C$ specify $\alpha$'s communication behavior along the interface of recorder variables but without access to $\alpha$'s state and time as required by $(\FV(\A) \cup \FV(\C)) \cap \BV(\alpha) \subseteq \TVar$,
whereas formulas $\varphi$ and $\psi$ 
act as pre- and postcondition as usual.
Formula $[ \alpha ] \ac \psi$  promises that $\C$ holds after each communication event of $\alpha$ assuming $\A$ held before each event. 
Moreover, if the program terminated \emph{and} $\A$ held before and after each communication event, %
the final state satisfies $\psi$.

\begin{example} \label{ex:convoy_safety}
	The safety condition about $\progtt{follower}$ and $\progtt{leader}$ below expresses:
	If they start driving with a distance of at least $d$ and a speed ${\le}\nicefrac{d}{\epsilon}$ that prevents the $\progtt{follower}$ from reaching the $\progtt{leader}$ within $\epsilon$ time units,
	then the cars do never collide.
	\rref{sec:example} shows a proof of this formula.
	\begin{equation*}
		\epsilon \ge 0 \wedge
			\range{0}{v_f}{\safevelo{d}} \wedge
			v_f \le \maxvelo \wedge
			x_f + d < x_l
		\rightarrow [ \progtt{follower} \parOp \progtt{leader} ]
		\, x_f < x_l 
	\end{equation*}
\end{example}

Closed systems,
where communication has an internal partner,
can be specified using boxes $[ \cdot ] \psi$ (see \rref{ex:convoy_safety}) since their safety does not depend on the environment.
Ac-boxes $[ \cdot ] \ac \psi$ come into play when such systems are decomposed since the constituents $\progtt{follower}$ and $\progtt{leader}$ are each other's environment.

\subsection{Semantics} \label{sec:semantics}

CHPs have a denotational linear history semantics merging ideas from \dL \cite{DBLP:journals/jar/Platzer08} and ac-reasoning \cite{AcSemantics_Zwiers} and adding synchronization in the global time.
The basic domains are traces $\traces$ and states $\states$.
A \emph{trace} $\trace \in \traces$ is a finite sequence $(\trace_1, ..., \trace_n)$ of communication events $\trace_i = \comItem{\ch{}_i, a_i, \duration_i}$ with channel $\ch{}_i \in \Chan$, value $a_i \in \reals$, and timestamp $\duration_i \in \reals$ that is \emph{chronological}, \iest $\duration_i \le \duration_j$ for all $1 \le i < j \le n$.
The empty trace is denoted $\epsilon$, the concatenation of traces $\trace_1$ and $\trace_2$ is $\trace_1 \cdot \trace_2$, and for $\cset \subseteq \Chan$, the projection $\trace \downarrow \cset$ is the subsequence of $\trace$ consisting exactly of those $\comItem{\ch{}, a, \duration}$ and $\comItem{\historyVar, \ch{}, a, \duration}$ with $\ch{} \in \cset$.%
\footnote{We use the same operators for corresponding syntax and semantics, \iest $\te \downarrow \cset$ and $\trace \downarrow \cset$ are the projection on $\cset$ for trace term $\te$ and semantic trace $\trace$, respectively.}
By $\trace[pre] \preceq \trace$ and $\trace[pre] \prec \trace$,
we express that $\trace[pre]$ is a prefix or proper prefix of $\trace$, respectively.
A \emph{recorded trace} $\trace \in \recTraces$ is a trace that has an additional recorder variable $\historyVar_i \in \TVar$ for each communication event such that $\trace_i = \comItem{\historyVar_i, \ch{}_i, a_i, \duration_i}$.
Raw traces $\traces$ represent trace terms in a state, whereas
recorded traces originate from programs.

A \emph{state} is a map $\pstate{v} : \V \rightarrow \reals \cup \naturals \cup \traces$ that assigns a value from $\type(\arbitraryVar)$ to each variable $\arbitraryVar \in \V$,
where $\type(\expr) = \mathbb{M}$ if $\expr \in \SortedTrm{\mathbb{M}}$ for $\mathbb{M} \in \{\reals, \naturals, \traces\}$.
The \emph{updated state} $\pstate{v} \subs{\arbitraryVar}{d}$ is defined by $\pstate{v} \subs{\arbitraryVar}{d} = \pstate{v}$ on $\{ \arbitraryVar \}^\complement$ and $\pstate{v} \subs{\arbitraryVar}{d}(\arbitraryVar) = d$.
\emph{State-trace concatenation} $\pstate{v} \cdot \trace$ appends recorded communication $\trace \in \recTraces$ to the corresponding trace variable in $\pstate{v} \in \states$.
It is defined by $\pstate{v} \cdot \trace = \pstate{v}$ on $\TVar^\complement$ and $(\pstate{v} \cdot \trace)(\historyVar) = \pstate{v}(\historyVar) \cdot \trace(\historyVar)$ for all $\historyVar \in \TVar$,
where $\trace(\historyVar)$ denotes the subtrace of $\trace$ consisting of the raw versions $\comItem{\ch{}, a, \duration}$ of communication events $\comItem{\historyVar, \ch{}, a, \duration}$ in $\trace$ recorded by $\historyVar$.

\subsubsection{Term Semantics}

The value $\sem{\expr} \pstate{v} \in \type(\expr)$ of term $\expr$ at state $\pstate{v} \in \states$ is according to its sort $\type(\expr)$ (see \rref{app:semantics}). 
The evaluation of real and integer terms is as usual. 
Additionally, $\val{\at{\te}{\ie}}$ evaluates to the value, $\stamp{\at{\te}{\ie}}$ to the timestamp, and $\chan{\at{\te}{\ie}}$ to the channel name of the $\ie$-th communication event in $\te$ with indices from $0$ to $\len{\te} - 1$.
Moreover, $\len{\te}$ evaluates to the length of $\te$. 
The evaluation of trace terms is aligned with the semantic operators on traces \cite{AcSemantics_Zwiers, Zwiers_Phd}, \eg $\sem{\te \downarrow \cset} \pstate{v} = (\sem{\te} \pstate{v}) \downarrow \cset$ and $\sem{\comItem{\ch{}, \rp_1, \rp_2}} \pstate{v} = \comItem{\ch{}, \sem{\rp_1} \pstate{v}, \sem{\rp_2} \pstate{v}}$.

\subsubsection{Domain of Computation}

The denotational semantics $\sem{\alpha} \subseteq \pDomain$ of a CHP~$\alpha$ has domain $\pDomain = \states \times \recTraces \times \botop{\states}$ with $\botop{\states} = \states \cup \{\bot\}$,
\iest the observables of a CHP started from a state are communication and a final state.
The marker $\bot$ indicates an unfinished execution that either can be continued or was aborted due to a failing test.
Since communication can even be observed from unfinished computations,
a meaningful semantics of communicating programs is prefix-closed and total (see \rref{def:prefixClosedAndTotal} below).
Totality captures that every program can at least start computation even if it aborts immediately like $\test{\false}$ and has not emitted communication initially.
For \rref{def:prefixClosedAndTotal}, we extend the prefix relation $\preceq$ on traces $\recTraces$ to a partial order $\preceq$ on observable behavior $\recTraces \times \botop{\states}$ expressing that $\observable[pre]$ is a prefix of $\observable$ if 
$\big( (w = \pstate[pre]{w} \text{ and } \trace = \trace[pre])$ or $(\pstate[pre]{w} = \bot \text{ and } \trace[pre] \preceq \trace) \big)$.

\begin{definition}[Prefix-closedness and totality] \label{def:prefixClosedAndTotal}
	A set $U \subseteq \pDomain$ is \emph{prefix-closed} if $\computation \in U$ and $\observable[pre] \preceq \observable$ implies $\computation[pre] \in U$. 
	The set is \emph{total} if $\pLeast \subseteq U$ with $\pLeast = \states \times \{\epsilon\} \times \{\bot\}$,
	\iest $\leastComputation \in U$ for every $\pstate{v} \in \states$.
\end{definition}

\subsubsection{Program Semantics}

The semantics of compound programs is compositionally defined in terms of semantical operators:
For $U, M \subseteq \pDomain$,
we define $\botop{U} = \{ (\pstate{v}, \trace, \bot) \mid \computation \in U \}$
and $\computation \in U \continuation M$ if $(\pstate{v}, \trace_1, \pstate{u}) \in U$ and $\pstate{u} \neq \bot$ and $(\pstate{u}, \trace_2, \pstate{w}) \in M$ exists such that $\trace = \trace_1 \cdot \trace_2$.
The operator $\closedComposition$ is for sequential composition.
For $U, M \subseteq \pDomain$, we define $U \closedComposition M = \botop{U} \cup (U \continuation M)$.
Semantic iteration $U^m$ is defined by $U^0 = \pIdentity = \pLeast \cup (\states \times \{ \epsilon \} \times \states)$ and $U^{n+1} = U \closedComposition U^n$ for $n > 0$. 
Accordingly, $\alpha^0 \equiv{} \test{\true}$ and $\alpha^{n+1} \equiv \alpha \seq \alpha^n$ defines syntactic iteration.

Parallel composition $\alpha \parOp \beta$ requires that $\alpha$ and $\beta$ have disjoint bound variables (\rref{def:syntax_chps}) except for $\{\globalTime\} \cup \TVar$,
where they will always agree.
Thus, the merged state $\pstate{w}_\alpha \merge \pstate{w}_\beta \in \botop{\states}$ for states $\pstate{w}_\alpha, \pstate{w}_\beta \in \botop{\states}$ can be unambiguously determined as follows:
$\pstate{w}_\alpha \merge \pstate{w}_\beta = \bot$ if at least one of the states is $\bot$.
Otherwise, define $(\pstate{w}_\alpha \merge \pstate{w}_\beta)(\arbitraryVar) = \pstate{w}_\alpha(\arbitraryVar)$ if $\arbitraryVar \in \BV(\alpha)$ and $(\pstate{w}_\alpha \merge \pstate{w}_\beta)(\arbitraryVar) = \pstate{w}_\beta(\arbitraryVar)$ if $\arbitraryVar \not\in \BV(\alpha)$.%
\footnote{The alternative condition $\arbitraryVar \in \BV(\beta)$ leads to an equivalent definition when $\pstate{w}_\alpha = \pstate{w}_\beta$ on $(\BV(\alpha) \cup \BV(\beta))^\complement$, which is the case for the final states in parallel composition.}
For program $\alpha$, 
the set $\CN(\alpha) \subseteq \Chan$ consists of all channel names occurring in $\alpha$, 
\iest in send $\send{}{}{}$ and receive $\receive{}{}{}$ statements.
The projection $\trace \downarrow \CN(\alpha)$ is abbreviated as $\trace \downarrow \alpha$.
The \emph{semantic parallel operator} is defined as follows for programs $\alpha, \beta \in \Chp$:
\begin{equation*}
	\sem{\alpha} \parOp \sem{\beta} 
	= \Bigg\{ (\pstate{v}, \trace, \pstate{w}_\alpha \merge \pstate{w}_\beta) \in \pDomain
	\;\bigg\vert\; 
	\begin{aligned}
		&\computation[proj=\alpha] \in \sem{\alpha},
		\computation[proj=\beta] \in \sem{\beta}, \\
		&\statetime{\pstate{w}_\alpha} = \statetime{\pstate{w}_\beta},
		\trace = \trace \downarrow (\alpha \parOp \beta)
	\end{aligned} 
	\Bigg\}
\end{equation*}	

Instead of computing explicit interleavings, 
the parallel operator $\parOp$ characterizes the joint communication $\trace$ implicitly via any order that
the subprograms can agree on.
Thereby $\trace \in \recTraces$ rules out non-chronological ordering of communication events that are exclusive to either $\trace \downarrow \alpha$ or $\trace \downarrow \beta$.
Moreover, by $\trace = \trace \downarrow (\alpha \parOp \beta)$, 
the trace $\trace$ must not contain any junk,
\iest communication not caused by $\alpha$ or~$\beta$. 
Communication along joint channels of $\alpha$ and $\beta$ must agree in its recorder variable, value, and timestamp as it occurs in $\trace \downarrow \alpha$ and $\trace \downarrow \beta$.
By $\statetime{\pstate{w}_\alpha} = \statetime{\pstate{w}_\beta}$ both computations need to meet at the same point in global time.%
\footnote{We consider $\statetime{\pstate{w}_\alpha} = \statetime{\pstate{w}_\beta}$ fulfilled if $\pstate{w}_\alpha = \bot$ or $\pstate{w}_\beta = \bot$.}

\begingroup
\allowdisplaybreaks
\begin{definition}[Program semantics]\label{def:programsemantics}
	The semantics $\sem{\alpha} \subseteq \pDomain$ of a program $\alpha \in \Chp$ is inductively defined as follows,
	where $\pLeast = \states \times \{\epsilon\} \times \{\bot\}$ 
	and~$\vDash$ is the satisfaction relation for formulas (\rref{def:formulaSemantics}):
	\begingroup
	\begin{align*}%
		&\sem{x \ceq \rp}
		= \pLeast \cup \{ (\pstate{v}, \epsilon, \pstate{w}) \mid \pstate{w} = \pstate{v} \subs{x}{\sem{\rp} \pstate{v}} \} \\
		&\sem{x \ceq *} 
		= \pLeast \cup \{ (\pstate{v}, \epsilon, \pstate{w}) \mid \pstate{w} = \pstate{v} \subs{x}{a} \text{ where } a \in \reals \} \\
		&\sem{\test{}} 
		= \pLeast \cup \{ (\pstate{v}, \epsilon, \pstate{v}) \mid \pstate{v} \vDash \chi \} \\
		&\sem{\evolution{}{}} 
		= \pLeast \cup \big\{ (\odeSolution(0), \epsilon, \odeSolution(\duration)) \mid
			 \odeSolution(\zeta) \vDash \globalTime' = 1 \wedge x' = \rp \wedge \chi
			 \text{ and } \\
			&\qquad 
			\odeSolution(\zeta) = \odeSolution(0) \text{ on } \{ x, \globalTime\}^\complement
			\text{ for all } \zeta \in [0, \duration] 
			\text{ and a solution }\\
			&\qquad
			\odeSolution : [0, \duration] \rightarrow \states \text{ with } \odeSolution(\zeta)(\arbitraryVar') = \solutionDerivative{\odeSolution}{\arbitraryVar}(\zeta) \text{ for } \arbitraryVar \in \{x,\globalTime\} \big\} \\
		&\sem{\send{}{}{}}
		= \{ \computation \mid \observable \preceq ( \comItem{\historyVar, \ch{}, \sem{\rp} \pstate{v}, \statetime{\pstate{v}}}, \pstate{v} ) \} \\
		&\sem{\receive{}{}{}} 
		= \{ \computation \mid \observable \preceq ( \comItem{\historyVar, \ch{}, a, \statetime{\pstate{v}}}, \pstate{v} \subs{x}{a} ) \text{ where } a \in \reals \} \\
		&\sem{\alpha \cup \beta} 
		= \sem{\alpha} \cup \sem{\beta}\\
		&\sem{\alpha \seq \beta} 
		= \sem{\alpha} \closedComposition \sem{\beta} = \botop{\sem{\alpha}} \cup (\sem{\alpha} \continuation \sem{\beta})\\
		&\sem{\repetition{\alpha}} 
		= \bigcup_{n \in \naturals} \sem{\alpha}^n 
		= \bigcup_{n \in \naturals} \sem{\alpha^n} \\
		&\sem{\alpha \parOp \beta} 
		= \sem{\alpha} \parOp \sem{\beta}
	\end{align*}
	\endgroup
\end{definition}
\endgroup

In the semantics of continuous evolution,
the solution for the ODE gives meaning to the primed variable $x'$ as in \dL \cite{DBLP:journals/jar/Platzer17}.
By $\globalTime' = 1$,
the global time $\globalTime$ always evolves with slope $1$ with every continuous evolution.

The semantics $\sem{\alpha} \subseteq \pDomain$ is prefix-closed and total (\rref{def:prefixClosedAndTotal}) for every program $\alpha$ (see \rref{app:semantics}).
For atomic non-communicating programs,
$\pLeast = \states \times \{\epsilon\} \times \{\bot\}$ ensures prefix-closedness.
If $\test{}$ or $\evolution{}{}$ abort, $\pLeast$ also guarantees totality.
Keeping the unfinished computations $\botop{\sem{\alpha}}$ preserves prefix-closedness of $\alpha \seq \beta$.

\subsubsection{Formula Semantics}

The semantics of the first-order fragment is as usual.
Like in dynamic logic \cite{Harel1979},
the box $[ \alpha ] \psi$ means that $\psi$ is true after all finished computations,
\iest the final state and communication of $\computation \in \sem{\alpha}$ with $\pstate{w} \neq \bot$. 
The ac-box $[ \alpha ] \ac \psi$ additionally means that 
the communication of (un)finished computations fulfills commitment $\C$.
In our modal treatment of ac-reasoning \cite{AcHoare_Zwiers}, 
assumption $\A$ and program $\alpha$ determine the reachable worlds together,
\iest only computations need to be considered whose incoming communication meets $\A$.

\begin{definition}[Formula semantics]\label{def:formulaSemantics}
	The semantics $\sem{\varphi} \subseteq \states$ of a formula $\varphi \in \Fml$ is defined as \(\sem{\varphi} = \{\pstate{v} \mid \pstate{v} \vDash \varphi\}\) using the \emph{satisfaction relation}~$\vDash$.
	The relation $\vDash$ is defined by induction on the structure of $\varphi$ as follows:
	\begin{enumerate}
		\item $\pstate{v} \vDash \expr_1 {=} \expr_2$ if $\sem{\expr_1} \pstate{v} = \sem{\expr_2} \pstate{v}$. Accordingly, for $\re_1 {\ge} \re_2$, $\ie_1 {\ge} \ie_2$, $\te_1 {\preceq} \te_2$
		\item $\pstate{v} \vDash \varphi \wedge \psi$ if $\pstate{v} \vDash \varphi$ and $\pstate{v} \vDash \psi$. Accordingly, for $\neg, \vee, \rightarrow$
		\item $\pstate{v} \vDash \fa{\arbitraryVar} \varphi$ if $\pstate{v} \subs{\arbitraryVar}{d} \vDash \varphi$ for all $d \in \type(\arbitraryVar)$
		\item $\pstate{v} \vDash \ex{\arbitraryVar} \varphi$ if $\pstate{v} \subs{\arbitraryVar}{d} \vDash \varphi$ for some $d \in \type(\arbitraryVar)$
		\item \label{itm:dynBoxSem}
		$\pstate{v} \vDash [ \alpha ] \psi$ if $\pstate{w} \cdot \trace\vDash \psi$ for all $\computation \in \sem{\alpha}$ with $\pstate{w} \neq \bot$ 
		\item \label{itm:acBoxSem}
		$\pstate{v} \vDash [ \alpha ] \ac \psi$ if for all $\computation \in \sem{\alpha}$ the following conditions hold:
		\begin{align}
			\vspace{-.7em}
			&\proppre{\pstate{v}}{\trace} \vDash \A \text{ implies } \pstate{v} \cdot \trace \vDash \C \tag{commit} \label{eq:commit}\\
			&\big( \preeq{\pstate{v}}{\trace} \vDash \A \text{ and } \pstate{w} \neq \bot \big) \text{ implies } \pstate{w} \cdot \trace \vDash \psi \tag{post} \label{eq:post}
		\end{align}%
		Where $U \vDash \varphi$ for a set of states $U \subseteq \states$ and any formula $\varphi \in \Fml$ if $\pstate{v} \vDash \varphi$ for all $\pstate{v} \in U$. 
		In particular, $\emptyset \vDash \varphi$.
	\end{enumerate}
\end{definition}

In \rref{itm:acBoxSem}, \acCommit is checked after each communication event as desired since \emph{all} communication prefixes are reachable worlds by prefix-closedness of the program semantics $\sem{\alpha}$ (\rref{def:prefixClosedAndTotal}).
Via state-trace concatenation $\pstate{v} \cdot \trace$ and $\pstate{w} \cdot \trace$ in \rref{itm:dynBoxSem} and \rref{itm:acBoxSem},
the communication events recorded in $\trace$ become observable.
This follows the realization that the reachable worlds of a CHP consist of the final state and the communication trace.

\begin{remark} \label{rem:unsoundCommit}
	In \acCommit, assumptions are only available about the communication \emph{strictly} before to prevent unsound circular reasoning \cite{AcHoare_Zwiers, Hooman1992}.
	With a non-strict definition, the formula $y = 0 \rightarrow [ \send{}{}{y} ] \ac \true$, 
	where $\A \equiv \C \equiv \len{\historyVar \downarrow \ch{}} > 0 \rightarrow \val{\historyVar \downarrow \ch{}} = 1$, would get valid.
	Locally, we are aware of the contradiction between the precondition $y = 0$ and what is assumed by $\A$, whereas
	the environment is not and would trust in the promise $\C$.
\end{remark}

\newcommand{\FmldL}{\FmlSet{\dL}{}}

\begin{proposition}[Conservative extension] \label{prop:conservative}
	The logic $\dLCHP$ is a conservative extension of $\dL$. That is, a formula $\varphi \in \Fml \cap \FmldL$ is valid in $\dLCHP$ iff it is valid in $\dL$, where $\FmldL$ is the set of $\dL$ formulas (see \rref{app:semantics}).
\end{proposition}

\subsection{Calculus}

This section develops a sound (see \rref{thm:soundness}) proof calculus for \dLCHP, summarized in \rref{fig:calculus} on page~\pageref{fig:calculus}.
In \rref{fig:derived} on page \pageref{fig:derived}, we provide common derived rules.
Since $\dLCHP$ is a conservative extension of $\dL$ (\rref{prop:conservative}), the entire $\dL$ sequent calculus 
\ifblind
\cite{DBLP:journals/jar/Platzer17,DBLP:journals/jacm/PlatzerT20} 
\else
\cite{DBLP:journals/jar/Platzer17,Platzer18,DBLP:journals/jacm/PlatzerT20}
\fi
can be used soundly for reasoning about $\dLCHP$ formulas. 
A \emph{sequent} $\Gamma \vdash \Delta$ with finite lists of formulas $\Gamma$, $\Delta$ is short for $\bigwedge_{\varphi \in \Gamma} \varphi \rightarrow \bigvee_{\psi \in \Delta} \psi$.

Each program statement is axiomatized by a dynamic box $[\cdot] \psi$ or an ac-box $[\cdot] \ac \psi$.
Axioms \RuleName{acNoCom} and \RuleName{boxesDual} for switching between dynamic and ac-boxes mediate between them.
The dynamic axioms are as usual in differential \cite{DBLP:journals/jar/Platzer17} dynamic logic \cite{Harel1979}.
The ac-axioms re-express Hoare-style ac-reasoning \cite{AcHoare_Zwiers} as a dynamic logic.
However, we design more atomic axioms for parallel composition and communication from which the proof rule \RuleName{acParCompRight} for parallel composition and proof rules for communication derive.

Noninterference (\rref{def:noninterference}) identifies valid instances of the formula $[ \alpha ] \ac \psi \rightarrow [ \alpha \parOp \beta ] \ac \psi$,
which we introduce as axiom \RuleName{acDropComp} in \rref{fig:calculus}.
For formula~$\chi$, the accessed channels $\CN(\chi) \subseteq \Chan$ are those channels whose communication may influence the truth value of $\chi$, \eg $\ch{}$ in $\len{\historyVar \downarrow \ch{}} > 0$.
That is, $\CN(\chi)$ plays a similar role for the communication traces $\restrict{\pstate{v}}{\TVar}$ (the state restricted to $\TVar$) with $\pstate{v} \in \states$ as $\FV(\chi)$ does for the overall state $\pstate{v}$.
For program $\alpha$, the set $\CN(\alpha)$ denotes the communication channels used.

\begin{definition}[Noninterference] \label{def:noninterference}
	Given an ac-box $[ \alpha \parOp \beta ] \ac \psi$ the CHP $\beta$ \emph{does not interfere} with its surrounding contract if the following conditions hold:%
	\footnote{The definition only restricts $\beta$'s influence on formulas $\A$, $\C$, $\psi$ but not on program~$\alpha$ because in parallel composition $\alpha \parOp \beta$, the subprograms must not share state anyway.}
	\begin{align}
		&\FV(\psi) \cap \BV(\beta) \subseteq \{ \globalTime \} \cup \TVar \label{eq:noninterference_0}\\
		&\FV(\chi) \cap \BV(\beta) \subseteq \TVar \sidecondition{for $\chi \in \{\A, \C\}$} \label{eq:noninterference_1}\\
		&\CN(\chi) \cap \CN(\beta) \subseteq \CN(\alpha) \sidecondition{for $\chi \in \{ \A, \C, \psi \}$} \label{eq:noninterference_2}
	\end{align}%
\end{definition}

Clearly, state variables bound by $\beta$ and free in $\chi \in \{ \A, \C, \psi \}$ would influence~$\chi$'s truth in $[ \alpha \parOp \beta ] \ac \psi$.
But \rref{eq:noninterference_0} and \rref{eq:noninterference_1} do not capture trace variables $\TVar$ since they are also $\alpha$'s interface with communication that might be joint with~$\beta$.
However, \rref{eq:noninterference_2} restricts access to trace variables in~$\chi$ to those 
channels whose communication can be observed either exclusively from $\alpha$ or as joint communication between $\alpha$ and $\beta$,
thus prevents influence of~$\beta$ on~$\chi$ beyond what is already caused by $\alpha$.
Still \rref{def:noninterference} allows full access to~$\alpha$'s communication including the joint communication with $\beta$.

\subsubsection{Dynamic ac-reasoning}

In Hoare-style ac-reasoning~\cite{AcHoare_Zwiers}, 
a distinguished history variable records communication globally.
Assuming that $\historyVar$ is such a variable in $\dLCHP$,
a tempting but \emph{wrong} axiomatization of the send statement would be
\begin{equation*} \tag{$\lightning$} \label{eq:wrongComAxiom}
	[ \send{}{non}{} ] \psi(\historyVar) \leftrightarrow \fa{\historyVar_0} \big( \historyVar_0 = \historyVar \cdot \comItem{\ch{}, \rp, \globalTime} \rightarrow \psi(\historyVar_0) \big) \text{.}
\end{equation*}
Applying it to 
\begin{equation}
	\label{eq:toexecute}
	\vdash [ \send{\ch{}_1}{non}{\rp_1} ][ \send{\ch{}_2}{non}{\rp_2} ] \psi(\historyVar)
\end{equation}
results in $\historyVar_0 = \historyVar \cdot \comItem{\ch{}_1, \rp_1, \globalTime} \vdash [ \send{\ch{}_2}{non}{\rp_2} ] \psi(\historyVar_0)$.
After this step, 
the ongoing history is $\historyVar_0$.
However, another application leads to
$\historyVar_0 = \historyVar \cdot \comItem{\ch{}_1, \rp_1, \globalTime}, \historyVar_1 = \historyVar \cdot \comItem{\ch{}_2, \rp_2, \globalTime} \vdash \psi(\historyVar_0)$.
Incorrectly, communication is appended to $\historyVar$ again and $\psi(\historyVar_0)$ still refers to $\historyVar_0$.
Problematically, the substitution $( [ \send{\ch{}_2}{non}{\rp_2} ] \psi(\historyVar) ) \synsubs{\historyVar}{\historyVar_0}$ during the first application does not guide $\ch{}_2 ! \rp_2$ to append its communication to $\historyVar_0$.
Without $\historyVar$ occurring syntactically but being free in $\send{\ch{}_2}{non}{\rp_2}$ the substitution does not even have a meaningful definition.
For a similar reason, axiom  
\begin{equation*}
	[ \send{}{non}{} ] \psi(\historyVar) \leftrightarrow \psi(\historyVar \cdot \comItem{\ch{}, \rp, \globalTime}) \tag{$\lightning$}
\end{equation*}
is unsound as applying it twice to \rref{eq:toexecute} leads to $\vdash \psi(\historyVar \cdot \comItem{\ch{}_2, \rp_2, \globalTime} \cdot \comItem{\ch{}_1, \rp_1, \globalTime})$ with the communication items in wrong order.
Here the axiom is not able to append the second item at the right position because there is no symbolic name for the state $\historyVar \cdot \comItem{\ch{}_1, \rp_1, \globalTime}$ of history after the first application.

To enable symbolic execution,
we drop the assumption of a distinguished history variable and annotate each communication statement $\send{}{}{}$ and $\receive{}{}{}$ with an explicit recorder variable $\historyVar$.
Now, substitution $\alpha \synsubs{\historyVar}{\historyVar_0}$ is defined easily as $\send{}{\historyVar_0}{}$ for $\alpha \equiv \send{}{}{}$, and as $\receive{}{\historyVar_0}{}$ for $\alpha \equiv \receive{}{}{}$, and as $\alpha$ for other atomic programs, and by recursive application otherwise.

\begin{figure}[h!tb]
	\begin{small}
		\begin{minipage}{\textwidth}
			\begin{calculus}
				\startAxiom{assign}
					$[ x \ceq \rp] \psi(x) \leftrightarrow \psi(\rp)$
				\stopAxiom
				\startAxiom{nondetAssign}
					$[ x \ceq *] \psi \leftrightarrow \fa{x} \psi$
				\stopAxiom
				\startAxiom{test}
					$[ \test{} ] \psi \leftrightarrow (\chi \rightarrow \psi)$
				\stopAxiom
				\startAxiom{boxesDual}
					$[ \alpha ] \psi \leftrightarrow [ \alpha ] \acpair{\true, \true} \psi$
				\stopAxiom
			\end{calculus}%
			\hspace*{1em}%
			\begin{calculus}
				\startAxiom{acComposition}
					$[\alpha \seq \beta] \ac \psi \leftrightarrow [\alpha] \ac [\beta] \ac \psi$
				\stopAxiom
				\startAxiom{acChoice}
					$[\alpha \cup \beta] \ac \psi \leftrightarrow [\alpha] \ac \psi \wedge [\beta] \ac \psi$
				\stopAxiom
				\startAxiom{acIteration}
					$[ \repetition{\alpha }] \ac \psi \leftrightarrow [\alpha^0] \ac \psi \wedge [\alpha] \ac [ \repetition{\alpha} ] \ac \psi$%
					\footnote{\label{ft:ind-base}Note that $[\alpha^0] \ac \psi \leftrightarrow \C \wedge (\A \rightarrow \psi)$ by \RuleName{acNoCom} and \RuleName{test} since $\alpha^0 \equiv \test{\true}$.}
				\stopAxiom
				\startAxiom{acInduction}
					$[ \repetition{\alpha} ] \ac \psi \leftrightarrow [\alpha^0] \ac \psi \wedge [ \repetition{\alpha} ] \acpair{\A, \true} (\psi \rightarrow [\alpha] \ac \psi)$%
					\footnoteref{ft:ind-base}%
				\stopAxiom
			\end{calculus}
				
			\begin{calculus}
				\startAxiom{gtime}
					$[\evolution{x' = \rp}{\chi}] \psi \leftrightarrow [\evolution{\globalTime' = 1, x' = \rp}{\chi}] \psi$%
				\stopAxiom
				\startAxiom{solution}
					$[x' = \rp(x)] \psi(x) \leftrightarrow \fa{t {\ge} 0} [ x \ceq y(t) ] \psi(x)%
					\sidecondition{$y'(t) = \rp(y)$ and $\globalTime \in x$}$%
					\footnote{Conservative extension only applies if $\globalTime' = 1$ is part of $x' = \rp$, which holds if $\globalTime \in x$ since every right-hand side for $\globalTime$ other than $1$ in an evolution is considered ill-formed.}
				\stopAxiom
				\startAxiom{send}
					$[ \send{}{}{} ] \psi(\historyVar) 
					\leftrightarrow \fa{\historyVar_0} 
						\big(
							\historyVar_0 = \historyVar \cdot \comItem{\ch{}, \rp, \globalTime} \rightarrow \psi(\historyVar_0)
						\big)$
					\sidecondition{$\historyVar_0$ fresh}
				\stopAxiom
				\startAxiom{acCom}
					$[ \send{}{}{} ] \ac \psi \leftrightarrow \C \wedge \Big( \A \rightarrow [ \send{}{}{} ] \big( \C \wedge (\A \rightarrow \psi ) \big)\Big)$
				\stopAxiom 
				\startAxiom{comDual}
					$[ \receive{}{}{} ] \ac \psi \leftrightarrow [ x \ceq * ] [ \send{}{}{x} ] \ac \psi$
				\stopAxiom
				\startAxiom{K}
					$(\universal \K) \wedge [ \alpha ] \acpair{\A_1 \wedge \A_2, \C_1 \wedge \C_2} \psi \rightarrow [ \alpha ] \acpair{\A, \C_1 \wedge \C_2} \psi$%
					\footnote{$\K$ is the compositionality condition $\Kexpanded$.
					The universal closure $\universal \varphi$ of $\varphi$ is defined by $\fa{\arbitraryVar_1, ..., \arbitraryVar_n} \varphi$, where $\FV(\varphi) = \{ \arbitraryVar_1, ..., \arbitraryVar_n \}$.}%
					\footnoteref{ft:well-formed}%
				\stopAxiom	
				\startAxiom{acDropComp}
					$[ \alpha ] \ac \psi \rightarrow [ \alpha \parOp \beta ] \ac \psi$%
					\sidecondition{$\beta$ does not interfere with $[ \alpha ] \ac \psi$ (\rref{def:noninterference})}
				\stopAxiom
			\end{calculus}
			
			\begin{calculus}
				\startAxiom{acNoCom}
					$[ \alpha ] \ac \psi \leftrightarrow \C \wedge (\A \rightarrow [ \alpha ] \psi)$%
					\sidecondition{$\CN(\alpha) = \emptyset$}%
					\footnote{\label{ft:well-formed}Care must be taken, for example, when \RuleName{acNoCom} is applied from right to left, that resulting ac-boxes are well-formed, \iest $(\FV(\A) \cup \FV(\C)) \cap \BV(\alpha) \subseteq \TVar$ for $[ \alpha ] \ac \psi$.}
				\stopAxiom
				\startAxiom{acWeak}
					$[ \alpha ] \ac \psi \leftrightarrow \C \wedge [ \alpha ] \ac ( \C \wedge ( \A \rightarrow \psi ) )$
				\stopAxiom

				\startAxiom{acBoxesDist}
					$[\alpha] \acpair{\A, \C_1 \wedge \C_2} (\psi_1 \!\wedge\! \psi_2) \leftrightarrow \bigwedge_{j=1}^2 [\alpha] \acpair{\A, \C_j} \psi_j$
				\stopAxiom
			\end{calculus}
			\begin{calculus}
				\startRule{acMono}
					\Axiom{$\A_2 \rightarrow \A_1$}
					\Axiom{$\C_1 \rightarrow \C_2$}
					\Axiom{$\psi_1 \rightarrow \psi_2$}
					\TrinaryInf{$[ \alpha ] \acj{1} \psi_1 \rightarrow [ \alpha ] \acj{2} \psi_2$}
				\stopRule
				\startRule{acG}
					\Axiom{$\C \wedge \psi$}
					\UnaryInf{$[ \alpha ] \ac \psi$}
				\stopRule
			\end{calculus}
		\end{minipage}
	\end{small}
	\vspace{-.2em}
	\caption{\dLCHP proof calculus}
	\label{fig:calculus}
	\vspace{-1em}
\end{figure}

\begin{figure}[h!tb]
	\vspace{.5em}
	\begin{small}
		\begin{minipage}{\textwidth}
			\begin{calculus}
				\startRule{CG}
					\Axiom{$\Gamma, \historyVar_0 = \historyVar \cdot \comItem{\ch{}, \rp, \globalTime} \vdash [ \alpha(\historyVar_0) ] \psi, \Delta$}

					\RightLabel{\sidecondition{$\historyVar \not \in \FV(\psi)$ and $\historyVar_0$ fresh}}
					\UnaryInf{$\Gamma, \vdash [ \alpha(\historyVar) ] \psi, \Delta$}
				\stopRule
			\end{calculus}
			
			\begin{calculus}
				\startRule{acLoop}
					\Axiom{$\Gamma \vdash \C \wedge \inv, \Delta$}
					\Axiom{$\C, \inv \vdash [ \alpha ] \ac \inv$}
					\Axiom{$\A, \C, \inv \vdash \psi$}
					\TrinaryInf{$\Gamma \vdash [ \repetition{\alpha} ] \ac \psi, \Delta$}
				\stopRule
			\end{calculus}
			\begin{calculus}
				\startRule{if}
					\Axiom{$\Gamma, \varphi \vdash [ \alpha ] \psi, \Delta$}

					\Axiom{$\Gamma, \neg \varphi \vdash \psi, \Delta$}

					\BinaryInf{$\Gamma \vdash [ \ifstat{\varphi}{\{\alpha\}} ] \psi, \Delta$}
				\stopRule
			\end{calculus}

			\begin{calculus}
				\startRule{acParCompRight}
					\Axiom{$\vdash \K$}

					\Axiom{$\Gamma \vdash [ \alpha_j ] \acj{j} \psi_j, \Delta$ \sidecondition{$j = 1, 2$}}

					\RightLabel{%
						\parbox{.45\textwidth}{\;\;\textnormal{\textcolor{gray}{(%
							$\alpha_{3-j}$ does not interfere with \\ 
							\hspace*{2em}$[ \alpha_j ] \acj{j} \psi_j$ for $j = 1, 2$ (\rref{def:noninterference})%
						)}%
						\footnote{$\K$ is the compositionality condition $\Kexpanded$.}%
						}}%
					}
					\BinaryInf{$\Gamma \vdash [ \alpha_1 \parOp \alpha_2  ] \acpair{\A, \C_1 \wedge \C_2} (\C_1 \wedge \C_2), \Delta$}
				\stopRule

				\startRule{acSendRight}
					\Axiom{$\Gamma \vdash \C(\historyVar), \Delta$}

					\Axiom{$\Gamma, H_0 \vdash \C(\historyVar_0), \Delta$}

					\Axiom{$\Gamma, H_0, \A(\historyVar_0) \vdash \psi(\historyVar_0), \Delta$}

					\RightLabel{%
						\sidecondition{$\historyVar_0$ fresh}%
						\footnote{Formula $H_0$ is short for $\historyVar_0 = \historyVar \cdot \comItem{\ch{}, \rp, \globalTime}$ recording the global time $\globalTime$ as timestamp.}%
					}
					\TrinaryInf{$\Gamma \vdash [ \send{}{}{} ] \acpair{\A(\historyVar), \C(\historyVar)} \psi(\historyVar), \Delta$}
				\stopRule
			\end{calculus}
		\end{minipage}
	\end{small}
	\caption{
		Derived \dLCHP proof rules
	}
	\label{fig:derived}
	\vspace{-1em}
\end{figure}

\subsubsection{Atomic Hybrid Programs}

For an (atomic) non-communicating program $\alpha$, 
we can flatten $[ \alpha ] \ac \psi$ by axiom \RuleName{acNoCom}to a dynamic formula because $\A$ and $\C$ only refer to the initial state when $\CN(\alpha) = \emptyset$.
Subsequently, we can execute the program by its dynamic axiom (\RuleName{assign}, \RuleName{nondetAssign}, \RuleName{test}, \RuleName{solution}).
Note that by conservative extension (\rref{prop:conservative}) axiom \RuleName{solution} only applies to $[ \evolution{}{} ] \psi$ if the ODE $x' = \rp$ matches the underlying semantics,
\iest if $\globalTime \in x$, 
which has right-hand side $1$ for well-formed continuous evolution $x' = \rp$.
Therefore, axiom \RuleName{gtime} allows to materialize the flow of global time $\globalTime$ as evolution $\globalTime' = 1$ whenever necessary.

\subsubsection{Compound Hybrid Programs} 

Compound CHPs in an ac-box cannot be handled by axiom \RuleName{acNoCom} if they communicate.
Instead, they have the axioms \RuleName{acComposition}, \RuleName{acChoice}, \RuleName{acIteration}, and \RuleName{acInduction}.
Compound programs in a dynamic box can be repackaged into an ac-box using axiom \RuleName{boxesDual}.
The ac-induction axiom \RuleName{acInduction} carefully generalizes that of dynamic logic \cite{Harel1979}.
Importantly, the reachable worlds,
where the induction step needs to hold,
only depend on program $\alpha$ and assumption $\A$ about incoming communication,
whereas commitment $\C$ is proven inductively.
In \rref{fig:derived}, we give the useful derived loop invariant proof rule \RuleName{acLoop}.
As usual, it derives from axiom \RuleName{acInduction} and the ac-version \RuleName{acG} of the Gödel-generalization rule,
which confirms that the embedding into dynamic logic is correct.

If we were to neglect the communication of aborting runs in the semantics of $\alpha \seq \beta$, axiom \RuleName{acComposition} would not be sound.
Proving $[ \alpha ] \ac [ \beta ] \ac \psi$ requires to show the commitment $\C$ after each communication of $\alpha$ even if $\beta$ aborts. 
To obtain this from $[ \alpha \seq \beta ] \ac \psi$ the semantics $\sem{\alpha \seq \beta} = \botop{\sem{\alpha}} \cup (\sem{\alpha} \continuation \sem{\beta})$ contains all respective runs of $\alpha$ up to $\beta$ with $\botop{\sem{\alpha}}$.

\subsubsection{Communication Statements}

Axiom \RuleName{acCom} unfolds \acCommit for an ac-box of a single send statement into a dynamic box.
The effect on the recorder variable~$\historyVar$ of executing sending $\send{}{}{}$ is captured by axiom \RuleName{send}.
It records the event $\comItem{\ch{}, \rp, \globalTime}$ using the current global time $\globalTime$ as timestamp and renames the history in the postcondition for subsequent proof steps.
Axiom \RuleName{comDual} allows to execute a receive statement by its duality with send.
Derived rule \RuleName{acSendRight} combines \RuleName{acNoCom} and \RuleName{send}, and decomposes the statement into two premises for \acCommit and one for \acPost.
Derived rule \RuleName{CG} is useful as it eliminates the need for case distinction about empty history by prefixing the history with additional ghost communication.

\subsubsection{Parallel Composition}

A non-interfering program $\beta$ (\rref{def:noninterference}) can be dropped from parallel composition $[ \alpha \parOp \beta ] \ac \psi$ by axiom \RuleName{acDropComp} because it has no influence on the surrounding contract $[ \alpha \parOp \!\raisebox{.1em}{\underline{\hspace{.5em}}}\, ] \ac \psi$.
Since $\parOp$ is associative and commutative (see \rref{app:calculus}),
axiom \RuleName{acDropComp} can drop any subprogram in a chain of parallel statements. 
In parallel composition of $[ \alpha_j ] \acj{j} \psi_j$ for $j = 1, 2$, 
the commitments mutually contribute to the assumptions. 
This can weaken the assumption of $\alpha_1 \parOp \alpha_2$ about its environment to $\A$ by axiom \RuleName{K} if the compositionality condition $\K \equiv \Kexpanded$ is valid.

The derived rule \RuleName{acParCompRight} combines axioms \RuleName{acDropComp} and \RuleName{K} for full decomposition of parallelism.
Reasoning about a parallel $[ \alpha_1 \parOp \alpha_2 ] \ac \psi$
with arbitrary $\A$, $\C$, and $\psi$ is the task of constructing $\A_j$, $\C_j$, and $\psi_j$ for $j = 1, 2$ such that $\C_1 \wedge \C_2 \rightarrow \C$, and $\psi_1 \wedge \psi_2 \rightarrow \psi$ are valid, and such that \RuleName{acParCompRight} is applicable.
Since the side condition of \RuleName{acParCompRight} about noninterference still allows~$\C_j$ and~$\psi_j$ for $j = 1, 2$ to access $\alpha_j$'s communication including the joint communication with~$\alpha_{3-j}$, the formulas can cover the complete communication of $\alpha_1 \parOp \alpha_2$.

\subsubsection{Miscellaneous} 
Ac-boxes distribute over conjunctions 
by axiom \RuleName{acBoxesDist} except for assumptions 
just like preconditions $\varphi_j$ do not distribute in $\varphi_1 \wedge \varphi_2 \rightarrow [ \alpha ] \psi$. 
Rule \RuleName{acMono} generalizes monotonicity from dynamic to ac-boxes.
Ac-weakening \RuleName{acWeak} exploits totality of the program semantics $\sem{\alpha}$ to add or drop $\C$ in the initial state.
Moreover, adding or dropping $\C$ and $\A \rightarrow \psi$ by \RuleName{acWeak} in the final state is due to \acCommit and \acPost, respectively.
Rule \RuleName{acG} is the ac-version of the G\"odel-generalization rule.

First-order formulas $\FolPA$ over $\Itrm$ without length computations $\len{\te}$, can be handled by an effective oracle proof rule (called \RuleName{PA}) since Presburger arithmetic is decidable \cite{Presburger1931}. 
Likewise, first-order real arithmetic $\FolRA$ is decidable \cite{Tarski1951}, and we use an oracle rule for it (called \RuleName{real}) as in \dL \cite{DBLP:journals/jar/Platzer08}.
However, the full first-order fragment $\Fol$ of $\dLCHP$ is not decidable because of alternating quantifiers of trace and integer variables \cite{Bradley2006}.

Instead, reasoning about trace terms is by simple algebraic laws for successive simplification (see \rref{app:calculus}) \cite{Zwiers_Phd}.
For applicability of rules \RuleName{real} and \RuleName{PA} trace subterms can be rewritten with fresh variables.
For example, $\val{\te_1[\ie]} < \val{\te_2[\ie]} \rightarrow \val{\te_1[\ie]} + \rp < \val{\te_2[\ie]} + \rp$ is valid since $x < y \rightarrow x + \rp < y + \rp$ is valid in $\FolRA$.
Ultimately, we use \RuleName{PA} and \RuleName{real} modulo trace terms, \iest perform rewritings silently.
\\

Our central contribution is the soundness theorem about the compositional \dLCHP proof calculus.
The proof rules given in \rref{fig:derived} derive (see \rref{app:calculus}).

\begin{theorem} \label{thm:soundness}
	The $\dLCHP$ calculus (see \rref{fig:calculus}) is sound (see \rref{app:calculus}).
\end{theorem}

\section{Demonstration of \dLCHP} \label{sec:example}

We demonstrate our calculus outlining a proof of the safety condition from \rref{ex:convoy_safety} on page \pageref{ex:convoy_safety} about the convoy in \rref{fig:followerLeader} on page \pageref{fig:followerLeader}.
After decomposing the parallel statement, the proof proceeds purely mechanical by statement-by-statement symbolic execution.
It starts in \rref{fig:pardecomp} using a standard pattern for decomposing a parallel statement:
First, introduce commitments and postconditions (see \rref{fig:postconditions}) by axiom \RuleName{boxesDual} and rule \RuleName{acMono}, which relate the subprograms, such that second, rule \RuleName{acParCompRight} becomes applicable.
The latter also makes the commitments mutual assumptions.

The formulas in \rref{fig:postconditions} relating the subprograms clearly reduce their complex interior and only depend on the small communication interface and local variables,
thus are independent of any knowledge about the internal structure of the respective other car.
Adding initial ghost communication (rule \RuleName{CG}) exploits the flexibility of explicit history variables and dynamic logic to avoid cumbersome case distinction about empty history.

\begin{figure}[ht]
	\vspace*{-1em}
	\begin{align}
		\!\!\varphi \equiv\;
			& \epsilon \ge 0 \wedge
			\range{0}{v_f}{\safevelo{d}} \wedge
			v_f \le \maxvelo \wedge
			x_f + d < x_l \tag{precond.\ convoy} \\
		\!\!\psi_f \equiv\;
			& 0 \le \safevelo{d} \wedge
			x_f + (\epsilon - \timespan{\globalTime}{\historyVar \downarrow \ch{pos}}) \safevelo{d} < \val{\historyVar \downarrow \ch{pos}}\! \tag{postcond.\ $\progtt{follower}$} \\
		\!\!\psi_l \equiv\;
			& \val{\historyVar \downarrow \ch{pos}} \le x_l \wedge  
			\timespan{\globalTime}{\historyVar \downarrow \ch{pos}} \le \epsilon \tag{postcond.\ $\progtt{leader}$} \\
		\!\!\fAssume \equiv \lCommit \equiv\;
			& \range{0}{\val{\historyVar \downarrow \ch{vel}}}{\maxvelo} \tag{ac-formulas}
	\end{align}
	\vspace{-1.5em}
	\caption{
		For $\ch{} \in \{ \ch{vel}, \ch{pos} \}$,
		the term $\timespan{\globalTime}{\historyVar \downarrow \ch{}}$ is short for $\globalTime - \stamp{\historyVar \downarrow \ch{}}$, 
		\iest the time elapsed since last communication along $\ch{}$ recorded by $\historyVar$.
		The commitment $\lCommit$ given by the $\progtt{leader}$,
		assumption $\fAssume$ made by the $\progtt{follower}$,
		and the postconditions $\psi_f$ and $\psi_l$ are used in the proof in \rref{fig:pardecomp}.
	}
	\label{fig:postconditions}
\end{figure}

\begin{figure}[ht]
	\begin{prooftree}
		\Axiom{$\triangleleft_1\, \text{\RuleName{trueR}}\hspace{-1em}$}
		
		\Axiom{$\triangleleft_2\, \text{\RuleName{real}}\hspace{-1em}$}

		\Axiom{$*$}
	
		\UnaryInf{$\vdash \lCommit \rightarrow \fAssume$}
	
		\Axiom{$\triangleright$ \rref{fig:follower_in_dist_safe}}
	
		\UnaryInf{$\Gamma \vdash [ \progtt{follower}(\historyVar) ] \acpair{\fAssume, \true} \psi_f$}
	
		\Axiom{$\triangleright$ \rref{fig:leader_vel_send}}
	
		\UnaryInf{$\Gamma \vdash [ \progtt{leader}(\historyVar) ] \acpair{\true, \lCommit} \psi_l$}
	
		\RuleNameRight{acParCompRight}
		\TrinaryInf{$\Gamma \vdash [ \progtt{follower}(\historyVar) \parOp \progtt{leader}(\historyVar) ] \acpair{\true, \lCommit} (\psi_f \wedge \psi_l)$}

		\RuleNameRight{acMono}
		\BinaryInf{$\Gamma \vdash [ \progtt{follower}(\historyVar) \parOp \progtt{leader}(\historyVar) ] \acpair{\true, \C} x_f < x_l$}

		\RuleNameRight{acMono}
		\BinaryInf{$\Gamma \vdash [ \progtt{follower}(\historyVar) \parOp \progtt{leader}(\historyVar) ] \acpair{\true, \true} x_f < x_l$}
	
		\RuleNameRight{boxesDual}
		\UnaryInf{$\Gamma \vdash [ \progtt{follower}(\historyVar) \parOp \progtt{leader}(\historyVar) ] x_f < x_l$}

		\RightLabel{$2 \times$\RuleName{CG}}
		\UnaryInf{$\varphi \vdash [ \progtt{follower}(\historyVar_0) \parOp \progtt{leader}(\historyVar_0) ] x_f < x_l$}

		\RuleNameRight{implR}
		\UnaryInf{$\vdash \varphi \rightarrow [ \progtt{follower}(\historyVar_0) \parOp \progtt{leader}(\historyVar_0) ] x_f < x_l$}
	\end{prooftree}
	\vspace{-1em}
	\caption{
		Safety proof for the convoy from \rref{ex:followerLeader},
		where $\Gamma$ is the formula list $\varphi, \historyVar_1 = \historyVar_0 \cdot \comItem{\ch{vel}, 0, \globalTime}, \historyVar = \historyVar_1 \cdot \comItem{\ch{pos}, x_l, \globalTime}$.		
		The open premises $\C \vdash \true$ ($\triangleleft_1$) and $\psi_f \wedge \psi_l \vdash x_f < x_l$ ($\triangleleft_2$) close by rules \RuleName{trueR} and \RuleName{real}, respectively.
	}
	\label{fig:pardecomp}
\end{figure}

\rref{fig:pardecomp} decomposes $x_f < x_l$ into $\psi_f$ and $\psi_l$  
since $\progtt{follower}$ stays behind $\progtt{leader}$'s last known position $\val{\historyVar_\ch{pos}}$, 
whereas $\progtt{leader}$ never drives backward by $\psi_l$.
Indeed, $\progtt{follower}$ stores the last known distance $\val{\historyVar_\ch{pos}} - x_f$ in $d$
and $\epsilon - \timespan{\globalTime}{\historyVar_\ch{pos}}$ bounds the waiting time till the next position update along channel $\ch{pos}$.
Thus, $\progtt{follower}$ stays behind $\progtt{leader}$ when driving with speed~$\nicefrac{d}{\epsilon}$.

\rref{fig:follower_in_dist_safe} continues from \rref{fig:pardecomp} 
and demonstrates the reasoning along one execution path of $\progtt{follower}$.
The invariant $\invariant{f}$ for induction \RuleName{acLoop} bounds $\progtt{follower}$'s speed by $\safevelo{d}$ such that it stays behind $\progtt{leader}$ before the next position update.
\RuleName{TA} indicates trace algebra reasoning.
The remaining proof is mechanical symbolic execution.
In particular, combining \RuleName{acNoCom} and \RuleName{acWeak} swallows the ac-formulas $\{\A, \true\}$ and using duality axiom \RuleName{comDual} allows to execute the communication by rule \RuleName{acSendRight}.
Axiom \RuleName{gtime} materializes the flow of the global time $\globalTime$ making the solution axiom \RuleName{solution} applicable.
Weakening \RuleName{WL} drops irrelevant premises,
\RuleName{forallR} introduces fresh for quantified variables,
and \RuleName{subsL} and \RuleName{subsR} perform substitution on the left and right, respectively.
Trace algebra \RuleName{TA} evaluates the assumption~$\fAssume$.
Finally, the proof concludes by real arithmetic \RuleName{real} modulo trace terms.

\newcommand{\UnaryText}[1]{\UnaryInf{%
	\hspace*{-2em}\vbox{\vspace*{.2em}%
	#1%
	\vspace*{.3em}%
}}}

\begin{figure}[ht!]
	\begin{small}
		\begin{prooftree}[shape=justified, JustifiedScoresWidth=.9\textwidth]
			\Axiom{$*$}
	
			\RuleNameRight{real}
			\UnaryInf{$\begin{aligned}
				& \invariant{f}, \fAssume, \range{0}{\tarvelo}{\maxvelo}, d > \maxdist, t \ge 0 \vdash \range{0}{\tarvelo}{\safevelo{d}} \\
				&\qquad \wedge v_f \le \maxvelo \wedge x_f + (\epsilon - \timespan{\globalTime + t}{\historyVar \downarrow \ch{pos}}) \safevelo{d} < \val{\historyVar \downarrow \ch{pos}}
			\end{aligned}$}
	
			\RuleNameRight{TA}
			\UnaryInf{$\begin{aligned}
				& \invariant{f}, \fAssume, \fAssume(\historyVar \cdot \comItem{\ch{}, \tarvelo, \globalTime}), d > \maxdist, t \ge 0 \vdash \invariant{f} \parameters{\tarvelo, x_f + t \cdot \tarvelo, \globalTime + t, \historyVar_\ch{vel}}
			\end{aligned}$}

			\RuleNameRight{subsL}
			\UnaryInf{$\begin{aligned}
				\invariant{f}, \fAssume, H_\ch{vel}, \fAssume \parameters{\historyVar_\ch{vel}}, d > \maxdist, t \ge 0 \vdash \invariant{f}\parameters{\tarvelo, x_f + t \cdot \tarvelo, \highlight{\globalTime + t, \historyVar_\ch{vel}}}
			\end{aligned}$}

			\RuleNameRight{forallR}
			\UnaryInf{$\begin{aligned}
				\nonrelevant{\invariant{f}, \fAssume, H_\ch{vel}, \fAssume \parameters{\historyVar_\ch{vel}}, d > \maxdist} \vdash \fa{t{\ge}0} \invariant{f} \parameters{\tarvelo, \highlight{x_f + t \cdot \tarvelo}, \highlight{\globalTime + t, \historyVar_\ch{vel}}}
			\end{aligned}$}

			\RuleNameRight{solution, assign}
			\UnaryInf{$\nonrelevant{\invariant{f}, \fAssume, H_\ch{vel}, \fAssume \parameters{\historyVar_\ch{vel}}, d > \maxdist} \vdash [ \evolution*{\globalTime' = 1, x_f' = \tarvelo}{non} ] \invariant{f} \parameters{\tarvelo, x_f, \globalTime, \historyVar_\ch{vel}}$}
	
			\RuleNameRight{gtime}
			\UnaryInf{$\nonrelevant{\invariant{f}, \fAssume, H_\ch{vel}, \fAssume \parameters{\historyVar_\ch{vel}}, d > \maxdist} \vdash [ \Plant_f \parameters{\highlight{\tarvelo}} ] \invariant{f} \parameters{\highlight{\tarvelo}, x_f, \globalTime, \historyVar_\ch{vel}}$}

			\SideAx$\triangleright \text{\rref{fig:follower_in_dist_unsafe}}$
			\RuleNameRight{assign}
			\UnaryInf{$\invariant{f}, \fAssume, H_\ch{vel}, \fAssume \parameters{\historyVar_\ch{vel}}, d > \maxdist \vdash [ v_f \ceq \tarvelo
			] [ \Plant_f \parameters{v_f} ] \invariant{f} \parameters{v_f, x_f, \globalTime, \historyVar_\ch{vel}}$}
	
			\SideAx$\triangleright\RuleName{trueR}$
			\RuleNameRight{if}
			\UnaryInf{$\invariant{f}, \fAssume, H_\ch{vel}, \fAssume \parameters{\historyVar_\ch{vel}} \vdash [ \ifstat{\text{safe}_d}{v_f \ceq \tarvelo}
			] [ \Plant_f ] \invariant{f} \nonrelevant{\parameters{v_f, x_f, \globalTime, \historyVar_\ch{vel}}}$}
			
			\SideAx$\triangleright\RuleName{trueR}$

			\RuleNameRight{acSendRight}
			\UnaryInf{$\invariant{f} \vdash [ \send{\ch{vel}}{\historyVar}{\tarvelo} ] \acpair{\fAssume(\historyVar), \true} [ \ifstat{\text{safe}_d}{v_f \ceq \tarvelo}
			] [ \Plant_f ] \invariant{f} \nonrelevant{\parameters{v_f, x_f, \globalTime, \historyVar}}$}

			\RuleNameRight{forallR}
			\UnaryInf{$\invariant{f} \vdash \fa{\tarvelo} [ \send{\ch{vel}}{\historyVar}{\tarvelo} ] \acpair{\fAssume(\historyVar), \true} [ \ifstat{\text{safe}_d}{v_f \ceq \tarvelo}
			] [ \Plant_f ] \invariant{f} \nonrelevant{\parameters{v_f, x_f, \globalTime, \historyVar}}$}

			\RuleNameRight{nondetAssign}
			\UnaryInf{$\invariant{f} \vdash [ \tarvelo \ceq * ] [ \send{\ch{vel}}{\historyVar}{\tarvelo} ] \acpair{\fAssume(\historyVar), \true} [ \ifstat{\text{safe}_d}{v_f \ceq \tarvelo}
			] [ \Plant_f ] \invariant{f} \nonrelevant{\parameters{v_f, x_f, \globalTime, \historyVar}}$}

			\RuleNameRight{comDual}
			\UnaryInf{$\invariant{f} \vdash [ \receive{\ch{vel}}{\historyVar}{\tarvelo} ] \acpair{\fAssume(\historyVar), \true} [ \ifstat{\text{safe}_d}{v_f \ceq \tarvelo}
			] [ \Plant_f ] \invariant{f} \nonrelevant{\parameters{v_f, x_f, \globalTime, \historyVar}}$}

			\RuleNameRight{acNoCom, acWeak}
			\UnaryInf{$\invariant{f} \vdash [ \receive{\ch{vel}}{\historyVar}{\tarvelo} ] \acpair{\fAssume(\historyVar), \true} [ \ifstat{\text{safe}_d}{v_f \ceq \tarvelo}
			] \acpair{\fAssume(\historyVar), \true} [ \Plant_f ] \invariant{f} \nonrelevant{\parameters{v_f, x_f, \globalTime, \historyVar}}$}
	
			\SideAx$\triangleright \text{\rref{fig:follower_mes_dist_safe}}$
			\RuleNameRight{acComposition}
			\UnaryInf{$\invariant{f} \vdash [ \ctrlVelocity(\historyVar) ] \acpair{\fAssume(\historyVar), \true} [ \Plant_f ] \invariant{f} \nonrelevant{\parameters{v_f, x_f, \globalTime, \historyVar}}$}

			\UnaryText{Execution by \RuleName{acComposition}, \RuleName{acChoice}, \RuleName{acNoCom}, \RuleName{acWeak}, and \RuleName{andR}}

			\SideAx$\triangleright_1 \text{\RuleName{TA, real}} \;\;\triangleright_2 \text{\RuleName{real}}$
			\UnaryInf{$\invariant{f} \vdash [ (\ctrlVelocity(\historyVar) \cup \ctrlDistance(\historyVar)) \seq \Plant_f ] \acpair{\fAssume(\historyVar), \true} \invariant{f} \nonrelevant{\parameters{v_f, x_f, \globalTime, \historyVar}}$}
	
			\RuleNameRight{acLoop}
			\UnaryInf{$\Gamma \vdash [ \progtt{follower}(\historyVar) ] \acpair{\fAssume(\historyVar), \true} \psi_f$}
		\end{prooftree}
	\end{small}
	\vspace*{-1em}
	\caption{
		Partial proof for $\progtt{follower}$.
		The induction \RuleName{acLoop} uses $\invariant{f} \equiv \range{0}{v_f}{\safevelo{d}} \wedge 
		v_f \le \maxvelo \wedge x_f + (\epsilon - \timespan{\globalTime}{\historyVar \downarrow \ch{pos}}) \safevelo{d} < \val{\historyVar \downarrow \ch{pos}}$ as invariant.
		Formula $H_\ch{vel}$ is short for $\historyVar_\ch{vel} = \historyVar \cdot \comItem{\ch{vel}, \tarvelo, \globalTime}$ and $\text{safe}_d$ is short for $d > \maxdist$.
		The induction base $\Gamma \vdash \invariant{f}$ ($\triangleright_1$) closes by trace algebra \RuleName{TA} and real arithmetic \RuleName{real}.
		Postcondition $\invariant{f} \vdash \psi_f$ ($\triangleright_2$) holds by \RuleName{real}.
		For clarity, we \highlight{highlight} substitutions.
	}
	\label{fig:follower_in_dist_safe}
	\vspace*{-1em}
\end{figure}

\section{Related Work}

Unlike CHPs, Hybrid CSP (HCSP) \cite{Jifeng1994} extends CSP \cite{Hoare1978} with \emph{eager} continuous evolution 
terminating on violation of the evolution constraint.
This reduced nondeterminism leaves negligible room for parallel programs to agree on a duration,
which easily results in empty behavior and vacuous proofs,
Non-eager evolution in CHPs subsumes eager runs.
Instead of exploiting their compositional models as in \dLCHP,
other hybrid process algebras are verified non-compositionally by translation to model checking \cite{Man2005, Cong2013, Song2005}.
Unlike CHPs, which demonstrated to model and reason about loss of communication out of the box,
meta-level components \cite{Lunel2019,DBLP:journals/sttt/MullerMRSP18,Kamburjan2020,Benvenuti2014,Frehse2008,Henzinger2001,Lynch2003} would need to be rethought to integrate lossy communication as for every other new application as well.

Hybrid Hoare-logic (HHL) for HCSP \cite{Liu2010} is non-compositional \cite{Wang2012}.
Wang\etal \cite{Wang2012} extend it with assume-guarantee reasoning (AGR) 
in a way that, unlike \dLCHP, becomes non-compositional again.
Unfortunately, their rule for parallel composition still explicitly unrolls all interleavings in the postcondition for communication traces reflecting the structure of the subprograms.
Assumptions and guarantees in HHL cannot specify the communication history but consider readiness for reasoning about deadlock freedom for future work \cite{Wang2012}.
Externalizing the complete observable behavior (and program structure) in this way devalues the whole point of 
compositionality \cite[Section 1.6.2]{deRoever2001} but only postpones reasoning about the exponentially many interleavings.
Similarly, Guelev\etal encode the semantics of the parallel composition into the postcondition \cite{Guelev2017}.

Hoare-style ac-reasoning \cite{AcHoare_Zwiers,AcSemantics_Zwiers,Hooman1987} including Hoare-style reasoning for HCSP \cite{Liu2010,Wang2012,Guelev2017} lacks symbolic execution as intuitive reasoning principle but manages with a distinguished history variable since multiple Hoare-triples cannot be considered together.
\dLCHP makes symbolic execution possible despite communication through explicit trace variables referring to the different possible states of the history in a proof.
The resulting combination of ac-reasoning and dynamic logic allows flexible switch between first-order, dynamic, and ac-reasoning while the axioms are simple capturing discrete, continues, or communication behavior.
Unlike \dLCHP, which has a global flow of time due to continuous evolution,
calculi for distributed real-time computation \cite{Hooman1987,Hooman1992} need to consider the waiting for termination of time-consuming discrete statements.

Unlike other \dL approaches \cite{Lunel2019,DBLP:journals/sttt/MullerMRSP18,Kamburjan2020},
\dLCHP has a parallel operator with built-in time-synchronization as first-class citizen in hybrid programs that can be arbitrarily nested with other programs,
rather than parallel composition of meta-level components with a explicit time model.
Modeling of parallelism by nondeterministic choice additionally requires extra care to ensure execution periodicity \cite{Lunel2019}.
In contrast to first-order constraints relating at most consecutive I/O events \cite{Lunel2019,DBLP:journals/sttt/MullerMRSP18,Kamburjan2020},
\dLCHP can reason about invariants of the whole communication history.
Different from our integrated reasoning about discrete, hybrid, and communication behavior,
Kamburjan\etal \cite{Kamburjan2020} separate reasoning about communication from hybrid systems reasoning.

Quantified differential dynamic logic \QdL \cite{DBLP:conf/csl/Platzer10} allows reasoning about parallel compositions of an unbounded number of distributed CPSs.
Unlike \dLCHP that can reason about interactions of entirely different programs, parallelism in \QdL is restricted to subprograms with a \emph{homogeneous} structure.

Different from the denotational semantics of CHPs,
parallel composition of hybrid automata \cite{Lynch2003, Frehse2004, Henzinger2001, Benvenuti2014}, just like Hoare-style reasoning about HCSP \cite{Wang2012,Guelev2017}, always fall back to the combinatorial exploration of parallelism.
Consequently, even AGR approaches \cite{Lynch2003,Frehse2004,Henzinger1996,Benvenuti2014} for hybrid automata that mitigate the state space explosion for subautomata,
eventually resort to large product automata later.
In contrast, \dLCHP's proof rule for parallel composition exploits the built-in compositionality of its semantics enabling verification of subprograms truly independent of their environment except for the communication interface.
Unlike ac-formulas in \dLCHP,
which can capture change, rate, delay, or noise for arbitrary pairings of communication channels,
overapproximation is limited to coarse abstractions by timed transition systems \cite{Frehse2004},
components completely dropping knowledge about continuous behavior \cite{Henzinger2001},
or static global contracts \cite{Benvenuti2014}.
Where \dLCHP inherits complete reasoning about differential equation invariants from \dL,
automata approaches are often limited to linear continuous dynamics \cite{Frehse2004,Henzinger2001}.

Concurrent dynamic logic (CDL) has no way for parallel programs to interact~\cite{Peleg1987}. 
CDL with communication \cite{Peleg1987a} has CSP-style \cite{Hoare1978} communication primitives but lacks continuous behavior and a proof calculus for verification.

\section{Conclusion}

This paper presented a dynamic logic \dLCHP for communicating hybrid programs (CHPs) with synchronous parallel composition in global time.
The \dLCHP proof calculus is the first truly compositional verification approach for communicating parallel hybrid systems.
To this end, \dLCHP exploits the flexibility of dynamic logic by complementing necessity and possibility modalities with assumption-commitment~(ac) modalities.
Crucially, this embedding of ac-reasoning enables compositional specification and verification of parallel hybrid behavior in a way that tames their complexity.
The practical feasibility of \dLCHP increases as it supports reasoning via intuitive symbolic execution in the presence of communication.
All technical subtleties in the semantic construction remain under the hood such that the actual calculus naturally generalizes dynamic logic reasoning.

Future work includes developing a uniform substitution calculus \cite{DBLP:journals/jar/Platzer17} for $\dLCHP$ in order to enable parsimonious theorem prover implementations \cite{DBLP:conf/cade/FultonMQVP15}.

\ifblind\else
\subsubsection{Funding Statement}
This project was funded in part by the Deutsche For\-schungs-gemeinschaft (DFG) -- \href{https://gepris.dfg.de/gepris/projekt/378803395?context=projekt&task=showDetail&id=378803395&}{378803395} (ConVeY) and an Alexander von Humboldt Professorship.
\fi

\let\oldthebibliography\thebibliography
\let\endoldthebibliography\endthebibliography
\renewenvironment{thebibliography}[1]{
  \begin{oldthebibliography}{#1}
    \setlength{\itemsep}{0em}
    \setlength{\parskip}{0em}
}
{
  \end{oldthebibliography}
}

\begingroup
\linespread{.87}
\renewcommand{\doi}[1]{doi: \href{https://doi.org/#1}{\nolinkurl{#1}}}
\bibliographystyle{splncs04}
\bibliography{platzer,literature}	
\endgroup

\appendix
\clearpage

\section{Details of the Semantics} \label{app:semantics}

We give a formal semantics for terms in \rref{def:termSemantics}, 
prove that the semantics of CHPs is prefix-closed and total in \rref{prop:prefixClosedAndTotal}, 
and prove that $\dLCHP$ is a conservative extension of $\dL$ (see \rref{prop:conservative}). 
Moreover, we prove that the operator $\closedComposition$ from \rref{sec:semantics} is associative.

We adapt the semantics of terms from Zwiers \cite[p. 113]{Zwiers_Phd} to our setup with real arithmetic as base terms.
In order to define the semantics of terms, we use the operators on traces of the following definition:

\begin{definition}[Trace operators] \label{def:traceOperators}
	Let $\trace = (\trace_1, ..., \trace_n) \in \traces$ be a trace with communication events $\trace_i$.
	Then we define its length $\semLen{\trace}$ to be $n$.
	Additionally, $\semLen{\epsilon} = 0$.
	For $\trace \in \traces$ and $k \in \naturals$,
	we define $\semAt{\trace}{k}$ to be $\trace_{k+1}$ for $0 \le k < \semLen{\trace}$.
	Otherwise, we define $\semAt{\trace}{k} = \epsilon$.

	For a communication event $\rawtrace = \comItem{\ch{}, a, \duration}$ or $\rawtrace = \comItem{\historyVar, \ch{}, a, \duration}$, we define channel access $\semChan{\rawtrace} = \ch{}$, value access
	$\semVal{\rawtrace} = a$, and access to the timestamp $\semTime{\rawtrace} = \duration$.
	Moreover, we define $\semChan{\epsilon} = \ch{}$ for some $\ch{} \in \Chan$, and $\semVal{\epsilon} = 0$, and $\semTime{\epsilon} = 0$.
\end{definition}

\begin{definition}[Value of a term] \label{def:termSemantics}
	The \emph{value} $\sem{\expr} \pstate{v} \subseteq \reals \cup \naturals \cup \Chan \cup \traces$ of a term $\expr \in \Trm$ over the state $\pstate{v}$ is inductively defined in \rref{fig:termSemantics}.
\end{definition}

\begin{figure}[ht]
	\vspace{-3em}
	\begin{minipage}{.5\linewidth}
		\begin{subfigure}[t]{\textwidth}%
			\begin{align*}
				\sem{x} \pstate{v} & = v(x) \sidecondition{$x \in \RVar \cup \{\globalTime \}$} \\
				\sem{c} \pstate{v} & = c \\
				\sem{\re_1 + \re_2} \pstate{v} & = \sem{\re_1} \pstate{v} + \sem{\re_2} \pstate{v} \\
				\sem{\re_1 \cdot \re_2} \pstate{v} & = \sem{\re_1} \pstate{v} \cdot \sem{\re_2} \pstate{v} \\	
				\sem{\val{\at{\te}{\ie}}} \pstate{v} & = \semVal{\semAt{\sem{\te} \pstate{v}}{\sem{\ie} \pstate{v}}} \\
				\sem{\stamp{\at{\te}{\ie}}} \pstate{v} & = \semTime{\semAt{\sem{\te} \pstate{v}}{\sem{\ie} \pstate{v}}}
			\end{align*}
			\vskip-.8em
			\caption[]{Real terms}
		\end{subfigure}
		\begin{subfigure}[t]{\textwidth}%
			\begin{align*}
				\sem{ch} \pstate{v} & = \ch{} \\
				\sem{\chan{\at{\te}{\ie}}} \pstate{v} & = \semChan{\semAt{\sem{\te} \pstate{v}}{\sem{\ie} \pstate{v}}}
			\end{align*}
			\vskip-.8em
			\caption[]{Channel terms}
		\end{subfigure}
	\end{minipage}%
	\begin{minipage}{.5\linewidth}
		\begin{subfigure}[t]{\textwidth}%
			\begin{align*}
				\sem{\intVar} \pstate{v} & = \pstate{v}(\intVar) \\
				\sem{\kappa} \pstate{v} & = \kappa \sidecondition{$\kappa \in \{ 0, 1 \}$} \\
				\sem{\ie_1 + \ie_2} \pstate{v} & = \sem{\ie_1} \pstate{v} + \sem{\ie_2} \pstate{v} \\
				\sem{\len{\te}} \pstate{v} & = \semLen{\sem{\te} \pstate{v}}
			\end{align*}
			\vskip-.8em
			\caption[]{Integer terms}
		\end{subfigure}
		\begin{subfigure}[t]{\textwidth}%
			\begin{align*}
				\sem{\historyVar} \pstate{v} & = \pstate{v}(\historyVar) \\
				\sem{\epsilon} \pstate{v} & = \epsilon \\
				\sem{\comItem{\ch{}, \rp_1, \rp_2}} \pstate{v} & = \comItem{\ch{}, \sem{\rp_1} \pstate{v}, \sem{\rp_2} \pstate{v}} \\
				\sem{\te_1 \cdot \te_2} \pstate{v} & = \sem{\te_1} \pstate{v} \cdot \sem{\te_2} \pstate{v} \\
				\sem{\te \downarrow \cset } \pstate{v} & = (\sem{\te} \pstate{v}) \downarrow \cset
			\end{align*}
			\vskip-.8em
			\caption[]{Trace terms}
		\end{subfigure}
	\end{minipage}
	\caption[Valuation of terms]{
		Inductive definition of the valuation $\sem{\expr} \pstate{v} \subseteq \reals \cup \naturals \cup \Chan \cup \traces$ of a term $\expr \in \Trm$ over the state $\pstate{v} \in \states$ \cite[p. 113]{Zwiers_Phd}.}
	\label{fig:termSemantics}
	\vspace*{-1em}
\end{figure}

\begin{proposition}[Prefix-closed and total] \label{prop:prefixClosedAndTotal}
	Let $\gamma \in \Chp$ be a program.
	Then its semantics $\sem{\gamma} \subseteq \pDomain$ is prefix-closed and total,
	\iest if $\computation \in \sem{\gamma}$ and $\observable[pre] \preceq \observable$,
	then $\computation[pre] \in \sem{\gamma}$, and $\pLeast \subseteq \sem{\gamma}$.
\end{proposition}
\begin{proof}
	The proof is by induction on the structure of program $\gamma$.
	We consider $\alpha^n$ to be structurally smaller than $\alpha^*$ for all programs $\alpha$ and all $n \in \naturals$.
	\Wlossg we assume $\observable[pre] \prec \observable$ in proving prefix-closedness because $\computation[pre] \in \sem{\gamma}$ trivially holds if $\computation \in \sem{\gamma}$ and $\observable[pre] = \observable$.

	\begin{enumerate}
		\item 
	 	$\gamma \in \{ x \ceq \rp, x \ceq *, \test{}, \evolution{}{} \}$, 
		then $\sem{\gamma} = \pLeast \cup U$ with $U \subseteq \pDomain$.
		Now, let $\computation \in \sem{\gamma}$ and $\observable[pre] \prec \observable$.
		Then $\trace[pre] = \epsilon$ since $\trace = \epsilon$, and $\pstate[pre]{w} = \bot$.
		Hence, $\computation[pre] \in \pLeast \subseteq \sem{\gamma}$.
		Finally, $\sem{\gamma}$ is total because $\pLeast \subseteq \sem{\gamma}$.

		\item $\gamma \in \{ \send{}{}{}, \receive{}{}{} \}$, 
		then let $\computation \in \sem{\gamma}$ and $\observable[pre] \prec \observable$.
		Since $\computation \in \sem{\gamma}$ iff $\observable \preceq \observable[alt]$ for some $\observable[alt]$,
		we obtain $\computation[pre] \in \sem{\gamma}$ because $\observable[pre] \prec \observable \preceq \observable[alt]$.
		Since $(\epsilon, \bot) \preceq \observable[alt]$ for any $\observable[alt]$,
		we have $\leastComputation \in \sem{\gamma}$ for each $\pstate{v} \in \states$ such that $\sem{\gamma}$ is total.

		\item
		$\gamma \equiv \alpha \cup \beta$, then $\sem{\alpha}$ and $\sem{\beta}$ are prefix-closed and total by IH.
		Then prefix-closedness and totality of $\sem{\gamma}$ follows easily from $\sem{\gamma} = \sem{\alpha} \cup \sem{\beta}$.
		
		\item \label{itm:compositionPrefixClosedAndTotal}
		$\gamma \equiv \alpha \seq \beta$, then let $\computation \in \sem{\gamma} = \botop{\sem{\alpha}} \cup \sem{\alpha} \continuation \sem{\beta}$ and $\observable[pre] \prec \observable$.
		Thus, $\pstate[pre]{w} = \bot$ and  $\trace[pre] \preceq \trace$.
		Observe that $\botop{\sem{\alpha}}$ is prefix-closed because $\sem{\alpha}$ is prefix-closed by IH.
		So if $\computation \in \botop{\sem{\alpha}}$, then $\computation[pre] \in \botop{\sem{\alpha}} \subseteq \sem{\gamma}$.
		If $\computation \in \sem{\alpha} \continuation \sem{\beta}$, then $\pstate{u} \neq \bot$ exists such that $(\pstate{v}, \trace_1, \pstate{u}) \in \sem{\alpha}$, and $(\pstate{u}, \trace_2, \pstate{w}) \in \sem{\beta}$, and $\trace = \trace_1 \cdot \trace_2$.
		If $\trace_1$ is fully contained in $\trace[pre]$, \iest $\trace[pre] = \trace_1 \cdot \trace[pre]_2$ for some $\trace[pre]_2$, then $\trace[pre]_2 \preceq \trace_2$ and $(\trace[pre]_2, \pstate[pre]{w}) \preceq (\trace_2, \pstate{w})$.
		Since $\sem{\beta}$ is prefix-closed by IH, we obtain $(\pstate{u}, \trace[pre]_2, \pstate[pre]{w}) \in \sem{\beta}$, 
		which implies $(\pstate{v}, \trace[pre], \pstate[pre]{w}) \in \sem{\alpha} \continuation \sem{\beta} \subseteq \sem{\gamma}$.
		If otherwise $\trace[pre] \prec \trace_1 \preceq \trace$, where $\trace_1$ is not fully contained in $\trace[pre]$, we have $(\trace[pre], \pstate[pre]{w}) \preceq (\trace_1, u)$ since $\pstate[pre]{w} = \bot$. 
		Thus, $\computation[pre] \in \sem{\alpha}$ because $\sem{\alpha}$ is prefix-closed by IH.
		Finally, $\computation[pre] \in \botop{\sem{\alpha}} \subseteq \sem{\gamma}$.

		Since $\sem{\alpha}$ is total by IH,
		$\pLeast \subseteq \sem{\alpha}$,
		which implies $\pLeast \botop{\sem{\alpha}} \subseteq \sem{\alpha \seq \beta}$.

		\item
		$\gamma \equiv \repetition{\alpha}$, then $\sem{\alpha^n}$ is prefix-closed for all $n \in \naturals$ by IH as we considered $\alpha^n$ to be structurally smaller than $\repetition{\alpha}$.
		This easily implies that $\sem{\repetition{\alpha}} = \bigcup_{n \in \naturals} \sem{\alpha^n}$ is prefix-closed.
		Finally, $\sem{\gamma}$ is total because $\sem{\alpha} \subseteq \sem{\gamma}$ and $\sem{\alpha}$ is total by IH.

		\item 
		$\gamma \equiv \alpha_1 \parOp \alpha_2$, then let $\computation \in \sem{\gamma}$ and $\observable[pre] \prec \observable$.
		Thus, $\pstate[pre]{w} = \bot$ and $\trace[pre] \preceq \trace$.
		Moreover, $\computation[proj={\alpha_j}] \in \sem{\alpha_j}$ for $j = 1,2$,
		and $\statetime{\pstate{w}_{\alpha_1}} = \statetime{\pstate{w}_{\alpha_2}}$, 
		and $\trace = \trace \downarrow \gamma$, 
		and $w = w_{\alpha_1} \merge w_{\alpha_2}$. 
		Now, observe that $\trace[pre] \downarrow \alpha_j \preceq \trace \downarrow \alpha_j$,
		which implies $(\trace[pre] \downarrow \alpha_j, \pstate[pre, ind={\alpha_j}]{w}) \preceq (\trace \downarrow \alpha_j, w_{\alpha_j})$ for $\pstate[pre]{w}_{\alpha_j} = \bot$.
		Thus, $(v, \trace[pre] \downarrow \alpha_j, w'_{\alpha_j}) \in \sem{\alpha_j}$ because $\sem{\alpha_j}$ is prefix-closed by IH.
		Since $\bot = \bot \merge \bot$, we have $\pstate[pre]{w} = \pstate[pre, ind={\alpha_1}]{w} \merge \pstate[pre, ind={\alpha_2}]{w}$.
		Moreover, $\trace[pre] = \trace[pre] \downarrow \gamma$ and $\trace[pre]$ is chronological as prefix $\trace[pre] \preceq \trace$ of the chronological trace $\trace$.
		Finally, $(v, \trace[pre], w') \in \sem{\gamma}$.

		For each $v \in \states$ and $j = 1, 2$,
		some $\leastComputation \in \sem{\alpha_j}$ since $\sem{\alpha_j}$ is total by IH.
		Then $\sem{\gamma}$ is total because $\leastComputation \in \sem{\gamma}$ for each $\pstate{v} \in \states$ since $\statetime{\bot} = \statetime{\bot}$,
		and $\epsilon = \epsilon \downarrow \gamma$, 
		and $\bot = \bot \merge \bot$.
		\qedhere
	\end{enumerate}
\end{proof}

\newcommand{\TrmdL}{\TrmSet{\dL}{}}
\newcommand{\FoldL}{\FolSet{\reals}{\RVar}}

\newcommand{\HpdL}{\SyntacticSet{HP}{\dL}{}}

\newcommand{\todLstate}[1]{\restrict{#1}{\RVar}}
\newcommand{\realStates}{\mathcal{R}}
\newcommand{\todLrela}[1]{\restrict{{#1}}{\realStates \times \realStates}}

\newcommand{\dLvDash}{\vDash_{\dL}}
\newcommand{\dLsem}[1]{\sem{#1}_{\dL}}

In preparation for the proof of conservative extension (\rref{prop:conservative}),
we roughly recap $\dL$ in the following \cite{DBLP:journals/jar/Platzer08}:

\begin{remark}[$\dL$]
	Terms $\TrmdL = \polynoms{\rationals}{\RVar} \subset \Trm$ in $\dL$  are the polynomials in $\RVar$ over $\rationals$.
	Moreover, hybrid programs $\HpdL \subset \Chp$ in $\dL$ are generated by the following grammar: $\alpha, \beta \cceq x \ceq \rp \mid x \ceq * \mid \test{} \mid \evolution{}{} \mid \alpha \seq \beta \mid \alpha \cup \beta \mid \repetition{\alpha}$,
	where $\rp \in \TrmdL$ and $\chi \in \FoldL$ is a formula of first-order real-arithmetic.
	Finally, $\FmldL$ denotes the set of all formulas in $\dL$
	that is generated by the grammar $\varphi, \psi \cceq \rp_1 \sim \rp_2 \mid \neg \varphi \mid \varphi \wedge \psi \mid \varphi \vee \psi \mid \varphi \rightarrow \psi \mid \fa{x} \varphi \mid \ex{x} \varphi \mid [ \alpha ] \psi \mid \langle \alpha \rangle \psi$,
	where $\sim\; \in \{ =, \ge \}$.
	Note that $\FoldL \subset \FmldL$.

	A real state is a map from $\RVar$ to $\reals$ and $\realStates$ denotes the set of all real states.
	With $\dLsem{\cdot} \subseteq \realStates$ we denote semantics of terms and with $\dLvDash\, \subseteq \realStates \times \FmldL$ the satisfaction relation in $\dL$, respectively.
	The semantics $\dLsem{\alpha} \subseteq \realStates \times \realStates$ of a hybrid program $\alpha \in \HpdL$ is inductively defined as follows:
	\begingroup
	\allowdisplaybreaks
	\begin{align*}%
		&\dLsem{x \ceq \rp} 
		= \{ (v, w) \in \realStates \times \realStates \mid w = v \subs{x}{\dLsem{\rp} v} \} \\
		&\dLsem{x \ceq *}  
		= \{ (v, w) \in \realStates \times \realStates \mid w(y) = v(y) \text{ for all } y \neq x \} \\
		&\dLsem{\evolution{}{}} 
		= \{ (\odeSolution(0), \odeSolution(r)) \in \realStates \times \realStates \mid \\
			&\qquad\qquad \odeSolution(\zeta) \dLvDash x' = \rp \wedge \chi \text{ and } \odeSolution(\zeta) = \odeSolution(0) \text{ on } \{x\}^\complement \text{ for all } \zeta \in [0,r] \\
			&\qquad\qquad \text{for a solution } \odeSolution : [0, r] \rightarrow \realStates \text{ with } \text{\footnotesize$\odeSolution(\zeta)(x') = \solutionDerivative{\odeSolution}{x}(\zeta)$} \} \\
		&\dLsem{\test{}}  
		= \{ (v, v) \in  \realStates \times \realStates \mid v \dLvDash \chi \} \\
		&\dLsem{\alpha \cup \beta}  
		= \dLsem{\alpha}  \cup \dLsem{\beta} \\
		&\dLsem{\alpha \seq \beta}  
		= \dLsem{\alpha}  \circ \dLsem{\beta} \sidecondition{$\circ$ is composition of relations} \\
		&\dLsem{\repetition{\alpha}} = \bigcup_{n \in \naturals} \dLsem{\alpha^n} 
	\end{align*}
	\endgroup
\end{remark}

For a state $\pstate{v} \in \states$, 
the real state $\todLstate{v} \in \realStates$ is the restriction of $\pstate{v}$ to $\RVar$.
The restriction $\todLrela{U}$ of a set $U \subseteq \states \times \recTraces \times \botop{\states}$ to $\realStates \times \realStates$ is defined by
\begin{equation*}
	\todLrela{U} = \{ (\todLstate{v}, \todLstate{w}) \in \realStates \times \realStates \mid (\pstate{v}, \trace, \pstate{w}) \in U \text{ with } \pstate{w} \neq \bot \} \text{.}
\end{equation*}

In preparation for the proof of \rref{prop:conservative}, 
\rref{lem:conservativeProgramParts} states that $\dLCHP$ is conservative extension of $\dL$ \wrt terms and formulas in hybrid programs $\HpdL$.
Building on this
\rref{lem:conservativeProgSemantics} shows that the semantics of CHPs is a conservative extension of the program semantics in $\dL$.
We continue to denote the semantics of terms and programs in \dLCHP by $\sem{\cdot}$ and the satisfaction relation by $\vDash$.

\begin{lemma}[Conservative extension of program parts]
	\label{lem:conservativeProgramParts}
	Let $\rp \in \TrmdL$ be a $\dL$-term and $\chi \in \FoldL$ a formula of first-order real arithmetic. 
	Then for all states $\pstate{v} \in \states$ the following holds:
	\begin{enumerate}
		\item $\dLsem{\rp} \todLstate{v} = \sem{\rp} \pstate{v}$
		\item $\todLstate{v} \dLvDash \chi$ iff $\pstate{v} \vDash \chi$
	\end{enumerate}
\end{lemma}
\begin{proof}
	The proof is by induction on the structure of $\rp$ and~$\chi$, respectively.
	\qedhere
\end{proof}

\begin{lemma}[Conservative program semantics] \label{lem:conservativeProgSemantics}
	The program semantics of $\dLCHP$ is a \emph{conservative extension} of the program semantics in $\dL$. 
	That is, 
	\begin{equation*}
		\dLsem{\gamma} = \todLrela{\sem{\gamma}} \; \nonrelevant{= \{ (\todLstate{v}, \todLstate{w}) \mid \computation \in \sem{\gamma} \text{ with } \pstate{w} \neq \bot \}}
	\end{equation*}
	for a hybrid program $\gamma \in \HpdL \subset \Chp$.
\end{lemma}
\begin{proof}
	The proof is by induction on the structure of program $\gamma$. 
	We consider $\alpha^n$ to be structurally smaller than $\repetition{\alpha}$ for all programs $\alpha$ and all $n \in \naturals$. 
	\begin{enumerate}
		\item For $\gamma \in \{ x \ceq \rp, x \ceq *, \test{}, \evolution{}{} \}$,
		we have $\sem{\gamma} = \pLeast \cup M_\gamma$ for some $M_\gamma \subseteq \states \times \{\epsilon\} \times \states$.
		Observe that $\todLrela{\sem{\gamma}} = \todLrela{M_\gamma}$ because $\pLeast = \states \times \{\epsilon\} \times \{\bot\}$.
		Finally, $\dLsem{\gamma} = \todLrela{M_\gamma}$ by definition of $\dLsem{\gamma}$ and $M_\gamma$.
		For instance, in case $\gamma \equiv x \ceq \rp$, 
		we conclude with
		$\todLrela{M_\gamma} 
		= \todLrela{\{ (\pstate{v}, \epsilon, \pstate{w}) \mid w 
		= \pstate{v} \subs{x}{\sem{\rp} \pstate{v}}  \}} 
		= \{ (\todLstate{v}, \todLstate{w} ) \mid \pstate{w} = \pstate{v} \subs{x}{\sem{\rp} \pstate{v}} \}
		= \{ (\todLstate{v}, \todLstate{w} ) \mid \todLstate{w} = (\todLstate{v}) \subs{x}{\dLsem{\rp} \todLstate{v}} \} 
		= \dLsem{x \ceq \rp}$ 
		because $\sem{\rp} \pstate{v} = \dLsem{\rp} \todLstate{v}$ by  \rref{lem:conservativeProgramParts}.
		
		\item For $\gamma \equiv \alpha \cup \beta$, observe that $\todLrela{\sem{\alpha}} \cup \todLrela{\sem{\beta}} = \todLrela{(\sem{\alpha} \cup \sem{\beta})}$.
		Then by IH, $\dLsem{\gamma} = \dLsem{\alpha} \cup \dLsem{\beta} \overset{\text{IH}}{=} \todLrela{\sem{\alpha}} \cup \todLrela{\sem{\beta}} = \todLrela{\sem{\gamma}}$.
		
		\item For $\gamma \equiv \alpha \seq \beta$, we have $\todLrela{\sem{\gamma}} = \todLrela{\botop{\sem{\alpha}}} \cup \todLrela{(\sem{\alpha} \continuation \sem{\beta})}$ since restriction $\todLrela{\cdot\,}$ distributes over union $\cup$. 
		Observe that $\todLrela{\botop{\sem{\alpha}}} = \emptyset$ and $\todLrela{(\sem{\alpha} \continuation \sem{\beta})} = \todLrela{\sem{\alpha}} \circ \todLrela{\sem{\beta}}$, where $\circ$ is composition of relations.
		Thus, $\dLsem{\gamma} = \dLsem{\alpha} \circ \dLsem{\beta} \overset{\text{IH}}{=} \todLrela{\sem{\alpha}} \circ \todLrela{\sem{\beta}} = \todLrela{\sem{\gamma}}$ by IH.
		
		\item $\gamma \equiv \repetition{\alpha}$, then $(\pstate{v}, \pstate{w}) \in \dLsem{\repetition{\alpha}}$ iff $(\pstate{v}, \pstate{w}) \in \dLsem{\alpha^n}$ for some $n \in \naturals$ iff, by IH, $(\pstate{v}, \pstate{w}) \in \todLrela{\sem{\alpha^n}}$ for some $n \in \naturals$ iff $(\pstate{v}, \pstate{w}) \in \todLrela{\sem{\repetition{\alpha}}}$.
		\qedhere
	\end{enumerate}
\end{proof}

\begin{proof}[\rref{prop:conservative}]
	We have to prove that a formula $\varphi \in \Fml \cap \FmldL$ is valid in $\dLCHP$ iff it is valid in $\dL$,
	where $\varphi$ is valid in $\dL$ if $\todLstate{v} \dLvDash \varphi$ for all $\pstate{v} \in \states$, and $\varphi$ is valid in $\dLCHP$ if $\pstate{v} \vDash \varphi$ for all $\pstate{v} \in \states$. 
	The proof is by induction on the structure of $\varphi$:
	\begin{enumerate}
		\item $\varphi \equiv \rp_1 \sim \rp_2$, where $\sim \; \in \{ =, \ge \}$, then $\pstate{v} \vDash \rp_1 \sim \rp_2$ iff $\sem{\rp_1} \pstate{v} \sim \sem{\rp_2} \pstate{v}$ iff,
		by \rref{lem:conservativeProgramParts}, $\dLsem{\rp_1} \todLstate{v} \sim \dLsem{\rp_2} \todLstate{v}$ iff $\todLstate{v} \dLvDash \expr_1 \sim \expr_2$.

		\item For the propositional connectives $\neg, \wedge, \vee, \rightarrow$ and quantifiers $\forall, \exists$ the proof is straightforward by IH.
		
		\item For $\varphi \equiv [ \alpha ] \psi$, observe $(\todLstate{v}, \todLstate{w}) \in \dLsem{\alpha}$ iff $\computation \in \sem{\alpha}$ with $\pstate{w} \neq \bot$ by \rref{lem:conservativeProgSemantics} and $\trace = \epsilon$ since $\alpha \in \Chp$.
		Thus, $\pstate{v} \vDash [ \alpha ] \psi$ iff $\pstate{w} \cdot \trace \vDash \psi$ for all $\computation \in \sem{\alpha}$ with $\pstate{w} \cdot \trace = \pstate{w} \neq \bot$ iff, 
		by IH, $\todLstate{w} \dLvDash \psi$ for all $\computation \in \sem{\alpha}$ with $\pstate{w} \neq \bot$ iff $\todLstate{w} \dLvDash \psi$ for all $(\todLstate{v}, \todLstate{w}) \in \sem{\alpha}_{\dL}$ iff $\todLstate{v} \dLvDash [ \alpha ] \psi$.
		\qedhere
	\end{enumerate}
\end{proof}

\rref{lem:compositionAssoc} shows that the composition operator $\closedComposition$ is associative. 
See \rref{sec:semantics} for the definition of the operators involved.

\begin{lemma}[Composition is associative] \label{lem:compositionAssoc}
	Let $U, M, L \subseteq \states \times \recTraces \times \botop{\states}$. 
	Then $U \closedComposition (M \closedComposition L) = (U \closedComposition M) \closedComposition L$.
\end{lemma}

\begin{proof}
	Let $\computation \in U \closedComposition (M \closedComposition L)$.
	If $\computation \in \botop{U}$, 
	then $\computation \in \botop{U} \cup \botop{(U \continuation M)} \subseteq \botop{(U \closedComposition M)} \subseteq (U \closedComposition M) \closedComposition L$. 
	If $\computation \in U \continuation (M \closedComposition L)$, 
	computations $(\pstate{v}, \trace_1, \pstate{u}) \in U$ and $(\pstate{u}, \trace_2, \pstate{w}) \in (M \closedComposition L)$ with $\trace = \trace_1 \cdot \trace_2$ exist.
	Now, if $(\pstate{u}, \trace_2, \pstate{w}) \in \botop{M}$,
	we have $\computation \in U \continuation \botop{M} = \botop{(U \continuation M)} \subseteq \botop{(U \closedComposition M)}$.
	If $(\pstate{u}, \trace_2, \pstate{w}) \in M \continuation L$, 
	then $\computation \in U \continuation (M \continuation L)$.
	Since $\continuation$ is associative, $\computation \in (U \continuation M) \continuation L \subseteq (U \closedComposition M) \continuation L \subseteq (U \closedComposition M) \closedComposition L$.

	Conversely, let $\computation \in (U \closedComposition M) \closedComposition L$.
	If $\computation \in \botop{(U \closedComposition M)}$, 
	then $\computation \in \botop{U}$ or $\computation \in \botop{(U \continuation M)}$.
	If $\computation \in \botop{U}$, we conclude $\botop{U} \subseteq U \closedComposition (M \closedComposition L)$.
	If $\computation \in \botop{(U \continuation M)}$, 
	computations $(\pstate{v}, \trace_1, \pstate{u}) \in U$ and $(\pstate{u}, \trace_2, \pstate{w}) \in \botop{M} \subseteq M \closedComposition L$ with $\trace = \trace_1 \cdot \trace_2$ exist.
	Thus, $\computation \in U \continuation (M \closedComposition L) \subseteq U \closedComposition (M \closedComposition L)$. 
	Finally, if $\computation \in (U \closedComposition M) \continuation L$, 
	computations $(\pstate{v}, \trace_1, \pstate{u}) \in U \closedComposition M$ and $(\pstate{u}, \trace_2, \pstate{w}) \in L$ with $\trace = \trace_1 \cdot \trace_2$ exist.
 	Since $\pstate{u} \neq \bot$, $(\pstate{v}, \trace_1, \pstate{u}) \in U \continuation M$ such that $\computation \in (U \continuation M) \continuation L$.
	Thus, $\computation \in U \continuation (M \continuation L) \subseteq U \continuation (M \closedComposition L) \subseteq U \closedComposition (M \closedComposition L)$ because $\continuation$ is associative.
	\qedhere
\end{proof}

\section{Static Semantics} \label{app:staticSemantics}

The static semantics captures with free variables $\FV(\cdot)$, 
which variables potentially influence terms, programs, or formulas,
and with bound variables $\BV(\cdot)$,
which variables are written by a program.
Following the static semantics of $\dL$ \cite{DBLP:journals/jar/Platzer17} in general,
we give more precise coincidence properties,
which take the communication history into account.
Therefore, we adapt the notions of accessed channels and communication-aware coincidence for terms from Zwiers \cite{Zwiers_Phd} and lift them to CHPs and formulas of dynamic logic.

\begin{remark}[Communication-awareness] \label{rem:awareness}
	Communication-aware coincidence reflects that the portion of a communication trace $\pstate{v}(\historyVar)$ a formula over trace variable $\historyVar$ depends on shrinks with projections of $\historyVar$. 
	For example, $\len{\historyVar} > 0$ depends on the full trace $\pstate{v}(\historyVar)$ but $\len{\historyVar \downarrow \ch{}} > 0$ only on $\pstate{v}(\historyVar) \downarrow \ch{}$. 
	The sets of accessed channels $\CN(\expr)$, and $\CN(\varphi)$ in terms $\expr$ and formulas $\varphi$ collect the relevant channels of all communication traces $\pstate{v}(\historyVar)$ with $\historyVar \in \TVar$.
\end{remark}

\begin{definition}[Bound variables] \label{def:boundVariables}
	The set of \emph{bound variables} 
	$\BV(\gamma)$ of a program $\gamma \in \Chp$ is inductively defined in \rref{fig:boundFormulaAndProgramVars}.
\end{definition}

\begin{figure}[ht]%
	\vspace{-3em}
	\begin{align*}
		\BV(x \ceq \expr) = \BV(x \ceq *) & = \{ x \} \\
		\BV(\evolution{}{}) & = \{ x, \globalTime \} \\
		\BV(\test{}) & = \emptyset \\
		\BV(\send{}{}{}) & = \{ \historyVar \} \\
		\BV(\receive{}{}{}) & = \{ \historyVar, x \} \\
		\BV(\alpha \cup \beta) = \BV(\alpha \seq \beta) = \BV(\alpha \parOp \beta) & = \BV(\alpha) \cup \BV(\beta) \\
		\BV(\repetition{\alpha}) & = \BV(\alpha)
	\end{align*}%
	\vspace*{-1.5em}
	\caption[Bound variables of formulas and programs]{
		Inductive definition of the set of bound variables $\BV(\gamma)$ of a program $\gamma \in \Chp$.
	}
	\label{fig:boundFormulaAndProgramVars}
	\vspace*{-1em}
\end{figure}

\rref{def:boundVariables} is as usual except that continuous evolution silently (no syntactic occurrence) binds the global time $\globalTime$ since it evolves with every ODE.
The bound effect property in \rref{lem:boundEffect} takes into account that the worlds reachable by a program consist of the final state and the communication as opposed to only considering the final state.

\begin{lemma}[Bound effect property] \label{lem:boundEffect}
	The set of bound variables $\BV(\gamma)$ has the \emph{bound effect property} for a program $\gamma \in \Chp$. 
	That is, $\pstate{v} = \pstate{w}$ on $\BV(\gamma)^\complement \cup \TVar$ and $\pstate{v} = \pstate{w} \cdot \trace$ on $\BV(\gamma)^\complement$ for all $\computation \in \sem{\gamma}$ with $\pstate{w} \neq \bot$.
\end{lemma}
\begin{proof}
	Let $\REC(\gamma) \subseteq \TVar$ be the set of recorder variables in $\gamma$.
	Then $\trace(\historyVar) = \epsilon$ for all $\historyVar \not\in \REC(\gamma)$ by an induction on the structure of $\gamma$,
	where $\trace(\historyVar)$ is the subtrace of $\trace$ of communication recorded by $\historyVar$.
	Hence, $\pstate{w} = \pstate{w} \cdot \trace$ on $\REC(\gamma)^\complement$.
	Since $\REC(\gamma) \subseteq \BV(\gamma)$,
	we obtain $\pstate{v} = \pstate{w} = \pstate{w} \cdot \trace$ on $\BV(\gamma)^\complement \cap \REC(\gamma)^\complement = (\BV(\gamma) \cup \REC(\gamma))^\complement = \BV(\gamma)^\complement$ if only $\pstate{v} = \pstate{w}$ on $\BV(\gamma)^\complement$.
	The latter also shows $\pstate{v} = \pstate{w}$ on $\BV(\alpha)^\complement \cup \TVar$ because $\pstate{v} = \pstate{w}$ on $\TVar$ by another induction on $\gamma$.
	Now, we prove $\pstate{v} = \pstate{w}$ on $\BV(\gamma)^\complement$ by induction on the structure of $\gamma$. 
	We consider $\alpha^n$ to be structurally smaller than $\alpha^*$ for all programs $\alpha$ and all $n \in \naturals$.
	\begin{enumerate}
		\item $\gamma \in \{ x \ceq \rp, x \ceq *, \receive{}{}{x} \}$ and $\computation \in \sem{\gamma}$,
		then $\pstate{w} = \pstate{v} \subs{x}{a}$ for some $a \in \reals$.
		Thus, $\pstate{v} = \pstate{w}$ on $\{ x \}^\complement = \BV(\gamma)^\complement$.
		
		\item $\gamma \equiv \evolution{}{}$ and $\computation \in \sem{\gamma}$, 
		then a solution $\odeSolution: [0, \duration] \rightarrow \states$ exists with $\odeSolution(\zeta) = \pstate{v}$ on $\{ x, \globalTime \}^\complement$ for all $\zeta \in [0, \duration]$ and $\pstate{w} = \odeSolution(\duration)$. 
		Thus, $\pstate{v} = \pstate{w}$ on $\{ x, \globalTime \}^\complement = \BV(\gamma)^\complement$.
		
		\item $\gamma \in \{ \test{}, \send{}{}{} \}$ and $\computation \in \sem{\gamma}$, then $\pstate{v} = \pstate{w}$. Thus, $\pstate{v} = \pstate{w}$ on $\BV(\gamma)^\complement$.

		\item $\gamma \equiv \alpha \seq \beta$ and $\computation \in \sem{\gamma}$, then $\computation \in \sem{\alpha} \continuation \sem{\beta} \subseteq \sem{\gamma}$. 
		Thus, $\pstate{u} \neq \bot$ exists such that $(\pstate{v}, \trace_1, \pstate{u}) \in \sem{\alpha}$, $(\pstate{u}, \trace_2, \pstate{w}) \in \sem{\beta}$, and $\trace = \trace_1 \cdot \trace_2$. 
		By IH, $\pstate{v} = \pstate{u}$ on $\BV(\alpha)^\complement$ and $\pstate{u} = \pstate{w}$ on $\BV(\beta)^\complement$. 
		Thus, $\pstate{v} = \pstate{w}$ on $\BV(\alpha)^\complement \cap \BV(\beta)^\complement = (\BV(\alpha) \cup \BV(\beta))^\complement = \BV(\alpha \seq \beta)^\complement$.
		
		\item $\gamma \equiv \alpha \cup \beta$ and $\computation \in \sem{\gamma}$, then $\computation \in \sem{\alpha}$ or $\computation \in \sem{\beta}$.
		By IH, $\pstate{v} = \pstate{w}$ on $\BV(\alpha)^\complement$ or on $\BV(\beta)^\complement$, respectively.
		Always, $\pstate{v} = \pstate{w}$ on $\BV(\alpha)^\complement \cap \BV(\beta)^\complement = (\BV(\alpha) \cup \BV(\beta))^\complement = \BV(\alpha \cup \beta)^\complement$.
		
		\item $\gamma \equiv \repetition{\alpha}$ and $\computation \in \sem{\gamma}$, then $\computation \in \sem{\alpha^n}$ for some $n \in \naturals$. 
		Since $\alpha^n$ is structurally smaller than $\repetition{\alpha}$, we obtain $\pstate{v} = \pstate{w}$ on $\BV(\alpha^n)^\complement$ by IH. 
		Finally, $\BV(\alpha^n)^\complement \supseteq \BV(\alpha)^\complement = \BV(\repetition{\alpha})^\complement$ (note that $\alpha^0 \equiv \test{\true}$).

		\item $\gamma \equiv \alpha \parOp \beta$ and $\computation \in \sem{\gamma}$, then $\computation[proj=\alpha] \in \sem{\alpha}$, $\computation[proj=\beta] \in \sem{\beta}$, and $\pstate{w} = \pstate{w}_\alpha \merge \pstate{w}_\beta$. 
		By IH, $\pstate{v} = \pstate{w}_\beta$ on $\BV(\beta)^\complement$. 
		Moreover, $\pstate{w}_\beta = \pstate{w}$ on $\BV(\alpha)^\complement$ by definition of $\merge$ in \rref{sec:semantics}. 
		Thus, $\pstate{v} = \pstate{w}$ on $\BV(\alpha)^\complement \cap \BV(\beta)^\complement = (\BV(\alpha) \cup \BV(\beta))^\complement = \BV(\alpha \parOp \beta)^\complement$.
		\qedhere
	\end{enumerate}
\end{proof}

\begin{definition}[Parameters of terms]
	\label{def:freeTermParameters}
	The sets of free variables $\FV(\expr)$ and accessed channels $\CN(\expr)$ of a term $\expr \in \Trm$ are inductively defined in \rref{fig:freeTermParameters}.
\end{definition}

\begin{figure}[h!]
	\vspace*{-1em}
	\begin{minipage}{.5\textwidth}
		\begin{align*}
			\FV(\arbitraryVar) & = \{ \arbitraryVar \} 
			\sidecondition{$\arbitraryVar \in \V$} \\
			\FV(\kappa) & = \;\emptyset \\
			\FV(\expr_1 \bowtie \expr_2) & = \FV(\expr_1) \cup \FV(\expr_2) \\
			\FV(f(\te, \ie)) & = \FV(\te) \cup \FV(\ie) \\
			\FV(\len{\te}) & =  \FV(\te) \\
			\FV(\comItem{\ch{}, \rp_1, \rp_2}) & = \FV(\rp_1) \cup \FV(\rp_2) \\
			\FV(\te \downarrow \cset) & = \FV(\te)
		\end{align*}
	\end{minipage}
	\begin{minipage}{.5\textwidth}
		\begin{align*}
			\CN(\historyVar) & = \Chan \sidecondition{$\historyVar \in \TVar$}\\
			\CN(\arbitraryVar) = \CN(\kappa) & = \emptyset \sidecondition{$\arbitraryVar \not \in \TVar$}\\
			\CN(\expr_1 \bowtie \expr_2) & = \CN(\expr_1) \cup \CN(\expr_2) \\
			\CN(f(\te, \ie)) & = \CN(\te) \cup \CN(\ie) \\
			\CN(\len{\te}) & = \CN(\te) \\
			\CN(\comItem{\ch{}, \rp_1, \rp_2}) & = \emptyset \\
			\CN(\te \downarrow \cset) & = \CN(\te) \cap \cset
		\end{align*}
	\end{minipage}
	\caption[Free parameters of terms]{Inductive definition of free variables $\FV(\expr)$ and accessed channels $\CN(\expr)$ of a term $\expr \in \Trm$, where $\arbitraryVar \in \V$, $\kappa$ is a constant of any sort, $\bowtie \; \in \{ +, \cdot \}$ is any operator, and $f(\te, \ie) \in \{ \chan{\at{\te}{\ie}}, \val{\at{\te}{\ie}}, \stamp{\at{\te}{\ie}} \}$ is a term over $\te$ and $\ie$. }
	\label{fig:freeTermParameters}
	\vspace*{-1em}
\end{figure}

Unsuprisingly, a trace variable $\historyVar$ potentially accesses all channels $\Chan$ by \rref{def:freeTermParameters},
and with a projection on $\cset$ the accessed channels shrink by $\cset$.
The communication item $\comItem{\ch{}, \rp_1, \rp_2}$ does not access any channels since $\rp_1, \rp_2 \in \FolRA$.

For use in the coincide properties,
we lift projection for traces to states by defining that $\pstate{v} \downarrow \cset = \pstate{v}$ on $\RVar \cup \NVar$ and $(\pstate{v} \downarrow \cset)(\historyVar) = \pstate{v}(\historyVar) \downarrow \cset$ for all $\historyVar \in \TVar$ and $\cset \subseteq \Chan$.
As for traces, we often write $\pstate{v} \downarrow \expr$ instead of $\pstate{v} \downarrow \CN(\expr)$.

\begin{lemma}[Coincidence for terms] \label{lem:termCoincidence}
	The pair of free variables $\FV(\expr)$ and accessed channels $\CN(\expr)$ has the \emph{communication-aware coincidence property} for a term $\expr \in \Trm$. 
	That is, if $\pstate{v} \downarrow \expr = \pstate[alt]{v} \downarrow \expr$ on $\FV(\expr)$, then $\sem{\expr} \pstate{v} = \sem{\expr} \pstate[alt]{v}$.
	In particular, $\sem{\expr} \pstate{v} = \sem{\expr} \pstate{v}$ if $\pstate{v} = \pstate[alt]{v}$ on $\FV(\expr)$.
\end{lemma}
\begin{proof}
	\Wlossg we push down projections in $\expr$ as far as possible by applying the axioms \RuleName{projNeutral}\!\!\!, \RuleName{concatDist}\!\!\!, and \RuleName{projCut}\!\!\!, and \RuleName{projIn}\!\!\!, and \RuleName{projNotIn} (see \rref{fig:traceAlgebra}) from left to right.
	This is justified since valuation $\sem{\expr}$, free variables $\FV(\expr)$, and accessed channels $\CN(\expr)$ are invariant under this rewritings.
	In the resulting normal form only raw trace variables occur in the scope of projections.
	Crucially, the communication item $\comItem{\ch{}, \rp_1, \rp_2}$ does not contain further trace variables since $\rp_1, \rp_2 \in \FolRA$.
	Now, the proof is by induction on the structure of the term $\expr$:
	\begin{enumerate}
		\item $\expr \equiv \arbitraryVar$ with $\arbitraryVar \in \V$,
		then $\FV(\expr) = \{ \arbitraryVar \}$.
		If $\arbitraryVar \in \V_{\reals\cup\naturals} \cup \{ \globalTime \}$, 
		then $\pstate{v}(\arbitraryVar) = \pstate[alt]{v}(\arbitraryVar)$ by premise.
		If $\arbitraryVar \in \TVar$, then $\pstate{v}(\arbitraryVar) = \pstate{v}(\arbitraryVar) \downarrow \expr = \pstate[alt]{v}(\arbitraryVar) \downarrow \expr = \pstate[alt]{v}(\arbitraryVar)$ since $\CN(\expr) = \Chan$ and $\pstate{v} \downarrow \expr = \pstate[alt]{v} \downarrow \expr$ on $\{\arbitraryVar\}$ by premise.
		Thus, $\sem{\expr} \pstate{v} = \pstate{v}(\arbitraryVar) = \pstate[alt]{v}(\arbitraryVar) = \sem{\expr} \pstate[alt]{v}$.
		
		\item $\expr \equiv \kappa$, where $\kappa$ is a constant of any sort, 
		then $\sem{\expr} \pstate{v} = \kappa = \sem{\expr} \pstate[alt]{v}$.
		
		\item $\expr \equiv \expr_1 \bowtie \expr_2$ with $\bowtie \;\in \{ +, \cdot \}$,
		then $\pstate{v} \downarrow \expr = \pstate[alt]{v} \downarrow \expr$ on $\FV(\expr_j)$ since $\FV(\expr_j) \subseteq \FV(\expr)$.
		For $\historyVar \in \TVar$, observe that $\pstate{v}(\historyVar) \downarrow D = \pstate[alt]{v}(\historyVar) \downarrow D$ implies $\pstate{v}(\historyVar) \downarrow E = \pstate[alt]{v}(\historyVar) \downarrow E$ if $E \subseteq D \subseteq \Chan$. 
		Hence, we obtain $\pstate{v} \downarrow \expr_j = \pstate[alt]{v} \downarrow \expr_j$ on $\FV(\expr_j)$ from $\CN(\expr_j) \subseteq \CN(\expr)$.
		By IH, we conclude as follows: 
		\begin{equation*}
			\sem{\expr} \pstate{v} = \sem{\expr_1} \pstate{v} \bowtie \sem{\expr_2} \pstate{v} \overset{\text{IH}}{=} \sem{\expr_1} \pstate[alt]{v} \bowtie \sem{\expr_2} \pstate[alt]{v} = \sem{\expr} \pstate[alt]{v}
		\end{equation*}
		
		\item $\expr \in \{ \chan{\at{\te}{\ie}}, \val{\at{\te}{\ie}}, \stamp{\at{\te}{\ie}}, \comItem{\ch{}, \rp_1, \rp_2}, \len{\te} \}$, 
		then we conclude similar to the last case by IH.
		
		\item $\expr \equiv \te \downarrow \cset$, then $\te \equiv \historyVar$ for some $\historyVar \in \TVar$ since $\expr$ is in normal form by assumption.
		From $\CN(\expr) = \cset$ and $\FV(\expr) = \{ \historyVar \}$, 
		we obtain $\pstate{v}(\historyVar) \downarrow \cset = \pstate[alt]{v}(\historyVar) \downarrow \cset$ by premise such that $\sem{\expr} \pstate{v} = \pstate{v}(\historyVar) \downarrow \cset = \pstate[alt]{v}(\historyVar) \downarrow \cset = \sem{\expr} \pstate[alt]{v}$.
		\qedhere
	\end{enumerate}
\end{proof}

We adapt must-bound variables from $\dL$ \cite{DBLP:journals/jar/Platzer17} leading to a fine-grained definition of free variables in formulas and programs.
Defining $\FV([\alpha] \psi)$ as $\FV(\alpha) \cup \FV(\psi)$ would be sound \wrt the coincidence property for formulas but imprecise because $[ x \ceq 0 ] x \ge 0$ does not depend on $x$. 
However, defining $\FV([\alpha] \psi)$ as $\FV(\alpha) \cup (\FV(\psi) \setminus \BV(\alpha))$ is unsound because $[ x \ceq 0 \cup y \ceq 0 ] x \ge 0$ depends on~$x$ as $x$ is not bound on all execution paths.
The must-bound variables are the variables that are bound on all execution paths of a program.
Hence, they may be removed from $\FV(\psi)$ soundly in $\FV([ \alpha ] \psi) = \FV(\alpha) \cup (\FV(\psi) \setminus \MBV(\alpha))$.

\begin{definition}[Must-bound variables] \label{def:mustBoundVariables}
	The set of \emph{must-bound variables} $\MBV(\gamma)$ of a program $\gamma \in \Chp$ is inductively defined in \rref{fig:mustBoundVars}.
\end{definition}

\begin{figure}[ht]
	\vspace*{-1em}
	\begin{align*}
		\MBV(\alpha) & = \BV(\alpha) \quad \sidecondition{$\alpha$ is any atomic program}\\
		\MBV(\alpha \cup \beta) & = \MBV(\alpha) \cap \MBV(\beta) \\
		\MBV(\alpha \seq \beta) = \MBV(\alpha \parOp \beta) & = \MBV(\alpha) \cup \MBV(\beta) \\
		\MBV(\repetition{\alpha}) & = \emptyset
	\end{align*}
	\caption[Must-bound variables of a program]{Inductive definition of the must-bound variables $\MBV(\gamma)$ of a program $\gamma \in \Chp$.}
	\vspace*{-1em}
	\label{fig:mustBoundVars}
	\vspace*{-1em}
\end{figure}

\begin{definition}[Free variables of programs] \label{def:freeProgramParameters}
	The set of free variables $\FV(\gamma)$ of a program $\gamma \in \Chp$ is inductively defined in \rref{fig:freeProgramParameters}.
\end{definition}

\begin{figure}[h!]
	\vspace*{-3em}
	\begin{align*}
		\FV(x \ceq \rp) & = \FV(\rp) \\
		\FV(x \ceq *) & = \emptyset \\
		\FV(\test{}) & = \FV(\chi) \\
		\FV(\evolution{}{}) & = \{ x, \globalTime \} \cup \FV(\rp) \cup \FV(\chi) \\
		\FV(\send{}{}{}) & = \{\historyVar, \globalTime\} \cup \FV(\rp) \\
		\FV(\receive{}{}{}) & = \{\historyVar, \globalTime\} \\
		\FV(\alpha \seq \beta) & = \FV(\alpha) \cup (\FV(\beta) \setminus \MBV(\alpha)) \\
		\FV(\alpha \cup \beta) = \FV(\alpha \parOp \beta) & = \FV(\alpha) \cup \FV(\beta) \\
		\FV(\repetition{\alpha}) & = \FV(\alpha)
	\end{align*}
	\vspace*{-1.5em}
	\caption[Free variables of programs]{%
		Inductive definition of free variables $\FV(\gamma)$ of a program $\gamma \in \Chp$.}
	\label{fig:freeProgramParameters}
	\vspace*{-1em}
\end{figure}

\rref{def:freeProgramParameters} is as usual except that a continuous evolution silently (no syntactic occurrence) depends on the global time $\globalTime$ as it evolves like $x$ with the ODE.
Moreover, $\globalTime$ is free in $\send{}{}{}$ and $\receive{}{}{}$ because the current time is recorded as timestamp for each communication event.

\begin{definition}[Free parameters of formulas] \label{def:freeFormulaParameters}
	The sets of free variables $\FV(\phi)$ and accessed channels $\CN(\phi)$ of a formula $\phi \in \Fml$ is inductively defined in \rref{fig:freeFormulaParameters}.
\end{definition}

\begin{figure}[h!]
	\begin{minipage}{.5\textwidth}
		\begin{align*}
			\FV(\expr_1 \sim \expr_2) & = \FV(\expr_1) \cup \FV(\expr_2) \\
			\FV(\neg \varphi) & = \FV(\varphi) \\
			\FV(\varphi \rightleftharpoons \psi) & = \FV(\varphi) \cup \FV(\psi) \\
			\FV(\quantor{\arbitraryVar} \varphi) & = \FV(\varphi) \setminus \{ \arbitraryVar \} \\
			\FV([ \alpha ] \psi) & = \FV(\alpha) \cup (\FV(\psi) \setminus \MBV(\alpha)) \\
			\FV([ \alpha ] \ac \psi) & = \FV([ \alpha ] \psi) \cup \FV(\A) \cup \FV(\C)
		\end{align*}
	\end{minipage}
	\begin{minipage}{.5\textwidth}
		\begin{align*}
			\CN(\expr_1 \sim \expr_2) & = \CN(\expr_1) \cup \CN(\expr_2) \\
			\CN(\neg \varphi) & = \CN(\varphi) \\
			\CN(\varphi \rightleftharpoons \psi) & = \CN(\varphi) \cup \CN(\psi) \\
			\CN(\quantor{\arbitraryVar} \varphi) & = \CN(\varphi) \\
			\CN( [ \alpha ] \psi) & = \CN(\psi) \\
			\CN([ \alpha ] \ac \psi) & = \CN([ \alpha ] \psi) \cup \CN(\A) \cup \CN(\C)
		\end{align*}
	\end{minipage}
	\caption[Free parameters of formulas]{%
		Inductive definition of free variables $\FV(\phi)$ and accessed channels $\CN(\phi)$ of a formula $\phi \in \Fml$, 
		where $\sim\; \in \{ =, \ge, \preceq \}$, and $\rightleftharpoons\; \in \{ \wedge, \vee, \rightarrow \}$, and $\mathcal{Q} \in \{ \forall, \exists \}$.}
	\label{fig:freeFormulaParameters}
	\vspace*{-1em}
\end{figure}

\begin{lemma}[Coincidence for formulas] \label{lem:formulaCoincidence}
	The pair of free variables $\FV(\phi)$ and accessed channels $\CN(\phi)$ has the \emph{communication-aware coincidence property} for a formula $\phi \in \Fml$: 
	If $\pstate{v} \downarrow \phi = \pstate[alt]{v} \downarrow \phi$ on $\FV(\phi)$, then: 
	$\pstate{v} \in \sem{\phi}$ iff $\pstate[alt]{v} \in \sem{\phi}$.
\end{lemma}

\begin{lemma}[Coincidence for programs] \label{lem:programCoincidence}
	The set of free variables $\FV(\gamma)$ has the \emph{communication-aware coincidence property} for a program $\gamma \in \Chp$: 
	If $\pstate{v} \downarrow \cset = \pstate[alt]{v} \downarrow \cset$ on $\varset \supseteq \FV(\gamma)$ with $\cset \subseteq \Chan$ and $\computation \in \sem{\gamma}$, then $\computation[alt] \in \sem{\gamma}$ exists with $\pstate{w} \downarrow \cset = \pstate[alt]{w} \downarrow \cset$ on $\varset \cup \MBV(\gamma)$, and $\trace = \trace[alt]$, and ($\pstate{w} = \bot$ iff $\pstate[alt]{w} = \bot$).
\end{lemma}

\begin{proof}[\rref{lem:formulaCoincidence} and \rref{lem:programCoincidence}]
	We generalize the proof of the coincidence properties for $\dL$ \cite[Lemma 17]{DBLP:journals/jar/Platzer17} to CHPs and formulas featuring ac-modalities. 
	The proof is by simultaneous induction on the structure of formulas and programs.
	
	We start with induction on formula $\phi$. 
	By premise, $\pstate{v} \downarrow \phi = \pstate[alt]{v} \downarrow \phi$ on $\FV(\phi)$. 
	Then we prove that $\pstate{v} \in \sem{\phi}$ implies $\pstate[alt]{v} \in \sem{\phi}$. 
	For the converse, swap $\pstate{v}$ and $\pstate[alt]{v}$.
	We consider $\fa{\arbitraryVar} \neg \varphi$ to be structurally smaller than $\ex{\arbitraryVar} \varphi$ for any formula $\varphi$.
	Note that $\pstate{v} \downarrow D = \pstate[alt]{v} \downarrow D$ implies $\pstate{v} \downarrow E = \pstate[alt]{v} \downarrow E$ if $E \subseteq D \subseteq \Chan$.
	We use this fact throughout the proof without further mentioning it.
	
	\begin{enumerate}
		\item $\pstate{v} \vDash \expr_1 \sim \expr_2$, where $\sim \; \in \{ =, \ge, \preceq \}$, iff $\sem{\expr_1} \pstate{v} \sim \sem{\expr_2} \pstate{v}$ holds.
		For $j = 1, 2$, $\FV(\expr_j) \subseteq \FV(\expr_1 \sim \expr_2)$ and $\CN(\expr_j) \subseteq \CN(\expr_1 \sim \expr_2)$ such that $\pstate{v} \downarrow \expr_j = \pstate[alt]{v} \downarrow \expr_j$ on $\FV(\expr_j)$.
		Hence, $\sem{\expr_j} \pstate{v} = \sem{\expr_j} \pstate[alt]{v}$ for $j = 1, 2$ by coincidence for terms (\rref{lem:termCoincidence}), which implies $\sem{\expr_1} \pstate[alt]{v} \sim \sem{\expr_2} \pstate[alt]{v}$.
		Thus, $\pstate[alt]{v} \vDash \expr_1 \sim \expr_2$.
		
		\item $\pstate{v} \vDash \neg \varphi$ implies $\pstate{v} \not \in \sem{\varphi}$. 
		By IH, $\pstate[alt]{v} \not \in \sem{\varphi}$ since $\FV(\neg \varphi) = \FV(\varphi)$ and $\CN(\neg \varphi) = \CN(\varphi)$. 
		Hence, $\pstate[alt]{v} \vDash \neg \varphi$.
		
		\item $\pstate{v} \vDash \varphi \wedge \psi$ iff $\pstate{v} \vDash \varphi$ and $\pstate{v} \vDash \psi$.
		Since $\FV(\chi) \subseteq \FV(\varphi \wedge \psi)$ and $\CN(\chi) \subseteq \CN(\varphi \wedge \psi)$ for $\chi \in \{ \varphi, \psi \}$,
		we have $\pstate{v} \downarrow \chi = \pstate[alt]{v} \downarrow \chi$ on $\FV(\chi)$.
		Then $\pstate[alt]{v} \vDash \varphi$ and $\pstate[alt]{v} \vDash \psi$ by IH.
		Finally, $\pstate[alt]{v} \vDash \varphi \wedge \psi$.
		
		\item The cases $\varphi \vee \psi$ and $\varphi \rightarrow \psi$ can be handled analogously to $\varphi \wedge \psi$.
		
		\item $\pstate{v} \vDash \fa{\arbitraryVar} \varphi$ iff $\pstate{v} \subs{\arbitraryVar}{a} \vDash \varphi$ for all $a \in \type(\arbitraryVar)$. 
		We have $\pstate{v} \subs{\arbitraryVar}{a} \downarrow \varphi = \pstate[alt]{v} \subs{\arbitraryVar}{a} \downarrow \varphi$ on $\FV(\varphi)$ since $\pstate{v} \downarrow \varphi = \pstate[alt]{v} \downarrow \varphi$ on $\FV(\fa{\arbitraryVar} \varphi) = \FV(\varphi) \setminus \{ \arbitraryVar \}$ and $\CN(\fa{\arbitraryVar} \varphi) = \CN(\varphi)$. 
		Hence, by IH, $\pstate[alt, subs=\subs{\arbitraryVar}{a}]{v} \vDash \varphi$ for all $a \in \type(\arbitraryVar)$, 
		which implies $\pstate[alt]{v} \vDash \fa{\arbitraryVar} \varphi$.
		
		\item $\pstate{v} \vDash \ex{\arbitraryVar} \varphi$ iff $\pstate{v} \nvDash \fa{\arbitraryVar} \neg \varphi$, 
		which implies $\pstate[alt]{v} \nvDash \fa{\arbitraryVar} \neg \varphi$ by IH because $\FV(\fa{\arbitraryVar} \neg \varphi) = \FV(\ex{\arbitraryVar} \varphi)$, and $\CN(\fa{\arbitraryVar} \neg \varphi) = \CN(\ex{\arbitraryVar} \varphi)$, and we considered $\fa{z} \neg \varphi$ to be structurally smaller than $\ex{z} \varphi$.
		Hence, $\pstate[alt]{v} \vDash \ex{\arbitraryVar} \varphi$.

		\item $\phi \equiv [ \alpha ] \psi$, then let $\pstate{v} \vDash \phi$. 
		To show $\pstate[alt]{v} \vDash \phi$, let $\computation[alt] \in \sem{\alpha}$ with $\pstate[alt]{w} \neq \bot$.
		Since $\pstate{v} \downarrow \phi = \pstate[alt]{v} \downarrow \phi$ on $\FV(\phi) \supseteq \FV(\alpha)$, 
		a computation $\computation \in \sem{\alpha}$ with $\pstate{w} \neq \bot$, and $\pstate{w} \downarrow \phi = \pstate[alt]{w} \downarrow \phi$ on $\FV(\phi) \cup \MBV(\alpha)$, and $\trace = \trace[alt]$ exists by the simultaneous IH.
		Now, $\pstate{w} \cdot \trace \vDash \psi$ because $\pstate{v} \vDash \phi$.
		By $\FV(\phi) = \FV(\alpha) \cup (\FV(\psi) \setminus \MBV(\alpha))$, 
		we obtain $\pstate{w} \downarrow \phi = \pstate[alt]{w} \downarrow \phi$ on $\FV(\psi) \subseteq \FV(\phi) \cup \MBV(\alpha)$, 
		which implies $\pstate{w} \downarrow \psi = \pstate[alt]{w} \downarrow \psi$ on $\FV(\psi)$ by $\CN(\phi) = \CN(\psi)$.
		Furthermore, $(\pstate{w} \cdot \trace) \downarrow \psi = (\pstate[alt]{w} \cdot \trace[alt]) \downarrow \psi$ because $\trace = \trace[alt]$.
		Thus, IH is applicable on $\pstate{w} \cdot \trace \vDash \psi$ such that $\pstate[alt]{w} \cdot \trace[alt] \vDash \psi$ by IH.
			
		\item $\phi \equiv [ \alpha ] \ac \psi$, then let $\pstate{v} \vDash \phi$.
		In order to show $\pstate[alt]{v} \vDash \phi$ let $\computation[alt] \in \sem{\alpha}$.
		Since $\pstate{v} \downarrow \phi = \pstate[alt]{v} \downarrow \phi$ on $\FV(\phi) \supseteq \FV(\alpha)$, 
		a computation $\computation \in \sem{\alpha}$ with $\pstate{w} \downarrow \phi = \pstate[alt]{w} \downarrow \phi$ on $\FV(\phi) \cup \MBV(\alpha)$, and $\trace = \trace[alt]$, and ($\pstate{w} = \bot$ iff $\pstate[alt]{w} = \bot$) exists by the simultaneous IH.
		
		For \acCommit, assume $\proppre{\pstate[alt]{v}}{\trace[alt]} \vDash \A$.
		Then $\proppre{\pstate{v}}{\trace} \vDash \A$ by IH since $\trace[alt] = \trace$ and $\pstate{v} \downarrow \phi = \pstate[alt]{v} \downarrow \phi$ on $\FV(\phi)$ implies $\pstate{v} \downarrow \A = \pstate[alt]{v} \downarrow \A$ on $\FV(\A)$ by $\FV(\A) \subseteq \FV(\phi)$ and $\CN(\A) \subseteq \CN(\phi)$.
		Hence, $\pstate{v} \cdot \trace \vDash \C$ because $\pstate{v} \vDash \phi$, 
		which in turn implies $\pstate[alt]{v} \cdot \trace[alt] \vDash \C$ by IH because $(\pstate{v} \cdot \trace) \downarrow \C = (\pstate[alt]{v} \cdot \trace[alt]) \downarrow \C$ on $\FV(\C)$.
		
		For \acPost, assume $\pstate[alt]{w} \neq \bot$ and $\preeq{\pstate[alt]{v}}{\trace[alt]} \vDash \A$.
		Since $\pstate{v} \downarrow \phi = \pstate[alt]{v} \downarrow \phi$ on $\FV(\phi) \supseteq \FV(\A)$ with $\CN(\phi) \supseteq \CN(\A)$, we obtain $\preeq{\pstate{v}}{\trace} \vDash \A$ by IH.
		Hence, $\pstate{w} \cdot \trace \vDash \psi$ by $\pstate{v} \vDash \phi$.
		Moreover, $\pstate{w} \downarrow \psi = \pstate[alt]{w} \downarrow \psi$ on $\FV(\psi)$ because $\FV(\psi) \subseteq \FV(\phi) \cup \MBV(\alpha)$ and $\CN(\psi) \subseteq \CN(\phi)$.
		Finally, $\pstate[alt]{w} \cdot \trace \vDash \psi$ by IH.
	\end{enumerate}

	Now, we proceed with induction on program $\alpha$.
	By premise, $\pstate{v} \downarrow \cset = \pstate[alt]{v} \downarrow \cset$ on $\varset \supseteq \FV(\gamma)$. 
	\Wlossg $\varset \cap \TVar = \emptyset$ because $\pstate{v} = \pstate{w}$ and $\pstate[alt]{v} = \pstate[alt]{w}$ both on $(\BV(\alpha) \setminus \TVar)^\complement$ by the bound effect property (\rref{lem:boundEffect}) such that $\pstate[alt]{w} \downarrow \cset = \pstate[alt]{v} \downarrow \cset = \pstate{v} \downarrow \cset = \pstate{w} \downarrow \cset$ on $\varset \cap \TVar$.
	Hence, $\pstate{v} = \pstate{v} \downarrow \cset = \pstate[alt]{v} \downarrow \cset = \pstate[alt]{v}$ on $\varset$.
	Further, we consider $\alpha^n$ to be structurally smaller than $\repetition{\alpha}$ for all programs $\alpha$ and all $n \in \naturals$, and $x \ceq * \seq \send{}{}{x}$ to be smaller than $\receive{}{}{}$.
	
	If $\gamma$ is an atomic non-communicating program, 
	we can handle $\computation \in \sem{\gamma}$ with $\pstate{w} = \bot$ uniformly as follows: 
	Let $(v, \trace, \bot) \in \sem{\gamma}$. 
	Then $\trace = \epsilon$, and we define $\trace[alt] = \epsilon$ and $\pstate[alt]{w} = \bot$. 
	By totality and prefix-closedness (\rref{prop:prefixClosedAndTotal}), $\computation[alt] \in \sem{\gamma}$. 
	Further, $\pstate{w} = \pstate[alt]{w}$ on $\varset \cup \MBV(\gamma)$ since $\bot = \bot$ on any set of variables. 
	Therefore, we assume $\pstate{w} \neq \bot$ \wlossg in the first four cases below.
	
	\begin{enumerate}
		\item $\computation \in \sem{x \ceq \rp}$, then $\trace = \epsilon$ and $\pstate{w} = \pstate{v} \subs{x}{\sem{\rp} \pstate{v}}$. 
		We define $\trace[alt] = \epsilon$ and $\pstate[alt]{w} = \pstate[alt]{v} \subs{x}{\sem{\rp} \pstate[alt]{v}}$ such that $\computation[alt] \in \sem{x \ceq \rp}$.
		By coincidence (\rref{lem:termCoincidence}), $\pstate{w}(x) = \sem{\rp} \pstate{v} = \sem{\rp} \pstate[alt]{v} = \pstate[alt]{w}(x)$ since $\pstate{v} = \pstate[alt]{v}$ on $\varset \supseteq \FV(x \ceq \rp) = \FV(\rp)$.
		Moreover, $\pstate{w} = \pstate[alt]{w}$ on $\varset \setminus \{ x \}$ because $\pstate{v} = \pstate[alt]{v}$ on $\varset$, and $\pstate{v} = \pstate{w}$ and $\pstate[alt]{v} = \pstate[alt]{w}$ on $\{ x \}^\complement$.  
		Hence, $\pstate{w} = \pstate[alt]{w}$ on $\varset \cup \MBV(x \ceq \rp)$ since $\MBV(x \ceq \rp) = \{ x \}$, and $\trace = \trace[alt]$, and $\pstate[alt]{w} \neq \bot$.
		
		\item $\computation \in \sem{x \ceq *}$, then $\trace = \epsilon$ and $\pstate{w} = \pstate{v}\subs{x}{a}$ for some $a \in \reals$.
		We define $\trace[alt] = \epsilon$ and $\pstate[alt]{w} = \pstate[alt]{v} \subs{x}{a}$. 
		Obviously $\pstate{v} = \pstate{w}$ and $\pstate[alt]{v} = \pstate[alt]{w}$ on $ \{ x \}^\complement$.
		Thus, $\pstate{v} = \pstate[alt]{v}$ on $\varset$ implies $\pstate{w} = \pstate[alt]{w}$ on $\varset \setminus \{ x \}$. 
		Moreover, $\pstate{w}(x) = \pstate{v}\subs{x}{a}(x) = \pstate[alt]{v} \subs{x}{a}(x) = \pstate[alt]{w}(x)$.
		Overall, $\computation[alt] \in \sem{\alpha}$ exists with $\pstate{w} = \pstate[alt]{w}$ on $\varset \cup \MBV(x \ceq *)$ since $\MBV(x \ceq *) = \{ x \}$.
		Finally, $\trace = \trace[alt]$ and $\pstate[alt]{w} \neq \bot$.
		
		\item $\gamma \equiv \evolution{}{}$ and $\computation \in \sem{\gamma}$, then $\trace = \epsilon$ and a solution $\odeSolution: [0, \duration] \rightarrow \states$ exists with $\pstate{v} = \odeSolution(0)$, and $\pstate{w} = \odeSolution(\duration)$, and $\odeSolution(\zeta) \vDash x' = \rp \wedge \chi$ and $\odeSolution(\zeta)(\globalTime) = \pstate{v}(\globalTime) + \zeta$ for all $\zeta \in [0, \duration]$, where $\odeSolution(\zeta)(x') = \solutionDerivative{\odeSolution}{x}(\zeta)$, and $\odeSolution(\zeta)(y) = \pstate{v}(y)$ for $y \not \in \{ x, \globalTime \}$.
		We define a solution $\odeSolution[alt]: [0, \duration] \rightarrow \states$ of equal duration $\duration$ for $\zeta \in [0, \duration]$, 
		by $\odeSolution[alt](\zeta) = \odeSolution(\zeta)$ on $\{ x, \globalTime \}$ and if $y \not \in \{ x, \globalTime \}$ by $\odeSolution[alt](\zeta)(y) = \pstate[alt]{v}(y)$.
		
		Since $\odeSolution(\zeta)(x) = \odeSolution[alt](\zeta)(x)$ for all $\zeta \in [0, \duration]$,
		we have $\odeSolution(\zeta)(x') = \odeSolution[alt](\zeta)(x')$.
		Moreover, $\odeSolution[alt](\zeta) = \odeSolution(\zeta)$ on $\{ x, \globalTime \}$ and $\odeSolution(\zeta) = \pstate{v} = \pstate[alt]{v} = \odeSolution[alt](\zeta)$ on $\varset \setminus \{ x, \globalTime \}$ 
		imply $\odeSolution(\zeta) = \odeSolution[alt](\zeta)$ on $\varset \cup \MBV(\gamma)$ because $\MBV(\gamma) = \{ x, \globalTime \}$.
		Since $\varset \supseteq \FV(\gamma) \supseteq \FV(\rp) \cup \FV(\chi)$ by premise,
 		$\sem{\rp} \odeSolution(\zeta) = \sem{\rp} \odeSolution[alt](\zeta)$ by coincidence for terms (\rref{lem:termCoincidence}).
		Moreover, $\odeSolution(\zeta) \vDash \chi$ iff $\odeSolution[alt](\zeta) \vDash \chi$ by the simultaneous IH.
		Thus, $\odeSolution[alt](\zeta) \vDash x' = \rp \wedge \chi$ for all $\zeta \in [0, \duration]$. 
		Moreover, $\odeSolution[alt](\zeta)(y) = \pstate[alt]{v}(y)$ for $y \not \in \{ x, \globalTime \}$ by definition of $\odeSolution[alt]$. 
		If we define $\pstate[alt]{v} = \odeSolution[alt](0)$, and $\trace[alt] = \epsilon$, and $\pstate[alt]{w} = \odeSolution[alt](\duration)$, 
		then $\computation[alt] \in \sem{\alpha}$.
		Finally, $\pstate[alt]{w} = \odeSolution[alt](\duration) = \odeSolution(\duration) = \pstate{w}$ on $\varset \cup \MBV(\gamma)$, and $\trace = \trace[alt]$, and $\pstate[alt]{w} \neq \bot$. 
		
		\item $\computation \in \sem{\test{}}$, 
		then $\trace = \epsilon$, and $\pstate{w} = \pstate{v}$, and $\pstate{v} \vDash \chi$. 
		We define $\trace[alt] = \epsilon$ and $\pstate[alt]{w} = \pstate[alt]{v}$. 
		Since $\pstate{v} \vDash \chi$ and $\pstate{v} = \pstate[alt]{v}$ on $\varset \supseteq \FV(\test{}) = \FV(\chi)$, 
		we obtain $\pstate[alt]{v} \vDash \chi$ by the simultaneous IH. 
		Hence, $\computation[alt] \in \sem{\test{}}$.
		Moreover, $\pstate{w} = \pstate{v} = \pstate[alt]{v} = \pstate[alt]{w}$ on $\varset \cup \MBV(\test{})$ because $\MBV(\test{}) = \emptyset$. 
		Finally, $\trace = \trace[alt]$ and $\pstate[alt]{w} \neq \bot$.

		\item $\gamma \equiv \send{}{}{}$ and $\computation \in \sem{\gamma}$, 
		then $(\trace, \pstate{w}) \preceq (\comItem{\historyVar, \ch{}, \sem{\rp} \pstate{v}, \statetime{\pstate{v}}}, \pstate{v})$.
		We define $\trace[alt] = \trace$, and $\pstate[alt]{w} = \bot$ if $\pstate{w} = \bot$ and $\pstate[alt]{w} = \pstate[alt]{v}$ otherwise.
		Since $\pstate{v} = \pstate[alt]{v}$ on $\varset \supseteq \FV(\gamma) = \{\historyVar, \globalTime\} \cup \FV(\rp)$,
		we have $\sem{\rp} \pstate{v} = \sem{\rp} \pstate[alt]{v}$ and $\statetime{\pstate{v}} = \statetime{\pstate[alt]{v}}$.
		Hence, $(\trace[alt], \pstate[alt]{w}) \preceq (\comItem{\historyVar, \ch{}, \sem{\rp} \pstate[alt]{v}, \statetime{\pstate[alt]{v}}}, \pstate[alt]{v})$, 
		which implies $\computation[alt] \in \sem{\gamma}$.
		Finally, if $\pstate{w} \neq \bot$,
		then $\pstate{w} = \pstate[alt]{w}$ on $\varset \cup \MBV(\gamma) = \varset \cup \{\historyVar\}$ because $\pstate{w} = \pstate{v} = \pstate[alt]{v} = \pstate[alt]{w}$ on $\varset$ 
		since $\pstate{v} = \pstate[alt]{v}$ on $\varset \subseteq \FV(\gamma) \supseteq \{\historyVar\}$.
		Moreover, $\trace = \trace[alt]$ and ($\pstate{w} = \bot$ iff $\pstate[alt]{w} = \bot$).

		\item $\computation \in \sem{\receive{}{}{}}$, then $\computation \in \sem{\alpha}$,
		where $\alpha \equiv x \ceq * \seq \send{}{}{x}$. 
		Since $\FV(\receive{}{}{}) = \FV(\alpha)$,
		we have $\pstate{v} = \pstate[alt]{v}$ on $\varset \supseteq \FV(\alpha)$.
		Hence, by IH, $\computation[alt] \in \sem{\alpha}$ exists with $\trace = \trace[alt]$ and $\pstate{w} = \pstate[alt]{w}$ on $\varset \cup \MBV()$ and ($\pstate{w} = \bot$ iff $\pstate[alt]{w} = \bot$) because we considered $\alpha$ to be structurally smaller than $\receive{}{}{}$.
		Finally, $\computation[alt] \in \sem{\receive{}{}{}}$ and observe that $\MBV(\receive{}{}{}) = \MBV(\alpha)$.
		
		\item $\computation \in \sem{\alpha_1 \cup \alpha_2}$, then $\computation \in \sem{\alpha_j}$ for some $j \in \{ 1, 2 \}$.
		Since $\varset \supseteq \FV(\alpha_1 \cup \alpha_2) \supseteq \FV(\alpha_j)$, 
		we have $\pstate{v} = \pstate[alt]{v}$ on $\varset \supseteq \FV(\alpha_j)$.
		Thus, $\computation[alt] \in \sem{\alpha_j} \subseteq \sem{\alpha_1 \cup \alpha_2}$ exists by IH such that $\pstate{w} = \pstate[alt]{w}$ on $\varset \cup \MBV(\alpha_j)$, and $\trace = \trace[alt]$, and ($\pstate{w} = \bot$ iff $\pstate[alt]{w} = \bot$). 
		Finally, $\pstate{w} = \pstate[alt]{w}$ on $\varset \cup \MBV(\alpha_1 \cup \alpha_2)$ because $\MBV(\alpha_1 \cup \alpha_2) = \MBV(\alpha_1) \cap \MBV(\beta_2)$.
		
		\item $\computation \in \sem{\alpha \seq \beta}$, then $\computation \in \sem{\alpha}_\bot$ or $\computation \in \sem{\alpha} \continuation \sem{\beta}$. 
		
		If $\computation \in \botop{\sem{\alpha}} \subseteq \sem{\alpha}$, a computation $\computation[alt] \in \sem{\alpha}$ exists by IH such that $\trace = \trace[alt]$ and $\pstate[alt]{w} = \bot$. 
		Since $\pstate{w} = \pstate[alt]{w} = \bot$, we have $\computation[alt] \in \sem{\alpha}_\bot \subseteq \sem{\alpha \seq \beta}$ with $\pstate{w} = \pstate[alt]{w}$ on $\varset \cup \MBV(\alpha \seq \beta)$, and $\trace = \trace[alt]$, and ($\pstate[pre]{w} = \bot$ iff $\pstate[alt]{w} = \bot$).
		
		If $\computation \in \sem{\alpha} \continuation \sem{\beta}$, computations $(\pstate{v}, \trace_1, \pstate{u}) \in \sem{\alpha}$ and $(\pstate{u}, \trace_2, \pstate{w}) \in \sem{\beta}$ with $\trace = \trace_1 \cdot \trace_2$ exist.
		Since $\varset \supseteq \FV(\alpha \seq \beta) \supseteq \FV(\alpha)$, by IH, $(\pstate[alt]{v}, \trace[alt]_1, \pstate[alt]{u}) \in \sem{\alpha}$ exists with $\pstate{u} = \pstate[alt]{u}$ on $\varset \cup \MBV(\alpha)$, and $\trace_1 = \trace[alt]_1$, and $\pstate[alt]{u} \neq \bot$ because $u \neq \bot$. 
		Observe that $\varset \cup \MBV(\alpha) \supseteq \FV(\alpha \seq \beta) \cup \MBV(\alpha) = \FV(\alpha) \cup (\FV(\beta) \setminus \MBV(\alpha)) \cup \MBV(\alpha) \supseteq \FV(\beta)$.
		Therefore, $(\pstate[alt]{u}, \trace[alt]_2, \pstate[alt]{w}) \in \sem{\beta}$ exists by IH again with $\pstate{w} = \pstate[alt]{w}$ on $\varset \cup \MBV(\alpha) \cup \MBV(\beta) = \varset \cup \MBV(\alpha \seq \beta)$, $\trace_2 = \trace[alt]_2$, and ($\pstate{w} = \bot$ iff $\pstate[alt]{w} = \bot$). 
		Reassembling gives us $\computation[alt] \in \sem{\alpha} \continuation \sem{\beta} \subseteq \sem{\alpha \seq \beta}$ with $\trace[alt] = \trace[alt]_1 \cdot \trace[alt]_2$. 
		Finally, $\pstate{w} = \pstate[alt]{w}$ on $\varset \cup \MBV(\alpha \seq \beta)$, and $\trace = \trace[alt]$, and ($\pstate{w} = \bot$ iff $\pstate[alt]{w} = \bot$).
		
		\item $\computation \in \sem{\repetition{\alpha}}$, then $\computation \in \sem{\alpha^n}$ for some $n \in \naturals$. 
		By induction on $n$, we show that $\computation[alt] \in \sem{\alpha^n}$ exists with $\pstate{w} = \pstate[alt]{w}$ on $\varset$, and $\trace = \trace[alt]$, and ($\pstate{w} = \bot$ iff $\pstate[alt]{w} = \bot$). 
		Finally, we conclude by $\sem{\alpha^n} \subseteq \sem{\repetition{\alpha}}$ and $\MBV(\repetition{\alpha}) = \emptyset$.
		
		\begin{enumerate}
			\item $n = 0$: 
			$\computation \in \sem{\alpha^0} = \sem{\test{\true}}$, then $\trace = \epsilon$, and $\pstate{w} = \bot$ or $\pstate{w} = \pstate{v}$. 
			We define $\trace[alt] = \epsilon$, and $\pstate[alt]{w} = \bot$ if $\pstate{w} = \bot$ and $\pstate[alt]{w} = \pstate[alt]{v}$ otherwise. 
			Hence, $\computation[alt] \in \sem{\alpha^0}$ and $\pstate{w} = \pstate[alt]{w}$ on $\varset$ because if $\pstate{w} \neq \bot$, then $\pstate{v} = \pstate[alt]{v}$ on $\varset$ implies $\pstate{w} = \pstate{v} = \pstate[alt]{v} = \pstate[alt]{w}$ on $\varset$. 
			Moreover, $\trace = \trace[alt]$, and ($\pstate{w} = \bot$ iff $\pstate[alt]{w} = \bot$).
			
			\item $n > 0$: 
			$\computation \in \sem{\alpha^n}$, then $\computation \in \sem{\alpha \seq \alpha^{n-1}} = \botop{\sem{\alpha}} \cup (\sem{\alpha} \continuation \sem{\alpha^{n-1}})$. 
			If $\computation \in \botop{\sem{\alpha}} \subseteq \sem{\alpha}$, by IH, $\computation[alt] \in \sem{\alpha}$ exists with $\trace = \trace[alt]$ and $\pstate[alt]{w} = \bot$. 
			That is, $\computation[alt] \in \botop{\sem{\alpha}} \subseteq \sem{\alpha^n}$ with $\pstate{w} = \pstate[alt]{w}$ on $\varset$ since $\bot = \bot$ on any set of variables, $\trace = \trace[alt]$, and ($\pstate{w} = \bot$ iff $\pstate[alt]{w} = \bot$).
			
			If $\computation \in \sem{\alpha} \continuation \sem{\alpha^{n-1}}$, computations $(\pstate{v}, \trace_1, \pstate{u}) \in \sem{\alpha}$ and $(\pstate{u}, \trace_2, \pstate{w}) \in \sem{\alpha^{n-1}}$ exist with $\trace = \trace_1 \cdot \trace_2$. 
			Since $\varset \supseteq \FV(\repetition{\alpha}) = \FV(\alpha)$, by IH, $(\pstate[alt]{v}, \trace[alt]_1, \pstate[alt]{u}) \in \sem{\alpha}$ exists with $\pstate{u} = \pstate[alt]{u}$ on $\varset$, and $\trace_1 = \trace[alt]_1$, and $\pstate[alt]{u} \neq \bot$. 
			Since $\varset \supseteq \FV(\repetition{\alpha}) = \FV(\alpha^{n-1})$ and $\alpha^{n-1}$ is structurally smaller than~$\repetition{\alpha}$, 
			by IH again, $(\pstate[alt]{u}, \trace[alt]_2, \pstate[alt]{w}) \in \sem{\alpha^{n-1}}$ with $\pstate{w} = \pstate[alt]{w}$ on $\varset$, and $\trace_2 = \trace[alt]_2$, and ($\pstate{w} = \bot$ iff $\pstate[alt]{w} = \bot$).
			Reassembling gives us $\computation[alt] \in \sem{\alpha} \continuation \sem{\alpha^{n-1}} \subseteq \sem{\alpha^n}$ with $\trace[alt] = \trace[alt]_1 \cdot \trace[alt]_2$. 
			Finally, $\pstate{w} = \pstate[alt]{w}$ on $\varset$, and $\trace = \trace[alt]$, and ($\pstate{w} = \bot$ iff $\pstate[alt]{w} = \bot$).
		\end{enumerate}
		
		\item $\computation \in \sem{\alpha_1 \parOp \alpha_2}$, then $\computation[proj={\alpha_j}] \in \sem{\alpha}$ exists for $j = 1, 2$ such that $\statetime{\pstate{w}_{\alpha_1}} = \statetime{\pstate{w}_{\alpha_2}}$, and $\trace = \trace \downarrow (\alpha_1 \parOp \alpha_2)$ and $\pstate{w} = \pstate{w}_{\alpha_1} \merge \pstate{w}_{\alpha_2}$.
		Since $\varset \supseteq \FV(\alpha_1 \parOp \alpha_2) \supseteq \FV(\alpha_j)$, 
		by IH, $(\pstate[alt]{v}, \trace[alt]_{\alpha_j}, \pstate[alt]{w}_{\alpha_j}) \in \sem{\alpha_j}$ exists with $\pstate{w}_{\alpha_j} = \pstate[alt]{w}_{\alpha_j}$ on $\varset \cup \MBV(\alpha_j)$, 
		and $\trace \downarrow \alpha = \trace[alt]_{\alpha_j}$, 
		and ($\pstate{w}_{\alpha_j} = \bot$ iff $\pstate[alt]{w}_{\alpha_j} = \bot$).
		We define $\trace[alt] = \trace$. 
		Thus, $\trace[alt] \downarrow \alpha_j = \trace[alt]_{\alpha_j}$, and $\trace[alt] = \trace[alt] \downarrow (\alpha_1 \parOp \alpha_2)$. 
		Moreover, $\statetime{\pstate[alt]{w}_{\alpha_1}} = \statetime{\pstate[alt]{w}_{\alpha_2}}$ either because $\pstate[alt]{w}_{\alpha_1} = \pstate[alt]{v} = \pstate[alt]{w}_{\alpha_2}$ on $\globalTime$ by \rref{lem:boundEffect} if $\globalTime$ is not bound in both programs, 
		or because $\globalTime \in \varset$ since $\globalTime \in \BV(\alpha_1 \parOp \alpha_2)$ implies  $\globalTime \in \FV(\alpha_1 \parOp \alpha_2)$ such that $\pstate[alt, proj={\alpha_1}]{w} = \pstate[proj={\alpha_1}]{w} = \pstate[proj={\alpha_2}]{w} = \pstate[alt, proj={\alpha_2}]{w}$ on $\globalTime$.
		In summary, $\computation[alt] \in \sem{\alpha_1 \parOp \alpha_2}$,
		where $\pstate[alt]{w} = \pstate[alt]{w}_{\alpha_1} \merge \pstate[alt]{w}_{\alpha_2}$, 
		such that $\trace = \trace[alt]$, and ($\pstate{w} = \bot$ iff $\pstate[alt]{w} = \bot$).
		
		Finally, we show $\pstate{w} = \pstate[alt]{w}$ on $\varset \cup \MBV(\alpha_1 \parOp \alpha_2)$.
		Let $\arbitraryVar \in \varset \cup \MBV(\alpha_1 \parOp \alpha_2) = \varset \cup \MBV(\alpha_1) \cup \MBV(\alpha_2)$.
		If $\arbitraryVar \in \BV(\alpha_1)$, then $\arbitraryVar \in \varset \cup \MBV(\alpha_1) \cup \{\globalTime\} \cup \TVar$ because $\V(\alpha_2) \cap \BV(\alpha_1) \subseteq \{\globalTime\} \cup \TVar$ (well-formedness) and $\MBV(\alpha_2) \subseteq \V(\alpha_2)$.
		Moreover, $\BV(\alpha_1) \cap (\{\globalTime\} \cup \TVar)\subseteq \MBV(\alpha_1)$.
		Now, $\pstate{w}(\arbitraryVar) = \pstate{w}_{\alpha_1}(\arbitraryVar) = \pstate[alt]{w}_{\alpha_1}(\arbitraryVar) = \pstate[alt]{w}(\arbitraryVar)$ since $\pstate{w}_{\alpha_1} = \pstate[alt]{w}_{\alpha_1}$ on $\varset \cup \MBV(\alpha)$, and moreover, $\pstate{w} = \pstate{w}_{\alpha_1}$ on $\BV(\alpha_1)$ and $\pstate[alt]{w} = \pstate[alt]{w}_{\alpha_1}$ on $\BV(\alpha_1)$ by the definition of $\merge$ in \rref{sec:semantics}.
		If $\arbitraryVar \not \in \BV(\alpha_1)$, then $\arbitraryVar \in \varset \cup \MBV(\alpha_2)$. 
		We conclude by $\pstate{w}(\arbitraryVar) = \pstate{w}_{\alpha_1}(\arbitraryVar) = \pstate[alt]{w}_{\alpha_2}(\arbitraryVar) = \pstate[alt]{w}(\arbitraryVar)$ because $\pstate{w} = \pstate{w}_{\alpha_2}$ on $\BV(\alpha_1)^\complement$ and $\pstate[alt]{w} = \pstate[alt]{w}_{\alpha_2}$ on $\BV(\alpha_1)^\complement$ by the definition of $\merge$ again, and $\pstate{w}_{\alpha_2} = \pstate[alt]{w}_{\alpha_2}$ on $\varset \cup \MBV(\alpha_2)$.
		\qedhere
	\end{enumerate}
\end{proof}

\section{Details of the Calculus} \label{app:calculus}

Reasoning on trace-terms is by simple algebraic laws (see \rref{fig:traceAlgebra}). 
Their soundness should be clear from the semantics of terms (see \rref{def:termSemantics}).

\renewcommand{\rulesep}{\\[.4em]}
\begin{figure}[ht]
	\begin{minipage}{\textwidth}
		\begin{calculus}
			\startAxiom{concatDist}
				$(\te_1 \cdot \te_2) \downarrow \cset = \te_1 \downarrow \cset \cdot \te_2 \downarrow \cset$
			\stopAxiom
			\startAxiom{projIn}
				$\comItem{\ch{}, \rp_1, \rp_2} \downarrow \cset = \comItem{\ch{}, \rp_1, \rp_2} \sidecondition{$\ch{} \in \cset$}$
			\stopAxiom
			\startAxiom{projCut}
				$(\te \downarrow \cset') \downarrow \cset = \te \downarrow (\cset' \cap \cset)$
			\stopAxiom
			\startAxiom{valAccessBase}
				$\len{\te} = \ie \rightarrow \val{\at{(\te \cdot \te_0)}{\ie}} = \rp_1$
			\stopAxiom
			\startAxiom{timeAccessBase}
				$\len{\te} = \ie \rightarrow \stamp{\at{(\te \cdot \te_0)}{\ie}} = \rp_2$
			\stopAxiom
			\startAxiom{chanAccessBase}
				$\len{\te} = \ie \rightarrow \chan{\at{(\te \cdot \te_0)}{\ie}} = \ch{}$
			\stopAxiom
		\end{calculus}
		\begin{calculus}
			\startAxiom{concatAssoc}
				$(\te_1 \cdot \te_2) \cdot \te_3 = \te_1 \cdot (\te_2 \cdot \te_3)$
			\stopAxiom
			\startAxiom{concatNeutral}
				$\te \cdot \epsilon = \te = \epsilon \cdot \te$
			\stopAxiom
			\startAxiom{projNeutral}
				$\epsilon \downarrow \cset = \epsilon$
			\stopAxiom
			\startAxiom{projNotIn}
				$\comItem{\ch{}, \rp_1, \rp_2} \downarrow \cset = \epsilon \sidecondition{$\ch{} \not \in \cset$}$
			\stopAxiom
			\startAxiom{nonNegative}
				$\len{\te} \ge 0$
			\stopAxiom
			\startAxiom{unroll}
				$\len{\te \cdot \te_0} = \len{\te} + 1$
			\stopAxiom
		\end{calculus}
		
		\vspace{.5em}
		\begin{calculus}
			\startAxiom{valAccessInd}
				$\len{\te} > \ie \rightarrow \val{\at{(\te \cdot \te_0)}{\ie}} = \val{\at{\te}{\ie}}$
			\stopAxiom
			\startAxiom{timeAccessInd}
				$\len{\te} > \ie \rightarrow \stamp{\at{(\te \cdot \te_0)}{\ie}} = \stamp{\at{\te}{\ie}}$
			\stopAxiom
			\startAxiom{chanAccessInd}
				$\len{\te} > \ie \rightarrow \chan{\at{(\te \cdot \te_0)}{\ie}} = \chan{\at{\te}{\ie}}$
			\stopAxiom
		\end{calculus}
	\end{minipage}
	\caption{Algebra of traces \cite{Zwiers_Phd}, where $\te_0 \equiv \comItem{\ch{}, \rp_1, \rp_2}$.}
	\label{fig:traceAlgebra}
\end{figure}

In preparation for the soundness proof (see \rref{thm:soundness}),
\rref{lem:ac_invariance} shows that $\A$ and $\C$ do not depend on the state of $\alpha$ if $[ \alpha ] \ac \psi$ is well-formed.
Moreover, \rref{lem:noninterference} exploits \rref{def:noninterference} showing that a program $\beta$ not interfering with $[ \alpha ] \ac \psi$ has no influence on the validity of $\A$, $\C$, and $\psi$ in the parallel composition $[ \alpha \parOp \beta ] \ac \psi$.

\begin{lemma} \label{lem:ac_invariance}
	Let $[ \alpha ] \ac \psi$ be well-formed and $\computation \in \sem{\alpha}$ with $\pstate{w} \neq \bot$.
	Then $\pstate{v} = \pstate{w}$ on $\FV(\A) \cup \FV(\C)$.
	Moreover, $\pstate{v} \cdot \trace[alt] \vDash \chi$ iff $\pstate{w} \cdot \trace[alt] \vDash \chi$ for $\chi \in \{ \A, \C \}$ and arbitrary $\trace[alt] \in \recTraces$.
\end{lemma}
\begin{proof}
	By the bound effect property (\rref{lem:boundEffect}), $\pstate{v} = \pstate{w}$ on $\BV(\alpha)^\complement \cup \TVar$.
	Since $[ \alpha ] \ac \psi$ is well-formed, we have $(\FV(\A) \cup \FV(\C)) \cap \BV(\alpha) \subseteq \TVar$, which implies $\pstate{v} = \pstate{w}$ on $\FV(\A) \cup \FV(\C)$.
	By coincidence (\rref{lem:formulaCoincidence}), $\pstate{v} \cdot \trace[alt] \vDash \chi$ iff $\pstate{w} \cdot \trace[alt] \vDash \chi$.
	\qedhere
\end{proof}

\begin{lemma}[Noninterference prevents invalidation] \label{lem:noninterference}
	Let $\beta \in \Chp$ be a program, which does not interfere with $[ \alpha ] \ac \psi$ (\rref{def:noninterference}). 
	Moreover, let $\computation \in \sem{\alpha \parOp \beta}$, 
	\iest among others $(\pstate{v}, \trace \downarrow \alpha, \pstate{w}_\alpha) \in \sem{\alpha}$ and $(\pstate{v}, \trace \downarrow \beta, \pstate{w}_\beta) \in \sem{\beta}$ with $\pstate{w} = \pstate{w}_\alpha \merge \pstate{w}_\beta$.
	Then for $\chi \in \{ \A, \C \}$, the following holds:
	\begin{enumerate}
		\item \label{itm:invalidation1}
		$\pstate{v} \cdot (\trace \downarrow \alpha) \vDash \chi$
		\;iff\;
		$\pstate{v} \cdot \trace \vDash \chi$
		
		\item \label{itm:invalidation2} 
		$\pstate{w} \neq \bot$ \;implies\; $\big( \pstate{w}_\alpha \cdot (\trace \downarrow \alpha) \vDash \psi$
		\;iff\;
		$\pstate{w} \cdot \trace \vDash \psi \big)$
	\end{enumerate}
\end{lemma}
\begin{proof}
	First, we show $(\trace \downarrow \alpha) \downarrow \lambda = \trace \downarrow \lambda$ for $\lambda \in \{ \A, \C, \psi \}$, \iest that $\lambda$ only depends on $\trace \downarrow \alpha$.
	This holds if only a communication event $\rawtrace \in \{ \comItem{\ch{}, a, \duration}, \comItem{\historyVar, \ch{}, a, \duration} \}$ in $\trace$,
	which is not removed by $\downarrow \lambda$ is also not removed by $\downarrow \alpha$.
	So let $\rawtrace \downarrow \lambda = \rawtrace$.
	Then $\ch{} \in \CN(\lambda)$.
	If $\ch{} \not \in \CN(\beta)$,
	then $\ch{} \in \CN(\alpha)$ because $\ch{} \in \CN(\alpha) \cup \CN(\beta)$ as $\rawtrace$ is emitted by $\alpha \parOp \beta$.
	Otherwise, if $\ch{} \in \CN(\beta)$,
	then $\ch{} \in \CN(\alpha)$ by noninterference (\rref{def:noninterference}).
	Hence, $\ch{} \in \CN(\alpha)$ such that $\rawtrace$ is not removed by $\downarrow \alpha$.
	
	Since $(\trace \downarrow \alpha) \downarrow \chi = \trace \downarrow \chi$, we have 
	\begin{equation} \label{eq:historyProjs}
		\big( \pstate{v} \cdot (\trace \downarrow \alpha) \big) \downarrow \chi
		= (\pstate{v} \downarrow \chi) \cdot \big( (\trace \downarrow \alpha) \downarrow \chi \big)
		= (\pstate{v} \downarrow \chi) \cdot (\trace \downarrow \chi)
		= (\pstate{v} \cdot \trace) \downarrow \chi
	\end{equation}
	on $\V \supseteq \FV(\chi)$.
	Now, \rref{itm:invalidation1} follows by coincidence (\rref{lem:formulaCoincidence}).
	
	For \rref{itm:invalidation2}, assume $\pstate{w} \neq \bot$. 
	Then $\pstate{w}_\alpha \neq \bot$ and $\pstate{w}_\beta \neq \bot$ by the definition of $\merge$ in \rref{sec:semantics}. 
	First, observe that $\pstate{w}_\alpha = \pstate{w}$ on $\BV(\alpha)$ by the definition of $\merge$, thus $\pstate{w}_\alpha = \pstate{w}$ on $\BV(\alpha) \cap \BV(\beta)^\complement$.
	Second, $\pstate{w}_\alpha = \pstate{v}$ on $\BV(\alpha)^\complement \cup \TVar$ by the bound effect property (\rref{lem:boundEffect}), 
	$\pstate{v} = \pstate{w}_\beta$ on $\BV(\beta)^\complement \cup \TVar$ also by the bound effect property, 
	and $\pstate{w}_\beta = \pstate{w}$ on $\BV(\alpha)^\complement$ by the definition of $\merge$.
	Chaining the equalities of the last sentence gives us $\pstate{w}_\alpha = \pstate{w}$ on
	\begin{equation*}
		(\BV(\alpha)^\complement \cup \TVar) \cap (\BV(\beta)^\complement \cup \TVar) \cap \BV(\alpha)^\complement \supseteq \BV(\alpha)^\complement \cap \BV(\beta)^\complement \text{.}
	\end{equation*}
	Third, $\pstate{w}_\alpha = \pstate{v} = \pstate{w}_\beta$ on $\TVar$ by the bound effect property such that $\pstate{w}_\alpha = \pstate{w}_\alpha \merge \pstate{w}_\beta = \pstate{w}$ on $\TVar$.
	Summarizing, where $\pstate{w}_\alpha$ equals $\pstate{w}$, 
	we obtain $\pstate{w}_\alpha = \pstate{w}$ on
	\begin{equation*}
		(\BV(\alpha) \cap \BV(\beta)^\complement) \cup (\BV(\alpha)^\complement \cap \BV(\beta)^\complement) \cup \TVar = \BV(\beta)^\complement \cup \TVar \text{.}
	\end{equation*} 
	Thus, $\pstate{w}_\alpha = \pstate{w}$ on $\FV(\psi)$ because $\beta$ does not interfere with $[ \alpha ] \ac \psi$ (\rref{def:noninterference}), which implies $\FV(\psi) \subseteq \BV(\beta)^\complement \cup \TVar$.
	Therefore, $(\pstate[ind=\alpha]{w} \cdot (\trace \downarrow \alpha)) \downarrow \chi = (\pstate{w} \cdot \trace) \downarrow \chi$ on $\FV(\psi)$ since $(\pstate[ind=\alpha]{w} \cdot (\trace \downarrow \alpha)) \downarrow \chi = (\pstate{w}_\alpha \cdot \trace) \downarrow \chi$ holds analogously to \rref{eq:historyProjs}.
	Finally, \rref{itm:invalidation2} holds by coincidence.
	\qedhere
\end{proof}

\begin{proposition} \label{prop:parAssocComm}
	The parallel operator $\parOp$ is associative and commutative.
\end{proposition}
\begin{proof}[\rref{prop:parAssocComm}]
	We prove that the parallel operator $\parOp$ is associative and commutative:
	\begin{enumerate}
		\item Let $\alpha\beta$ be short for $\alpha \parOp \beta$.
		Now, $\computation \in \sem{\alpha \parOp \beta\gamma}$
		iff $\computation[proj=\pi] \in \sem{\pi}$ for $\pi \in \{ \alpha, \beta\gamma \}$,
		and $\statetime{\pstate{w}_\alpha} = \statetime{\pstate{w}_{\beta\gamma}}$,
		and $\trace \downarrow (\alpha \parOp \beta\gamma) = \trace$,
		and $\pstate{w} = \pstate{w}_\alpha \merge \pstate{w}_{\beta\gamma}$.
		Moreover, $\computation[proj={\beta\gamma}] \in \sem{\beta\gamma}$ iff $\computation[proj=\pi] \in \sem{\pi}$ since $(\trace \downarrow \beta\gamma) \downarrow \pi = \trace \downarrow \pi$ for $\pi \in \{ \beta, \gamma \}$,
		and $(\trace \downarrow \beta\gamma) \downarrow \beta\gamma = \trace \downarrow \beta\gamma$,
		and $\statetime{\pstate{w}_\beta} = \statetime{\pstate{w}_\gamma}$,
		and $\pstate{w}_{\beta\gamma} = \pstate{w}_\beta \merge \pstate{w}_\gamma$.

		Likewise, $\computation \in \sem{\alpha\beta \parOp \gamma}$ iff $\computation[proj=\pi] \in \sem{\pi}$ for $\pi \in \{ \alpha, \beta, \gamma \}$,
		and $\statetime{\pstate{w}_{\alpha\beta}} = \statetime{\pstate{w}_\gamma}$,
		and $\statetime{\pstate{w}_\alpha} = \statetime{\pstate{w}_\beta}$,
		and $\trace \downarrow (\alpha\beta \parOp \gamma) = \trace$,
		and $(\trace \downarrow \alpha\beta) \downarrow \alpha\beta = \trace \downarrow \alpha\beta$,
		and $\pstate{w} = \pstate{w}_{\alpha\beta} \merge \pstate{w}_\gamma$,
		and $\pstate{w}_{\alpha\beta} = \pstate{w}_\alpha \merge \pstate{w}_\beta$.

		Both decompositions can be reassembled into the other one because $\trace \downarrow \cset$ is chronological for any~$\cset$ as $\trace$ is chronological.
		Moreover, all $\statetime{\pstate{w}_\pi}$ with $\pi \in \{ \alpha, \beta, \gamma \}$ equate since $\statetime{\pstate{w}_{\beta\gamma}} = \statetime{\pstate{w}_\beta}$ and $\statetime{\pstate{w}_{\alpha\beta}} = \statetime{\pstate{w}_\alpha}$.
		Finally, $\pstate{w}_\alpha \merge (\pstate{w}_\beta \merge \pstate{w}_\gamma) = (\pstate{w}_\alpha \merge \pstate{w}_\beta) \merge \pstate{w}_\gamma$:

		If $\pstate{w}_\pi = \bot$ for some $\pi \in \{ \alpha, \beta, \gamma \}$,
		then both sides equal $\bot$.
		Otherwise, if $\pstate{w}_\pi \neq \bot$ for all $\pi \in \{ \alpha, \beta, \gamma \}$, 
		the equation holds 
		because each variable is written by at most one of $\alpha$, $\beta$, or $\gamma$.

		\item Let $\computation \in \sem{\alpha \parOp \beta}$.
		Then $\computation[proj=\alpha] \in \sem{\alpha}$, 
		and $\computation[proj=\beta] \in \sem{\beta}$, 
		and $\trace \downarrow (\alpha \parOp \beta) = \trace$, 
		and $\statetime{\pstate{w}_\alpha} = \statetime{\pstate{w}_\beta}$, and $\pstate{w} = \pstate{w}_\alpha \merge \pstate{w}_\beta$.
		First, $\trace \downarrow (\alpha \parOp \beta) = \trace \downarrow (\beta \parOp \alpha)$ since $\CN(\alpha \parOp \beta) = \CN(\alpha) \cup \CN(\beta)$.
		By the bound effect property (\rref{lem:boundEffect}), $\pstate{w}_\alpha \merge \pstate{w}_\beta = \pstate{v} = \pstate{w}_\beta \merge \pstate{w}_\alpha$ on $(\BV(\alpha) \cup \BV(\beta))^\complement$. 
		Moreover, $\pstate{w}_\alpha \merge \pstate{w}_\beta = \pstate{w}_\alpha = \pstate{w}_\beta \merge \pstate{w}_\alpha$ on $\BV(\alpha)$ since $\BV(\alpha) \subseteq \BV(\beta)^\complement$ because $\alpha$ and $\beta$ do not share state.
		Accordingly, $\pstate{w}_\alpha \merge \pstate{w}_\beta = \pstate{w}_\beta \merge \pstate{w}_\alpha$ on $\BV(\beta)$ such that overall $\pstate{w}_\alpha \merge \pstate{w}_\beta = \pstate{w}_\beta \merge \pstate{w}_\alpha$.
		Finally, $\computation \in \sem{\beta \parOp \alpha}$.
		\qedhere
	\end{enumerate}
\end{proof}

\begin{proof}[\rref{thm:soundness}]
	We prove soundness of the novel ac-axioms and rules. 
	Since \dLCHP is a conservative extension of \dL (\rref{prop:conservative}), we can soundly use the \dL proof calculus for reasoning about \dLCHP formulas. 
	Hence, we point to the literature for soundness of the axioms and rules adopted from~$\dL$ \cite{DBLP:journals/jar/Platzer08, DBLP:journals/jar/Platzer17, Platzer18}. 
	
	\begin{itemize}[leftmargin=1em, itemindent=-1em]
		\itemsep.5em
		\item[] \RuleName{acNoCom}:
		Let $\pstate{v} \vDash [ \alpha ] \ac \psi$. 
		Then $\pstate{v} \vDash \C$ by axiom \RuleName{acWeak}.
		Now, let $\computation \in \sem{\alpha}$, 
		and assume $\pstate{w} \neq \bot$ and $\pstate{v} \vDash \A$. 
		Then $\trace = \epsilon$ because $\CN(\alpha) = \emptyset$. 
		Therefore, the precondition $\preeq{\pstate{v}}{\trace} \vDash \A$ of \acPost is fulfilled, and we conclude $\pstate{w} \vDash \psi$. 
		Conversely, let $\pstate{v} \vDash \C \wedge (\A \rightarrow [ \alpha ] \psi)$ and $\computation \in \sem{\alpha}$.
		Condition \acCommit holds since $\trace = \epsilon$ and $\pstate{v} \vDash \C$ by precondition. 
		For \acPost, assume $\pstate{w} \neq \bot$ and $\preeq{\pstate{v}}{\trace} \vDash \A$, which contains $\pstate{v} \vDash \A$.
		Then $\pstate{v} \vDash \A \rightarrow [ \alpha ] \psi$ implies $\pstate{w} \vDash \psi$.
		\item[] \RuleName{acWeak}:
		Let $\pstate{v} \vDash [ \alpha ] \ac \psi$.
		Then $\pstate{v} \vDash \C$ by \acCommit since $(v, \epsilon, \bot) \in \sem{\alpha}$ by  totality (\rref{prop:prefixClosedAndTotal}) and $\proppre{\pstate{v}}{\epsilon} \vDash \A$ trivially holds as $\proppre{\pstate{v}}{\epsilon} = \emptyset$.
		For $[ \alpha ] \ac ( \C \wedge ( \A \rightarrow \psi ) )$, \acCommit holds by assumption. 
		For \acPost, let $\computation \in \sem{\alpha}$, and assume $\pstate{w} \neq \bot$ and $\preeq{\pstate{v}}{\trace} \vDash \A$. 
		By \acCommit, $\pstate{v} \cdot \trace \vDash \C$ such that $\pstate{w} \cdot \trace \vDash \C$ by \rref{lem:ac_invariance}.
		From \acPost, we obtain $\pstate{w} \cdot \trace \vDash \psi$,
		which implies $\pstate{w} \cdot \trace \vDash \A \rightarrow \psi$. 
		The converse direction derives by rule \RuleName{acMono}.
		\item[] \RuleName{K}:
		Let $\pstate{v} \vDash \universal \K$, where $\K \equiv \Kexpanded$ and $\universal \K$ denotes the universal closure of $\K$.
		Moreover, let $\pstate{v} \vDash \Phi$ for $\Phi \equiv [ \alpha ] \acpair{\A_1 \wedge \A_2, \C_1 \wedge \C_2} \psi$ and $\computation \in \sem{\alpha}$.
		By induction on the length of $\trace$,
		we simultaneously prove $\pstate{v} \vDash [ \alpha ] \acpair{\A, \C_1 \wedge \C_2} \psi$, and $\proppre{\pstate{v}}{\trace} \vDash \A$ implies $\proppre{\pstate{v}}{\trace} \vDash \A_1 \wedge \A_2$:
		
		\begin{enumerate}
			\item $\semLen{\trace} = 0$, 
			then $\pstate{v} \vDash \Phi$ implies $\pstate{v} \vDash \C_1 \wedge \C_2$ by axiom \RuleName{acWeak}.
			Hence, \acCommit holds since $\trace = \epsilon$.
			For \acPost, assume $\pstate{w} \neq \bot$ and $\preeq{\pstate{v}}{\trace} \vDash \A$. 
			Then $\preeq{\pstate{v}}{\trace} \vDash \A_1 \wedge \A_2$ since $\pstate{v} \vDash \C_1 \wedge \C_2$ holds, $\K$ is valid, and $\trace = \epsilon$.
			Thus, $\pstate{w} \cdot \trace \vDash \psi$ by assumption. 
			Note that $\proppre{\pstate{v}}{\trace} \vDash \A_1 \wedge \A_2$ is trivially fulfilled since $\proppre{\pstate{v}}{\trace} = \emptyset$.
			
			\item $\semLen{\trace} > 0$,
			then $\trace = \trace_0 \cdot \rawtrace$ with $\semLen{\rawtrace} = 1$.
			For \acCommit, assume $\proppre{\pstate{v}}{\trace} \vDash \A$.
			Then $\preeq{\pstate{v}}{\trace_0} \vDash \A$, which implies $\proppre{\pstate{v}}{\trace_0} \vDash \A_1 \wedge \A_2$ by IH.
			Since $(v, \trace_0, \bot) \in \sem{\alpha}$ by prefix-closedness (\rref{prop:prefixClosedAndTotal}),
			we obtain $\pstate{v} \cdot \trace_0 \vDash \C_1 \wedge \C_2$ by $\pstate{v} \vDash \Phi$.
			Hence, $\pstate{v} \cdot \trace_0 \vDash \A_1 \wedge \A_2$ as $\K$ is valid and $\pstate{v} \cdot \trace_0 \vDash \A$ such that $\proppre{\pstate{v}}{\trace} \vDash \A_1 \wedge \A_2$. 
			Finally, $\pstate{v} \cdot \trace \vDash \C_1 \wedge \C_2$ by $\pstate{v} \vDash \Phi$ again. 
			
			For \acPost, assume $\pstate{w} \neq \bot$ and $\preeq{\pstate{v}}{\trace} \vDash \A$. 
			Then $\proppre{\pstate{v}}{\trace} \vDash \A$, which implies $\proppre{\pstate{v}}{\trace} \vDash \A_1 \wedge \A_2$ and $\pstate{v} \cdot \trace \vDash \C_1 \wedge \C_2$ as in case \acCommit. 
			By validity of $\K$ and $\pstate{v} \cdot \trace \vDash \A$, we obtain $\pstate{v} \cdot \trace \vDash \A_1 \wedge \A_2$. 
			In summary, $\preeq{\pstate{v}}{\trace} \vDash \A_1 \wedge \A_2$, which implies $\pstate{w} \cdot \trace \vDash \psi$ by $\pstate{v} \vDash \Phi$.
		\end{enumerate}
		\item[] \RuleName{acComposition}
		Let $\pstate{v} \vDash [ \alpha \seq \beta ] \ac \psi$. 
		To show $\pstate{v} \vDash [ \alpha ] \ac [\beta] \ac \psi$, 
		let $(v, \trace_1, \pstate{u}) \in \sem{\alpha}$.
		For \acCommit, assume $\proppre{\pstate{v}}{\trace_1} \vDash \A$ and observe that $(v, \trace_1, \bot) \in \botop{\sem{\alpha}} \subseteq \sem{\alpha \seq \beta}$.
		Then $\pstate{v} \cdot \trace_1 \vDash \C$ by assumption. 
		For \acPost, assume $\pstate{u} \neq \bot$ and $\preeq{\pstate{v}}{\trace_1} \vDash \A$. 
		To show $\pstate{u} \cdot \trace_1 \vDash [ \beta ] \ac \psi$, let $(\pstate{u}, \trace_2, \pstate{w}) \in \sem{\beta}$.
		
		\begin{enumerate}
			\item For \acCommit, assume $\proppre{\pstate{u} \cdot \trace_1}{\trace_2} \vDash \A$.
			Then $\proppre{\pstate{v} \cdot \trace_1}{\trace_2} \vDash \A$ by \rref{lem:ac_invariance}, which implies $\proppre{\pstate{v}}{\trace_1 \cdot \trace_2} \vDash \A$ by assumption $\preeq{\pstate{v}}{\trace_1} \vDash \A$.
			Then $\pstate{v} \cdot \trace_1 \cdot \trace_2 \vDash \C$ because $(v, \trace_1 \cdot \trace_2, \pstate{w}) \in \sem{\alpha} \continuation \sem{\beta} \subseteq \sem{\alpha \seq \beta}$. 
			Finally, $\pstate{u} \cdot \trace_1 \cdot \trace_2 \vDash \C$, using \rref{lem:ac_invariance} again.
			
			\item For \acPost, assume $\pstate{w} \neq \bot$ and $\preeq{\pstate{u} \cdot \trace_1}{\trace_2} \vDash \A$. 
			Then $\preeq{\pstate{v}}{\trace_1 \cdot \trace_2} \vDash \A$ by \rref{lem:ac_invariance}. 
			Since $(v, \trace_1 \cdot \trace_2, \pstate{w}) \in \sem{\alpha \seq \beta}$, we obtain $\pstate{w} \cdot \trace_1 \cdot \trace_2 \vDash \psi$ by assumption.
		\end{enumerate}
		
		Conversely, let $\pstate{v} \vDash [ \alpha ] \ac [ \beta ] \ac \psi$ and $\computation \in \sem{\alpha \seq \beta}$.
		If $\computation \in \botop{\sem{\alpha}}$, \acCommit holds by the assumption. 
		Since $\pstate{w} = \bot$, \acPost holds trivially.
		Otherwise, if $\computation \in \sem{\alpha} \continuation \sem{\beta}$,
		computations $(\pstate{v}, \trace_1, \pstate{u}) \in \sem{\alpha}$ and $(\pstate{u}, \trace_2, \pstate{w}) \in \sem{\beta}$ with $\trace = \trace_1 \cdot \trace_2$ exist. 
		Now, we conclude as follows:
		
		\begin{enumerate}
			\item 
			For \acCommit, assume $\proppre{\pstate{v}}{\trace} \vDash \A$. 
			If $\trace_2 = \epsilon$, \acCommit by the assumption because $(\pstate{v}, \trace, \bot) \in \botop{\sem{\alpha}}$.
			If $\trace_2 \neq \epsilon$, then $\preeq{\pstate{v}}{\trace_1} \vDash \A$.
			Hence, $\pstate{u} \cdot \trace_1 \vDash [ \beta ] \ac \psi$ by \acPost. 
			By \rref{lem:ac_invariance}, $\proppre{\pstate{u}}{\trace_1 \cdot \trace_2} \vDash \A$  
			which implies $\proppre{\pstate{u} \cdot \trace_1}{\trace_2} \vDash \A$.
			By $\pstate{u} \cdot \trace_1 \vDash [ \beta ] \ac \psi$ and $(\pstate{u}, \trace_2, \pstate{w}) \in \sem{\beta}$, we obtain $\pstate{u} \cdot \trace_1 \cdot \trace_2 \vDash \C$. 
			Using \rref{lem:ac_invariance} again, $\pstate{v} \cdot \trace_1 \cdot \trace_2 \vDash \C$.
			
			\item 
			For \acPost, assume $\preeq{\pstate{v}}{\trace} \vDash \A$ and $\pstate{w} \neq \bot$. 
			Then $\pstate{u} \cdot \trace_1 \vDash [ \beta ] \ac \psi$ as above.
			Finally, $\pstate{w} \cdot \trace \vDash \psi$ because $\preeq{\pstate{u} \cdot \trace_1}{\trace_2} \vDash \A$ by \rref{lem:ac_invariance}.
		\end{enumerate}
		\item[] \RuleName{acChoice}:
		The axiom follows directly from the semantics $\sem{\alpha \cup \beta} = \sem{\alpha} \cup \sem{\beta}$ of choice. 
		\item[] \RuleName{acIteration}
		Since $\sem{\repetition{\alpha}} = \bigcup_{n \in \naturals} \sem{\alpha^n}$, the formula $[ \repetition{\alpha} ] \ac \psi \leftrightarrow [ \alpha^0 ] \ac \psi \wedge [ \alpha \seq \alpha^* ] \ac \psi$ is valid. Axiom \RuleName{acIteration} follows from the formula by axiom \RuleName{acComposition}.    
		\item[] \RuleName{acInduction}:
		Let $\pstate{v} \vDash [ \repetition{\alpha} ] \ac \psi$. 
		Then $\pstate{v} \vDash [ \alpha^0 ] \ac \psi \wedge [ \alpha \seq \repetition{\alpha} ] \ac \psi$ by axioms \RuleName{acIteration} and \RuleName{acComposition}.
		Since $\sem{\alpha \seq \repetition{\alpha}} = \sem{\repetition{\alpha} \seq \alpha}$ by induction and \rref{lem:compositionAssoc}, 
		$\pstate{v} \vDash [ \repetition{\alpha} ] \ac [ \alpha ] \ac \psi$ by axiom \RuleName{acComposition}. 
		Finally, $\pstate{v} \vDash [ \repetition{\alpha} ] \acpair{\A, \true} (\psi \rightarrow [ \alpha ] \ac \psi)$ by rule \RuleName{acMono}.
		
		Conversely, let $\pstate{v} \vDash [ \alpha^0 ] \ac \psi \wedge [ \repetition{\alpha} ] \acpair{\A, \true} ( \psi \rightarrow [ \alpha ] \ac \psi )$ and $\computation \in \sem{\repetition{\alpha}}$.
		Since $\computation \in \sem{\alpha^n}$ for some $n \in \naturals$, we proceed by induction on $n$:
		\begin{enumerate}
			\item $n = 0$,
			then \acCommit holds by $\pstate{v} \vDash [ \alpha^0 ] \ac \psi$ and axiom \RuleName{acWeak} since $\trace = \epsilon$.
			For \acPost, additionally assume $\pstate{w} \neq \bot$ and observe that $\pstate{v} = \pstate{w}$ since $\alpha^0 \equiv \test{\true}$ such that $\pstate{w} \vDash \psi$.

			\item $n > 0$,
			then $\computation \in \sem{\alpha^n} = \sem{\alpha \seq \alpha^{n-1}}$, which,
			by \rref{lem:compositionAssoc},
			equals $\sem{\alpha^{n-1} \seq \alpha} = \botop{\sem{\alpha^{n-1}}} \cup \sem{\alpha^{n-1}} \continuation \sem{\alpha}$.
			In case $\botop{\sem{\alpha^{n-1}}} \subseteq \sem{\alpha^{n-1}}$, 
			\acCommit and \acPost hold by IH.
			Otherwise, if $\computation \in \sem{\alpha^{n-1}} \continuation \sem{\alpha}$, computations $(\pstate{v}, \trace_1, \pstate{u}) \in \sem{\alpha^{n-1}}$ and $(\pstate{u}, \trace_2, \pstate{w}) \in \sem{\alpha}$ exist with $\trace = \trace_1 \cdot \trace_2$.
			\begin{enumerate}
				\item If $\trace_2 = \epsilon$, \acCommit holds by IH since $(\pstate{v}, \trace, \bot) \in \sem{\alpha^{n-1}}$.
				If $\trace_2 \neq \epsilon$, assume $\proppre{\pstate{v}}{\trace} \vDash \A$, which implies $\preeq{\pstate{v}}{\trace_1} \vDash \A$.
				Thus, $\pstate{u} \cdot \trace_1 \vDash \psi \rightarrow [ \alpha ] \ac \psi$ by assumption.
				Since $\pstate{u} \cdot \trace_1 \vDash \psi$ by IH and \acPost, we obtain $\pstate{u} \cdot \trace_1 \vDash [ \alpha ] \ac \psi$.
				By \rref{lem:ac_invariance} and assumption $\proppre{\pstate{v}}{\trace} \vDash \A$, we obtain $\proppre{\pstate{u} \cdot \trace_1}{\trace_2} \vDash \A$.
				Thus, $\pstate{u} \cdot \trace_1 \cdot \trace_2 \vDash \C$, which implies $\pstate{v} \cdot \trace \vDash \C$ by \rref{lem:ac_invariance} again.

				\item For \acPost, let $\pstate{w} \neq \bot$ and $\preeq{\pstate{v}}{\trace} \vDash \A$.
				Then $\preeq{\pstate{v}}{\trace_1} \vDash \A$ and $\preeq{\pstate{v} \cdot \trace_1}{\trace_2} \vDash \A$. 
				As in case \acCommit, we obtain $\pstate{u} \cdot \trace_1 \vDash [ \alpha ] \ac \psi$. 
				Now, \rref{lem:ac_invariance} implies $\preeq{\pstate{u} \cdot \trace_1}{\trace_2} \vDash \A$ such that $\pstate{w} \cdot \trace \vDash \psi$.
			\end{enumerate}
		\end{enumerate}
		\item[] \RuleName{acDropComp}:
		Let $\pstate{v} \vDash [ \alpha ] \ac \psi$ and $\computation \in \sem{\alpha \parOp \beta}$. 
		Then $(\pstate{v}, \trace \downarrow \alpha, \pstate{w}_\alpha) \in \sem{\alpha}$ with $\pstate{w} = \pstate{w}_\alpha \merge \pstate{w}_\beta$ for some $\pstate{w}_\beta \in \botop{\states}$.
		For \acCommit, assume $\proppre{\pstate{v}}{\trace} \vDash \A$.
		If $\trace[pre]_\alpha \prec \trace \downarrow \alpha$, then $\trace[pre] \prec \trace$ exists such that $\trace[pre]_\alpha = \trace[pre] \downarrow \alpha$. 
		Thus, $(\pstate{v}, \trace[pre], \bot) \in \sem{\alpha \parOp \beta}$ and $(\pstate{v}, \trace[pre]_\alpha, \bot) \in \sem{\alpha}$ by prefix-closedness (\rref{prop:prefixClosedAndTotal}).
		Hence, $\proppre{\pstate{v}}{\trace \downarrow \alpha} \vDash \A$ by \rref{lem:noninterference} since $\beta$ does not interfere with $[ \alpha ] \ac \psi$.
		Therefore, $\pstate{v} \cdot (\trace \downarrow \alpha) \vDash \C$ by $\pstate{v} \vDash [ \alpha ] \ac \psi$.
		Now, we obtain $\pstate{v} \cdot \trace \vDash \C$ using \rref{lem:noninterference} again.
		For \acPost, assume $\pstate{w} \neq \bot$ and $\preeq{\pstate{v}}{\trace} \vDash \A$. 
		Then $\pstate[ind=\alpha]{w} \neq \bot$ by the definition of $\merge$ in \rref{sec:semantics} and $\preeq{\pstate{v}}{\trace \downarrow \alpha} \vDash \A$ by \rref{lem:noninterference} as above. 
		Hence, $\pstate{w}_\alpha \cdot (\trace \downarrow \alpha) \vDash \psi$ by $\pstate{v} \vDash [ \alpha ] \ac \psi$.
		Finally, $\pstate{w} \cdot \trace \vDash \psi$ by \rref{lem:noninterference}.
		\item[] \RuleName{boxesDual}:
		Let $\pstate{v} \vDash [ \alpha ] \psi$ and $\computation \in \sem{\alpha}$. 
		Then \acCommit holds trivially. 
		For \acPost, assume $\pstate{w} \neq \bot$ and $\preeq{\pstate{v}}{\trace} \vDash \true$.
		Then $\pstate{w} \cdot \trace \vDash \psi$ by assumption.
		Conversely, let $\pstate{v} \vDash [ \alpha ] \acpair{\true, \true} \psi$ and $\computation \in \sem{\alpha}$ with $\pstate{w} \neq \bot$. 
		Then $\pstate{w} \cdot \trace \vDash \psi$ holds by \acPost since $\preeq{\pstate{v}}{\trace} \vDash \true$ holds trivially.
		\item[] \RuleName{acBoxesDist}:
		The implication $[ \alpha ] \acpair{\A, \C_1} \psi_1 \wedge [ \alpha ] \acpair{\A, \C_2} \psi_2 \rightarrow [ \alpha ] \acpair{\A, \C_1 \wedge \C_2} (\psi_1 \wedge \psi_2)$ can be easily shown by the semantics.
		The converse impliciation derives by rule \RuleName{acMono}.
		\item[] \RuleName{acCom}:
		Let $\pstate{v} \vDash [ \send{}{}{} ] \ac \psi$.
		Then $\pstate{v} \vDash \C$ by \RuleName{acWeak}.
		Now, assume $\pstate{v} \vDash \A$ and let $\computation \in \sem{\send{}{}{}}$ with $\pstate{w} \neq \bot$.
		Since $\pstate{v} \vDash \A$ and $\semLen{\trace} = 1$,
		we have $\proppre{\pstate{v}}{\trace} \vDash \A$.
		Thus, $\pstate{v} \cdot \trace \vDash \C$ by \acCommit, 
		which implies $\pstate{w} \cdot \trace \vDash \C$ by \rref{lem:ac_invariance}.
		For $\A \rightarrow \psi$, assume $\pstate{w} \cdot \trace \vDash \A$.
		Then $\pstate{v} \cdot \trace \vDash \A$ by \rref{lem:ac_invariance} such that $\preeq{\pstate{v}}{\trace} \vDash \A$.
		Hence, $\pstate{w} \cdot \trace \vDash \psi$ by \acPost.
		
		Conversely, let $\pstate{v} \vDash \Phi([ \send{}{}{} ] \Phi(\psi))$, where $\Phi(\phi) \equiv \C \wedge (\A \rightarrow \phi)$ and $\computation \in \sem{\send{}{}{}}$.
		For \acCommit, assume $\proppre{\pstate{v}}{\trace} \vDash \A$.
		If $\trace = \epsilon$, we obtain $\pstate{v} \vDash \C$ by assumption.
		If $\trace \neq \epsilon$, then $\pstate{v} \vDash \A$ such that $\pstate{v} \vDash [ \send{}{}{} ] \Phi(\psi)$.
		\Wlossg $\pstate{w} \neq \bot$ since $\send{}{}{}$ has a terminating computation from any $\pstate{v}$.
		Hence, $\pstate{w} \cdot \trace \vDash \Phi(\psi)$ such that $\pstate{w} \cdot \trace \vDash \C$.
		By \rref{lem:ac_invariance}, we obtain $\pstate{v} \cdot \trace \vDash \C$.
		For \acPost, assume $\pstate{w} \neq \bot$ and $\preeq{\pstate{v}}{\trace} \vDash \A$.
		As above $\pstate{w} \cdot \trace \vDash \Phi(\psi)$ such that $\pstate{w} \cdot \trace \vDash \A \rightarrow \psi$.
		Since $\pstate{v} \cdot \trace \vDash \A$ by assumption, which implies $\pstate{w} \cdot \trace \vDash \A$ by \rref{lem:ac_invariance}, we obtain $\pstate{w} \cdot \trace \vDash \psi$.

		\item[] \RuleName{send}:
		\newcommand{\tstate}{\pstate{o}}%
		\newcommand{\qstate}{\pstate{u}}%
		Let $\pstate{v} \vDash \fa{\historyVar_0} (\historyVar_0 = \historyVar \cdot \comItem{\ch{}, \rp, \globalTime} \rightarrow \psi(\historyVar_0))$ and $\computation \in \sem{\send{}{}{}}$ with $\pstate{w} \neq \bot$.
		Instantiating $\fa{\historyVar_0}$ with $\trace_0 = \sem{\historyVar \cdot \comItem{\ch{}, \rp, \globalTime}} \pstate{v}$ to $\qstate = \pstate{v} \subs{\historyVar_0}{\trace_0}$,
		we have $\qstate \vDash \historyVar_0 = \historyVar \cdot \comItem{\ch{}, \rp, \globalTime}$.
		Therefore, $\qstate \vDash \psi(\historyVar_0)$.	
		By substitution, $\qstate \subs{\historyVar}{\sem{\historyVar_0} \qstate} \vDash \psi(\historyVar)$.
		Since $\historyVar_0 \not \in \FV(\rp)$,
		we have $\sem{\historyVar \cdot \comItem{\ch{}, \rp, \globalTime}} \qstate = \sem{\historyVar \cdot \comItem{\ch{}, \rp, \globalTime}} \pstate{v}$ by coincidence (\rref{lem:formulaCoincidence}).
		Moreover, $\sem{\historyVar_0} \qstate = \sem{\historyVar \cdot \comItem{\ch{}, \rp, \globalTime}} \qstate$ such that $\qstate \subs{\historyVar}{\sem{\historyVar \cdot \comItem{\ch{}, \rp, \globalTime}} \pstate{v}} \vDash \psi(\historyVar)$,
		which implies $\pstate{v} \subs{\historyVar}{\sem{\historyVar \cdot \comItem{\ch{}, \rp, \globalTime}} \pstate{v}} \vDash \psi(\historyVar)$ by coincidence because $\historyVar_0 \not \in \FV(\psi(\historyVar))$.
		Now, observe that $\trace = \comItem{\historyVar, \ch{}, \sem{\rp} \pstate{v}, \statetime{\pstate{v}}}$ such that $\pstate{v} \cdot \trace = \pstate{v} \subs{\historyVar}{\sem{\historyVar \cdot \comItem{\ch{}, \rp, \globalTime}} \pstate{v}}$.
		Finally, $\pstate{w} \cdot \trace \vDash \psi(\historyVar)$ because $\pstate{w} = \pstate{v}$.

		Conversly, let $\pstate{v} \vDash [ \send{}{}{} ] \psi(\historyVar)$.
		To prove the quantifier let $\trace_0 \in \traces$ and assume $\pstate{v} \subs{\historyVar_0}{\trace_0} \vDash \historyVar_0 = \historyVar \cdot \comItem{\ch{}, \rp, \globalTime}$.
		Now, observe that $(\pstate{v}, \trace, \pstate{v}) \in \sem{\send{}{}{}}$ with $\trace = \comItem{\historyVar, \ch{}, \sem{\rp} \pstate{v}, \statetime{\pstate{v}}}$ such that $\pstate{v} \cdot \trace \vDash \psi(\historyVar)$.
		Since $\pstate{v} \cdot \trace = \pstate{v} \subs{\historyVar}{\sem{\historyVar \cdot \comItem{\ch{}, \rp, \globalTime}} \pstate{v}}$ 
		and $\sem{\historyVar \cdot \comItem{\ch{}, \rp, \globalTime}} \pstate{v} = \sem{\historyVar_0} \pstate{v} = \trace_0$,
		we obtain $\pstate{v} \subs{\historyVar_0}{\trace_0} \vDash \psi(\historyVar_0)$ by substitution.

		\item[] \RuleName{comDual}:
		First, observe that $\sem{\receive{}{}{}} = \sem{x \ceq *} \continuation \sem{\send{}{}{x}}$.

		Now, let $\pstate{v} \vDash [ \receive{}{}{} ] \ac \psi$ and \wlossg $(\pstate{v}, \epsilon, \pstate{u}) \in \sem{x \ceq *}$ with $\pstate{u} \neq \bot$.
		Futher, let $(\pstate{u}, \trace, \pstate{w}) \in \sem{\send{}{}{x}}$.
		Then $\computation \in \sem{x \ceq *} \continuation \sem{\send{}{}{x}} = \sem{\receive{}{}{}}$.
		For \acCommit, assume $\proppre{\pstate{u}}{\trace} \vDash \A$.
		Since $[ \receive{}{}{} ] \ac \psi$ is well-formed,
		we have $(\FV(\A) \cup \FV(\C)) \cap \BV(\receive{}{}{}) \subseteq \TVar$.
		Since $x \in \BV(\receive{}{}{})$ and $\BV(x \ceq *) = \{x\}$, we have $\pstate{v} = \pstate{u}$ on $\FV(\A) \cup \FV(\C)$.
		By coincide (\rref{lem:formulaCoincidence}),
		we obtain $\proppre{\pstate{v}}{\trace} \vDash \A$ such that $\pstate{v} \cdot \trace \vDash \C$ by premise, which implies $\pstate{u} \cdot \trace \vDash \C$ by coincidence again.
		For \acPost, assume $\pstate{w} \neq \bot$ and $\preeq{\pstate{u}}{\trace} \vDash \A$.
		Then $\preeq{\pstate{v}}{\trace} \vDash \A$ as above.
		By premise, $\pstate{w} \cdot \trace \vDash \psi$.

		Conversely, let $\pstate{v} \vDash [ x \ceq * ] [ \send{}{}{x} ] \ac \psi$ and $\computation \in \sem{\receive{}{}{}}$.
		Then $(\pstate{v}, \epsilon, \pstate{u}) \in \sem{x \ceq *}$ and $(\pstate{u}, \trace, \pstate{w}) \in \sem{\send{}{}{x}}$ exist.
		By premise, we obtain $\pstate{u} \vDash \phi$,
		where $\phi \equiv [ \send{}{}{x} ] \ac \psi$.
		Now, we have to prove $\pstate{v} \vDash [ \receive{}{}{} ] \ac \psi$.
		Therfore, assume $\proppre{\pstate{v}}{\trace} \vDash \A$ for \acCommit.
		As in the converse direction, 
		$\proppre{\pstate{u}}{\trace} \vDash \A$ sucht that $\pstate{u} \cdot \trace \vDash \C$ by $\pstate{u} \vDash \phi$, 
		which again implies $\pstate{v} \cdot \trace \vDash \C$.
		For \acPost, assume $\pstate{w} \neq \bot$ and $\preeq{\pstate{v}}{\trace}$.
		Finally, $\pstate{w} \cdot \trace \vDash \psi$ by $\pstate{u} \vDash \phi$.

		\item[] \RuleName{acMono}:
		Let $\pstate{v} \vDash [ \alpha ] \acj{1} \psi$ and $\computation \in \sem{\alpha}$. 
		For \acCommit, assume $\proppre{\pstate{v}}{\trace} \vDash \A_2$.
		Then $\proppre{\pstate{v}}{\trace} \vDash \A_1$, which implies $\pstate{v} \cdot \trace \vDash \C_1$ by assumption. 
		Finally, $\pstate{v} \cdot \trace \vDash \C_2$ by premise $\C_1 \rightarrow \C_2$.
		For \acPost, assume $\pstate{w} \neq \bot$ and $\preeq{\pstate{v}}{\trace} \vDash \A_2$. 
		Then $\preeq{\pstate{v}}{\trace} \vDash \A_1$ by premise $\A_2 \rightarrow \A_1$. 
		Thus, $\pstate{w} \cdot \trace \vDash \psi_1$ by assumption. 
		Finally, $\pstate{w} \cdot \trace \vDash \psi_2$ by premise $\psi_1 \rightarrow \psi_2$.

		\item[] \RuleName{acG}:
		Let the premise $\C \wedge \psi$ be valid.
		To prove validity of $[ \alpha ] \ac \psi$, let $\pstate{v}$ be an arbitrary state and $\computation \in \sem{\alpha}$.
		Condition \acCommit holds since $\pstate{v} \cdot \trace \vDash \C$ by validity of $\C$.
		For \acPost, assume $\pstate{w} \neq \bot$.
		Then $\pstate{w} \cdot \trace \vDash \psi$ since $\psi$ is valid.

		\item[] \RuleName{gtime}:
		Let $\odeName \equiv \evolution{x' = \rp}{\chi}$.
		Now, $\computation \in \sem{\evolution{\globalTime' = 1, \odeName}{non}}$ iff a solution $\odeSolution : [0, \duration] \rightarrow \states$ from $\pstate{v}$ to $\pstate{w}$ exists such that $\odeSolution(\zeta) \vDash \globalTime' = 1 \wedge \globalTime' = 1 \wedge x' = \rp$
		and $\pstate{v} = \odeSolution(\zeta)$ on $\{x, \globalTime\}^\complement$ for $\zeta \in [0, \duration]$
		iff, since $\globalTime' = 1 \wedge \globalTime' = 1 \leftrightarrow \globalTime' = 1$ at all $\odeSolution(\zeta)$, we have $\computation \in \sem{\odeName}$.
		Thus, $\pstate{v} \vDash [ \evolution{\globalTime' = 1, \odeName}{non} ] \psi$ iff $\pstate{v} \vDash [ \odeName ] \psi$.
		\qedhere
	\end{itemize}
\end{proof}

\begin{corollary} \label{cor:derived}
	The axioms and rules in \rref{fig:derived} derive syntactically.
	Additionally, we derive the following rule \RuleName{acInvariant} in preparation of the proof of rule \RuleName{acLoop}:
	\vspace*{-1em}
	\begin{prooftree}
		\Axiom{$\inv \vdash [ \alpha ] \ac \psi$}	

		\RuleNameLeft{acInvariant}
		\UnaryInf{$\inv \vdash [ \repetition{\alpha} ] \ac \psi$}
	\end{prooftree}
\end{corollary}
\begin{proof}[\rref{cor:derived}]
	We prove the axioms and rules by derivation in our calculus.
	For rule \RuleName{CG}, see \rref{fig:proof_CG},
	for rule \RuleName{acParCompRight}, see \rref{fig:proof_acParComp}
	for rule \RuleName{acInvariant}, see \rref{fig:proof_acInvariant},
	for rule \RuleName{acLoop}, see \rref{fig:proof_acLoop}, and
	for rule \RuleName{acSendRight}, see \rref{fig:proof_acSend}.
	\qedhere
\end{proof}

\begin{figure}
	\begin{prooftree}
		\Axiom{$*$}

		\RuleNameLeft{equal}
		\UnaryInf{$\vdash \historyVar \cdot \te = \historyVar \cdot \te, \Phi(\historyVar)$}

		\Axiom{$*$}

		\RuleNameRight{Id}
		\UnaryInf{$\Phi(\historyVar) \vdash \Phi(\historyVar)$}

		\RuleNameLeft{implL}
		\BinaryInf{$\historyVar \cdot \te = \historyVar \cdot \te \rightarrow \Phi(\historyVar) \vdash \Phi(\historyVar)$}

		\RuleNameLeft{forallL}
		\UnaryInf{$\fa{\historyVar_0} ( \historyVar_0 = \historyVar \cdot \te \rightarrow \Phi(\historyVar_0)) \vdash \Phi(\historyVar) $}

		\RuleNameLeft{send}
		\UnaryInf{$[ \send{}{}{} ] \Phi(\historyVar) \vdash \Phi(\historyVar)$}

		\RuleNameLeft{WL, WR}
		\UnaryInf{$\Gamma, [ \send{}{}{} ] \Phi(\historyVar) \vdash \Phi(\historyVar), \Delta$}

		\Axiom{$\Gamma, \historyVar_0 = \historyVar \cdot \te \vdash \Phi(\historyVar_0), \Delta$}

		\RuleNameRight{implR}
		\UnaryInf{$\Gamma \vdash \historyVar_0 = \historyVar \cdot \te \rightarrow \Phi(\historyVar_0), \Delta$}

		\RuleNameRight{forallR}
		\UnaryInf{$\Gamma \vdash \fa{\historyVar_0} (\historyVar_0 = \historyVar \cdot \te \rightarrow \Phi(\historyVar_0)), \Delta$}

		\RuleNameRight{send}
		\UnaryInf{$\Gamma \vdash [ \send{}{}{} ] \Phi(\historyVar), \Delta$}

		\RuleNameRight{WR}
		\UnaryInf{$\Gamma \vdash [ \send{}{}{} ] \Phi(\historyVar), \Phi(\historyVar), \Delta$}

		\RuleNameRight{Cut}
		\BinaryInf{$\Gamma \vdash \Phi(\historyVar), \Delta$}
	\end{prooftree}
	\vspace*{-1em}
	\caption{Derivation of rule \RuleName{CG}. In the proof, $\te$ abbreviates $\comItem{\ch{}, \rp, \globalTime}$ and $\Phi(\historyVar)$ abbreviates $[ \alpha(\historyVar) ] \psi$.}
	\label{fig:proof_CG}
\end{figure}

\begin{figure}[ht!]
	\vspace*{-3em}
	\begin{small}
		\begin{prooftree}
			\Axiom{$\vdash \K$}

			\RuleNameLeft{forallR}
			\UnaryInf{$\vdash \universal \K$}

			\RuleNameLeft{WL, WR}
			\UnaryInf{$\Gamma \vdash \universal \K, \Delta$}

			\Axiom{$\triangleleft$}

			\RuleNameRight{acDropComp}
			\Axiom{$\Gamma \vdash [ \alpha_j ] \acpair{\A_j, \C_j} \psi_j, \Delta$ \sidecondition{for $j = 1, 2$}}

			\RuleNameRight{acDropComp}
			\UnaryInf{$\Gamma \vdash [ \alpha_1 \parOp \alpha_2 ] \acpair{\A_j, \C_j} \psi_j, \Delta$ \sidecondition{for $j = 1, 2$}}

			\RuleNameRight{acMono}
			\BinaryInf{$\Gamma \vdash [ \alpha_1 \parOp \alpha_2 ] \acpair{\A_1 \wedge \A_2, \C_j} \psi_j, \Delta$ \sidecondition{for $j = 1, 2$}}

			\RuleNameRight{acBoxesDist}
			\UnaryInf{$\Gamma \vdash [ \alpha_1 \parOp \alpha_2 ] \acpair{\A_1 \wedge \A_2, \C_1 \wedge \C_2} (\psi_1 \wedge \psi_2), \Delta$}

			\RuleNameRight{andR}
			\BinaryInf{$\Gamma \vdash \universal \K \wedge [ \alpha_1 \parOp \alpha_2 ] \acpair{\A_1 \wedge \A_2, \C_1 \wedge \C_2} (\psi_1 \wedge \psi_2), \Delta$}

			\RuleNameRight{K}
			\UnaryInf{$\Gamma \vdash [ \alpha_1 \parOp \alpha_2 ] \acpair{\A, \C_1 \wedge \C_2} (\psi_1 \wedge \psi_2), \Delta$}
		\end{prooftree}
	\end{small}
	\vspace*{-1em}
	\caption{
		Derivation of rule \RuleName{acParCompRight}.
		For $j = 1, 2$, 
		the open premise $\A_1 \wedge \A_2 \rightarrow \A_j$ marked by $\triangleleft$ obviously holds.
		Axiom \RuleName{acDropComp} can be applied since $\alpha_{3-j}$ does not interfere with $[ \alpha_j ] \acj{j} \psi_j$ by premise.
	}
	\label{fig:proof_acParComp}
\end{figure}

\begin{figure}[ht!]
	\vspace{-1em}
	\hspace{-2.8cm}\begin{tikzpicture}
		\node (lefttree)
		{\vbox{
			\begin{small}
				\begin{prooftree}
					\Axiom{$*$}
	
					\RuleNameLeft{Id}
					\UnaryInf{$I, \C, [ \alpha ] \ac I \vdash \C$}

					\RuleNameLeft{andL}
					\UnaryInf{$I, \C \wedge [ \alpha ] \ac I \vdash \C$}

					\RuleNameLeft{acWeak}
					\UnaryInf{$I, [ \alpha ] \ac I \vdash \C$}

					\Axiom{$I \vdash [ \alpha ] \ac I$}

					\RuleNameLeft{Cut}
					\BinaryInf{$I \vdash \C$}
				\end{prooftree}
			\end{small}
		}};

		\node (righttree) [right=-7.5cm of lefttree]
		{\vbox{
			\begin{small}
				\begin{prooftree}
					\Axiom{$*$}
	
					\RuleNameRight{Id}
					\UnaryInf{$I, \A, \true \vdash I$}

					\RuleNameRight{implR}
					\UnaryInf{$I, \A \vdash \true \rightarrow I$}

					\RuleNameRight{test}
					\UnaryInf{$I, \A \vdash [ \alpha^0 ] I$}

					\RuleNameRight{implR}
					\UnaryInf{$I \vdash \A \rightarrow [ \alpha^0 ] I$}
				\end{prooftree}
			\end{small}
		}};
		
		\node (proof) [below right= -.5cm and -9.5cm of lefttree] 
		{\vbox{\begin{small}
			\begin{prooftree}
				\Axiom{$$}
			
				\Axiom{$$}
			
				\RuleNameLeft{andR}
				\BinaryInf{$I \vdash \C \wedge (\A \rightarrow [ \alpha^0 ] I)$}
			
				\RuleNameLeft{acNoCom}
				\UnaryInf{$I \vdash [ \alpha^0 ] \ac I$}
			
				\Axiom{$*$}
			
				\UnaryInf{$\vdash \true$}
			
				\Axiom{$I \vdash [ \alpha ] \ac I$}
			
				\RuleNameRight{implR}
				\UnaryInf{$\vdash I \rightarrow [ \alpha ] \ac I$}
			
				\RuleNameRight{andR}
				\BinaryInf{$\vdash \true \wedge (I \rightarrow [ \alpha ] \ac I)$}
			
				\RuleNameRight{acG}
				\UnaryInf{$I \vdash [ \repetition{\alpha} ] \acpair{\A, \true} (I \rightarrow [ \alpha ] \ac I)$}
			
				\RuleNameRight{andR}
				\BinaryInf{$I \vdash [ \alpha^0 ] \ac I \wedge [ \repetition{\alpha} ] \acpair{\A, \true} (I \rightarrow [ \alpha ] \ac I)$}
			
				\RuleNameRight{acInduction}
				\UnaryInf{$I \vdash [ \repetition{\alpha} ] \ac I$}
			\end{prooftree}
		\end{small}}};
		
		\coordinate (arrowendleft) at (lefttree.south);
		\draw[-stealth] (proof.north) ++ (-3.3,-1.3) to [in=270, out=90] ([shift={(.5, .25)}]arrowendleft);

		\coordinate (arrowendright) at (righttree.south);
		\draw[-stealth] (proof.north) ++ (-1.2,-1.3) to [in=270, out=90] ([shift={(-0.95, .35)}]arrowendright);
	\end{tikzpicture}
	\vspace*{-2.5em}
	\caption{Derivation of rule \RuleName{acInvariant}}
	\label{fig:proof_acInvariant}
\end{figure}

\begin{figure}[ht!]
	\begin{small}
		\begin{prooftree}
			\Axiom{$\C, \inv \vdash [ \alpha ] \ac \inv$}
	
			\RuleNameLeft{andL}
			\UnaryInf{$\inv_\C \vdash [ \alpha ] \ac \inv$}
		
			\RuleNameLeft{acMono}
			\UnaryInf{$\inv_\C \vdash [ \alpha ] \ac \inv_\C$}

			\RuleNameLeft{acInvariant}
			\UnaryInf{$\inv_\C \vdash [ \repetition{\alpha} ] \ac \inv_\C$}
		
			\RuleNameLeft{implR}
			\UnaryInf{$\Gamma \vdash \inv_\C \rightarrow [ \repetition{\alpha} ] \ac \inv_\C, \Delta$}
		
			\Axiom{$\Gamma \vdash \inv_\C, \Delta$}
		
			\Axiom{$*$}

			\RuleNameLeft{Id}
			\UnaryInf{$\inv, \C \vdash \C$}
			
			\RuleNameLeft{andL}
			\UnaryInf{$\inv_\C \vdash \C$}

			\Axiom{$\C, \A, \inv \vdash \psi$}

			\RuleNameRight{andL}
			\UnaryInf{$\inv_\C, \A \vdash \psi$}

			\RuleNameRight{implR}
			\UnaryInf{$\inv_\C \vdash \A \rightarrow \psi$}

			\RuleNameRight{andR}
			\BinaryInf{$\inv_\C \vdash (\C \wedge (\A \rightarrow \psi))$}
		
			\RuleNameRight{acMono}
			\UnaryInf{$[ \repetition{\alpha} ] \ac \inv_\C \vdash [ \repetition{\alpha} ] \ac (\C \wedge (\A \rightarrow \psi))$}
		
			\RuleNameRight{acWeak}
			\UnaryInf{$[ \repetition{\alpha} ] \ac \inv_\C \vdash [ \repetition{\alpha} ] \ac \psi$}
		
			\RuleNameRight{implL}
			\BinaryInf{$\Gamma, \inv_\C \rightarrow [ \repetition{\alpha} ] \ac \inv_\C \vdash [ \repetition{\alpha} ] \ac \psi, \Delta$}
		
			\RuleNameRight{Cut}
			\BinaryInf{$\Gamma \vdash [ \repetition{\alpha} ] \ac \psi, \Delta$}
		\end{prooftree}
	\end{small}
	\vspace*{-1em}
	\caption{Derivation of rule \RuleName{acLoop}.
	In the proof,
	we silently apply weakening several times.
	The formula $\inv_\C \equiv \inv \wedge \C$ combines the invariant with the commitment $\C$. 
	Moreover, $\Psi$ is short for $\C \wedge (\A \rightarrow \psi)$.}
	\label{fig:proof_acLoop}
\end{figure}

\begin{figure}[ht!]
	\begin{small}
		\begin{prooftree}

			\Axiom{$\Gamma \vdash \C(\historyVar), \Delta$}

			\Axiom{$\Gamma, \A(\historyVar), H_0 \vdash \C(\historyVar_0), \Delta$}

			\Axiom{$\Gamma, \A(\historyVar), H_0, \A(\historyVar_0) \vdash \psi(\historyVar_0), \Delta$}

			\RuleNameRight{implR}
			\UnaryInf{$\Gamma, \A(\historyVar), H_0 \vdash \A(\historyVar_0) \rightarrow \psi(\historyVar_0), \Delta$}

			\RuleNameRight{andR}
			\BinaryInf{$\Gamma, \A(\historyVar), H_0 \vdash \C(\historyVar_0) \wedge (\A(\historyVar_0) \rightarrow \psi(\historyVar_0)), \Delta$}

			\RuleNameRight{implR}
			\UnaryInf{$\Gamma, \A(\historyVar) \vdash H_0 \rightarrow 
			 (\C(\historyVar_0) \wedge (\A(\historyVar_0) \rightarrow \psi(\historyVar_0))), \Delta$}

			\RuleNameRight{forallR}
			\UnaryInf{$\Gamma, \A(\historyVar) \vdash \fa{\historyVar_0} ( H_0 \rightarrow 
			 (\C(\historyVar_0) \wedge (\A(\historyVar_0) \rightarrow \psi(\historyVar_0))) ), \Delta$}

			\RuleNameRight{send}
			\UnaryInf{$\Gamma, \A(\historyVar) \vdash [ \send{}{}{} ]  (\C(\historyVar) \wedge (\A(\historyVar) \rightarrow \psi(\historyVar))), \Delta$}

			\RuleNameRight{implR}
			\UnaryInf{$\Gamma \vdash \A(\historyVar) \rightarrow [ \send{}{}{} ]  (\C(\historyVar) \wedge (\A(\historyVar) \rightarrow \psi(\historyVar))), \Delta$}

			\RuleNameRight{andR}
			\BinaryInf{$\Gamma \vdash \C(\historyVar) \wedge (\A(\historyVar) \rightarrow [ \send{}{}{} ]  (\C(\historyVar) \wedge (\A(\historyVar) \rightarrow \psi(\historyVar)))), \Delta$}

			\RuleNameRight{acCom}
			\UnaryInf{$\Gamma \vdash [ \send{}{}{} ] \acpair{\A(\historyVar), \C(\historyVar) } \psi(\historyVar), \Delta$}
		\end{prooftree}
	\end{small}
	\vspace*{-1em}
	\caption{Derivation of rule \RuleName{acSendRight}. In the proof, $H_0$ is short for $\historyVar_0 = \historyVar \cdot \comItem{\ch{}, \rp, \globalTime}$.}
	\label{fig:proof_acSend}
\end{figure}

\section{Details of the Example}

We continue the proof of the example from \rref{sec:example} and provide the remaining derivations:

\let\oldparameters\parameters
\renewcommand{\parameters}[1]{\oldparameters{\text{\scriptsize $#1$}}}

\begin{figure}[ht!]
	\begin{small}
		\begin{prooftree}[shape=justified]
			\Axiom{$\triangleright$\rref{fig:leader_vel_continuous}}
		
			\UnaryInf{$\lCommit, \invariant{l}, \chi_{vel} \parameters{v_0} \vdash [ \Plant_l \parameters{v_0} ] \invariant{l} \parameters{v_0, x_l, \waitvar, \globalTime, \historyVar}$}

			\RuleNameRight{projNotIn}
			\UnaryInf{$\lCommit, \invariant{l}, \chi_{vel} \parameters{v_0} \vdash [ \Plant_l \parameters{v_0} ] \invariant{l} \parameters{v_0, x_l, \waitvar, \globalTime, \historyVar \cdot \comItem{\ch{vel}, v_0, \globalTime}}$}

			\SideAx$\triangleright \RuleName{Id} \quad \triangleright\!\text{\rref{fig:leader_commit}}$
			\RuleNameRight{subsR}
			\UnaryInf{$\begin{aligned}
				& \lCommit, \invariant{l}, \chi_{vel} \parameters{v_0}, \historyVar_\ch{vel} = \historyVar \cdot \comItem{\ch{vel}, v_0, \globalTime} \\
				&\qquad \vdash [ \Plant_l \parameters{v_0} ] \invariant{l} \parameters{v_0, x_l, \waitvar, \globalTime, \historyVar_\ch{vel}}
			\end{aligned}$}
		
			\SideAx$\triangleright \text{\rref{fig:leader_vel_lose}}$
			\RuleNameRight{acSendRight}
			\UnaryInf{$\lCommit, \invariant{l}, \chi_{vel} \parameters{v_0} \vdash [ \send{\ch{vel}}{\historyVar}{v_0} ] \acpair{\true, \lCommit} [ \Plant_l \parameters{v_0} ] \invariant{l} \parameters{v_0, x_l, \waitvar, \globalTime, \historyVar}$}
		
			\RuleNameRight{acChoice}
			\UnaryInf{$\lCommit, \invariant{l}, \chi_{vel} \parameters{v_0} \vdash [ \send{\ch{vel}}{\historyVar}{v_0} \cup \skipProg ] \acpair{\true, \lCommit} [ \Plant_l \parameters{v_0} ] \invariant{l} \parameters{v_0, x_l, \waitvar, \globalTime, \historyVar}$}
		
			\RuleNameRight{test, implR}
			\UnaryInf{$\lCommit, \invariant{l} \vdash [ \test{\chi_{vel} \parameters{v_0}} ] [ \send{\ch{vel}}{\historyVar}{v_0} \cup \skipProg ] \acpair{\true, \lCommit} [ \Plant_l \parameters{v_0} ] \invariant{l} \parameters{v_0, x_l, \waitvar, \globalTime, \historyVar}$}
		
			\RuleNameRight{nondetAssign, forallR}
			\UnaryInf{$\lCommit, \invariant{l} \vdash [ v_l \ceq * ] [ \test{\chi_{vel} \parameters{v_l}} ] [ \send{\ch{vel}}{\historyVar}{v_l} \cup \skipProg ] \acpair{\true, \lCommit} [ \Plant_l \parameters{v_l} ] \invariant{l} \parameters{v_l, x_l, \waitvar, \globalTime, \historyVar}$}
		
			\SideAx$\triangleright\RuleName{Id}\!$
			\RuleNameRight{composition}
			\UnaryInf{$\lCommit, \invariant{l} \vdash [ v_l \ceq * \seq \test{\chi_{vel} \parameters{v_l}} ] [ \send{\ch{vel}}{\historyVar}{v_l} \cup \skipProg ] \acpair{\true, \lCommit} [ \Plant_l \parameters{v_l} ] \invariant{l} \parameters{v_l, x_l, \waitvar, \globalTime, \historyVar}$}
		
			\RuleNameRight{andR}
			\UnaryInf{$\lCommit, \invariant{l} \vdash \lCommit \wedge [ v_l \ceq * \seq \test{\chi_{vel} \parameters{v_l}} ] [ \send{\ch{vel}}{\historyVar}{v_l} \cup \skipProg ] \acpair{\true, \lCommit} [ \Plant_l \parameters{v_l} ] \invariant{l} \parameters{v_l, x_l, \waitvar, \globalTime, \historyVar}$}
		
			\RuleNameRight{acNoCom}
			\UnaryInf{$\lCommit, \invariant{l} \vdash [ v_l \ceq * \seq \test{\chi_{vel} \parameters{v_l}} ] \acpair{\true, \lCommit} [ \send{\ch{vel}}{\historyVar}{v_l} \cup \skipProg ] \acpair{\true, \lCommit} [ \Plant_l \parameters{v_l} ] \invariant{l} \parameters{v_l, x_l, \waitvar, \globalTime, \historyVar}$}
		
			\SideAx$\triangleright \text{\rref{fig:leader_pos}}$
			\RuleNameRight{acComposition}
			\UnaryInf{$\lCommit, \invariant{l} \vdash [ \ctrlNotify(\historyVar) ] \acpair{\true, \lCommit} [ \Plant_l \parameters{v_l} ] \invariant{l} \parameters{v_l, x_l, \waitvar, \globalTime, \historyVar}$}
		
			\UnaryText{Execution by \RuleName{acComposition}, \RuleName{acChoice}, \RuleName{acNoCom}, \RuleName{acWeak}, and \RuleName{andR}}

			\SideAx$\triangleright_1\RuleName{TA, real}\;\;\triangleright_2\RuleName{real}$
			\UnaryInf{$\lCommit, \invariant{l} \vdash [ (\ctrlNotify(\historyVar) \cup \ctrlUpdate(\historyVar)) \seq \Plant_l \parameters{v_l} ] \acpair{\true, \lCommit} \invariant{l} \parameters{v_l, x_l, \waitvar, \globalTime, \historyVar}$}

			\RuleNameRight{acLoop}
			\UnaryInf{$\Gamma \vdash [ \progtt{leader}(\historyVar) ] \acpair{\true, \lCommit} \psi_l$}
		\end{prooftree}
	\end{small}
	\vspace*{-1em}
	\caption{
		Continues the derivation from \rref{fig:pardecomp}.
		With $\Gamma$ we abbreviate the list 
		$\varphi, \historyVar_1 = \historyVar_0 \cdot \comItem{\ch{vel}, 0, \globalTime}, \historyVar = \historyVar_1 \cdot \comItem{\ch{pos}, x_l, \globalTime}$ 
		of formulas.
		The induction uses invariant 
		$\invariant{l} \equiv 
		\val{\historyVar\downarrow \ch{pos}} \le x_l \wedge 
		v_l \ge 0 \wedge 
		w = \timespan{\globalTime}{\historyVar\downarrow \ch{pos}} \le \epsilon$.
		The induction base $\Gamma \vdash \invariant{l}$ ($\triangleright_1$) closes by trace algebra \RuleName{TA} and real arithmetic \RuleName{real}.
		Postcondition $\invariant{l} \vdash \psi_l$ ($\triangleright_2$) holds by \RuleName{real}.
	}
	\label{fig:leader_vel_send}
\end{figure}

\begin{figure}[ht!]
	\vspace*{-3em}
	\begin{small}
		\begin{prooftree}[shape=justified, JustifiedScoresWidth=.85\textwidth]
			\Axiom{$*$}
		
			\RuleNameRight{real}
			\UnaryInf{$\begin{aligned}
				& \invariant{l}, \chi_{vel} \parameters{v_0}, t \ge 0, \waitvar + t \le \epsilon \\
				&\qquad \vdash \waitvar + t = \timespan{\globalTime + t}{\historyVar} \le \epsilon \wedge v_0 \ge 0 \wedge \val{\historyVar} \le x_l + t \cdot v_0
			\end{aligned}$} 

			\RuleNameRight{forallL}
			\UnaryInf{$\invariant{l}, \chi_{vel} \parameters{v_0}, t \ge 0, \mathcal{E} \vdash \invariant{l} \parameters{v_0, x_l + t \cdot v_0, \waitvar + t, \globalTime + t, \historyVar}$}

			\RuleNameRight{implR, WL}
			\UnaryInf{$\lCommit, \invariant{l}, \chi_{vel} \parameters{v_0}, t \ge 0 \vdash \mathcal{E} \rightarrow \invariant{l} \parameters{v_0, x_l + t \cdot v_0, \waitvar + t, \globalTime + t, \historyVar}$}

			\RuleNameRight{forallR}
			\UnaryInf{$\lCommit, \invariant{l}, \chi_{vel} \parameters{v_0} \vdash \fa{t{\ge}0} ( \mathcal{E} \rightarrow \invariant{l} \parameters{v_0, x_l + t \cdot v_0, \waitvar + t, \globalTime + t, \historyVar} )$}

			\RuleNameRight{solution, assign}
			\UnaryInf{$\lCommit, \invariant{l}, \chi_{vel} \parameters{v_0} \vdash [ \evolution{\globalTime' {=} 1, \waitvar' {=} 1, x_l' {=} v_0}{\waitvar \le \epsilon} ] \invariant{l} \parameters{v_0, x_l, \waitvar, \globalTime, \historyVar}$}

			\RuleNameRight{gtime}
			\UnaryInf{$\lCommit, \invariant{l}, \chi_{vel} \parameters{v_0} \vdash [ \Plant_l \parameters{v_0} ] \invariant{l} \parameters{v_0, x_l, \waitvar, \globalTime, \historyVar}$}
		\end{prooftree}
	\end{small}
	\vspace*{-1em}
	\caption{
		Continues the proofs from \rref{fig:leader_vel_send} and \rref{fig:leader_vel_lose}.
		Thereby $\mathcal{E}$ is short for $\fa{\range{0}{\tilde{t}}{t}} \waitvar \le \tilde{t} \le \epsilon$.
	}
	\label{fig:leader_vel_continuous}
\end{figure}

\begin{figure}[ht!]
	\begin{small}
		\begin{prooftree}[shape=justified]
			\Axiom{$*$}

			\SideAx$\triangleright \RuleName{PA}$
			\RuleNameRight{Id}
			\UnaryInf{$\lCommit, \invariant{l}, \chi_{vel} \parameters{v_0} \vdash \range{0}{v_0}{\maxvelo}$}

			\RuleNameRight{valAccessBase}
			\UnaryInf{$\lCommit, \invariant{l}, \chi_{vel} \parameters{v_0} \vdash \range{0}{\val{\historyVar_\ch{vel} \cdot \comItem{\ch{out}, v_0, \globalTime}}}{\maxvelo}$}

			\RuleNameRight{subsR}
			\UnaryInf{$\lCommit, \invariant{l}, \chi_{vel} \parameters{v_0}, \historyVar_\ch{vel} = \historyVar \cdot \comItem{\ch{out}, v_0, \globalTime} \vdash \range{0}{\val{\historyVar_0}}{\maxvelo}$}
		\end{prooftree}
	\end{small}
	\vspace*{-1em}
	\caption{
		Continues the proof from \rref{fig:leader_vel_send}
	}
	\label{fig:leader_commit}
\end{figure}

\begin{figure}[ht!]
	\begin{small}
		\begin{prooftree}[shape=justified, JustifiedScoresWidth=.85\textwidth]
			\Axiom{$\triangleright$ \rref{fig:leader_vel_continuous}}
			
			\SideAx$\triangleright \RuleName{Id}$
			\UnaryInf{$\lCommit, \invariant{l}, \chi_{vel} \parameters{v_0} \vdash [ \Plant_l \parameters{v_0} ] \invariant{l} \parameters{v_0, x_l, \waitvar, \globalTime, \historyVar}$}
	
			\RuleNameRight{andR}
			\UnaryInf{$\lCommit, \invariant{l}, \chi_{vel} \parameters{v_0} \vdash \lCommit \wedge [ \Plant_l \parameters{v_0} ] \invariant{l} \parameters{v_0, x_l, \waitvar, \globalTime, \historyVar}$}
		
			\RuleNameRight{implR, WL}
			\UnaryInf{$\lCommit, \invariant{l}, \chi_{vel} \parameters{v_0} \vdash \true \rightarrow \lCommit \wedge [ \Plant_l \parameters{v_0} ] \invariant{l} \parameters{v_0, x_l, \waitvar, \globalTime, \historyVar}$}

			\SideAx$\triangleright \RuleName{Id}$
			\RuleNameRight{test}
			\UnaryInf{$\lCommit, \invariant{l}, \chi_{vel} \parameters{v_0} \vdash [ \skipProg ] ( \lCommit \wedge [ \Plant_l \parameters{v_0} ] \invariant{l} \parameters{v_0, x_l, \waitvar, \globalTime, \historyVar})$}
		
			\RuleNameRight{andR}
			\UnaryInf{$\lCommit, \invariant{l}, \chi_{vel} \parameters{v_0} \vdash \lCommit \wedge [ \skipProg ] ( \lCommit \wedge [ \Plant_l \parameters{v_0} ] \invariant{l} \parameters{v_0, x_l, \waitvar, \globalTime, \historyVar})$}
		
			\RuleNameRight{acNoCom}
			\UnaryInf{$\lCommit, \invariant{l}, \chi_{vel} \parameters{v_0} \vdash [ \skipProg ] \acpair{\true, \lCommit} [ \Plant_l \parameters{v_0} ] \invariant{l} \parameters{v_0, x_l, \waitvar, \globalTime, \historyVar}$}
		\end{prooftree}
	\end{small}
	\vspace*{-1em}
	\caption{Continues proof from \rref{fig:leader_vel_send}}
	\label{fig:leader_vel_lose}
\end{figure}

\begin{figure}
	\begin{small}
		\begin{prooftree}[shape=justified]
			\Axiom{$*$}

			\SideAx$\triangleright\RuleName{PA}$
			\RuleNameRight{real}
			\UnaryInf{$\invariant{l}, t\ge0, \waitvar + t \le \epsilon 
			\vdash \waitvar + t = \timespan{\globalTime + t}{\historyVar_\ch{pos}} \le \epsilon \wedge v_l \ge 0 \wedge x_l + t \cdot v_l \le x_l$}

			\RuleNameRight{valAccessBase}
			\UnaryInf{$\invariant{l}, \waitvar = 0, t \ge 0, \mathcal{E} \vdash \invariant{l} \parameters{v_l, x_l + t \cdot v_l, \waitvar + t, \globalTime + t, \historyVar \cdot \comItem{\ch{pos}, x_l, \globalTime}}$}

			\RuleNameRight{subsR}
			\UnaryInf{$\invariant{l}, H_\ch{pos}, \waitvar = 0, t \ge 0, \mathcal{E} \vdash \invariant{l} \parameters{v_l, x_l + t \cdot v_l, \waitvar + t, \globalTime + t, \historyVar_\ch{pos}}$}

			\RuleNameRight{implR}
			\UnaryInf{$\invariant{l}, H_\ch{pos}, \waitvar = 0, t \ge 0 \vdash \mathcal{E} \rightarrow \invariant{l} \parameters{v_l, x_l + t \cdot v_l, \waitvar + t, \globalTime + t, \historyVar_\ch{pos}}$}

			\RuleNameRight{forallR}
			\UnaryInf{$\invariant{l}, H_\ch{pos}, \waitvar = 0 \vdash \fa{t{\ge}0} ( \mathcal{E} \rightarrow \invariant{l} \parameters{v_l, x_l + t \cdot v_l, \waitvar + t, \globalTime + t, \historyVar_\ch{pos}} )$}

			\RuleNameRight{solution}
			\UnaryInf{$\invariant{l}, H_\ch{pos}, \waitvar = 0, \vdash [ \evolution{\globalTime' {=} 1, \waitvar' {=} 1, x_l' {=} v_l}{\waitvar \le 0} ] \invariant{l} \parameters{v_l, x_l, \waitvar, \globalTime, \historyVar_\ch{pos}}$}

			\RuleNameRight{gtime}
			\UnaryInf{$\invariant{l}, H_\ch{pos}, \waitvar = 0 \vdash [ \Plant_l ] \invariant{l} \parameters{v_l, x_l, \waitvar, \globalTime, \historyVar_\ch{pos}}$}

			\SideAx$\triangleright \RuleName{Id}$
			\RuleNameRight{assign, WL}
			\UnaryInf{$\lCommit, \invariant{l}, H_\ch{pos} \vdash [ \waitvar \ceq 0 ] [ \Plant_l ] \invariant{l} \parameters{v_l, x_l, \waitvar, \globalTime, \historyVar_\ch{pos}}$}
	
			\RuleNameRight{acSendRight}
			\UnaryInf{$\lCommit, \invariant{l} \vdash [ \send{\ch{pos}}{\historyVar}{x_l}  ] \acpair{\true, \lCommit} [ \waitvar \ceq 0 ] [ \Plant_l ] \invariant{l} \parameters{v_l, x_l, \waitvar, \globalTime, \historyVar_\ch{pos}}$}

			\RuleNameRight{acNoCom, acWeak}
			\UnaryInf{$\lCommit, \invariant{l} \vdash [ \send{\ch{pos}}{\historyVar}{x_l}  ] \acpair{\true, \lCommit} [ \waitvar \ceq 0 ] \acpair{\true, \lCommit} [ \Plant_l ] \invariant{l} \parameters{v_l, x_l, \waitvar, \globalTime, \historyVar_\ch{pos}}$}

			\RuleNameRight{acComposition}
			\UnaryInf{$\lCommit, \invariant{l} \vdash [ \ctrlUpdate(\historyVar) ] \acpair{\true, \lCommit} [ \Plant_l ] \invariant{l} \parameters{v_l, x_l, \waitvar, \globalTime, \historyVar_\ch{pos}}$}
		\end{prooftree}
	\end{small}
	\vspace*{-1em}
	\caption{
		Continues the proof from \rref{fig:leader_vel_send}.
		We abbreviate $\historyVar_\ch{pos} = \historyVar \cdot \comItem{\ch{pos}, x_l, \globalTime}$ by $H_\ch{pos}$.
		Thereby $\mathcal{E}$ is short for $\fa{\range{0}{\tilde{t}}{t}} \waitvar + \tilde{t} \le \epsilon$.
	}
	\label{fig:leader_pos}
\end{figure}

\begin{figure}[ht!]
	\begin{small}
		\begin{prooftree}[shape=justified]
			\Axiom{$*$}
	
			\SideAx$\triangleright$
			\RuleNameRight{real}
			\UnaryInf{$\invariant{f}, d \le \maxdist, t\ge0 \vdash \invariant{f} \parameters{x_f + t \cdot v_f, \globalTime + t}$}

			\RuleNameRight{WL}
			\UnaryInf{$\invariant{f}, \fAssume, H_\ch{vel}, \fAssume \parameters{\historyVar_0}, d \le \maxdist, t \ge 0 \vdash \invariant{f} \parameters{x_f + t \cdot v_f, \globalTime + t}$}

			\RuleNameRight{forallR}
			\UnaryInf{$\invariant{f}, \fAssume, H_\ch{vel}, \fAssume \parameters{\historyVar_0}, d \le \maxdist \vdash \fa{t{\ge}0} \invariant{f} \parameters{x_f + t \cdot v_f, \globalTime + t}$}

			\RuleNameRight{solution}
			\UnaryInf{$\invariant{f}, \fAssume, H_\ch{vel}, \fAssume \parameters{\historyVar_0}, d \le \maxdist \vdash [ \evolution{\globalTime' {=} 1, x_f' {=} v_f}{non} ] \invariant{f} \parameters{x_f, \globalTime}$}

			\RuleNameRight{gtime}
			\UnaryInf{$\invariant{f}, \fAssume, H_\ch{vel}, \fAssume \parameters{\historyVar_0}, d \le \maxdist \vdash [ \evolution{x_f' {=} v_f}{non} ] \invariant{f} \parameters{x_f, \globalTime}$}
		\end{prooftree}
	\end{small}
	\vspace*{-1em}
	\caption{
		Continues the proof from \rref{fig:follower_in_dist_safe}.
		Formula $H_\ch{vel}$ is short for $\historyVar_0 = \historyVar_\ch{vel} \cdot \comItem{\ch{vel}, \tarvelo, \globalTime}$
	}
	\label{fig:follower_in_dist_unsafe}
\end{figure}

\begin{figure}[ht!]
	\begin{small}
		\begin{prooftree}[shape=justified, JustifiedScoresWidth=.87\textwidth]
			\Axiom{$*$}
	
			\SideAx$\triangleright\text{\RuleName{PA}}$
			\RuleNameRight{real}
			\UnaryInf{$\begin{aligned}
				& \invariant{f}, d_0 = \mespos - x_f, \mespos - x_f \le \maxdist, \orange{0}{v_0}{\safevelo{d_0}}, t \ge 0 \vdash \range{0}{v_0}{\safevelo{d_0}} \\
				&\qquad \wedge v_0 \le \maxvelo \wedge x_f + t \cdot v_0 + (\epsilon - (\globalTime + t - \globalTime)) \safevelo{d_0} < m
			\end{aligned}$}

			\SideAx$\triangleright\text{\RuleName{PA}}$
			\RuleNameRight{projIn, timeAccessBase}
			\UnaryInf{$\begin{aligned}
				& \invariant{f}, d_0 = \mespos - x_f, \mespos - x_f \le \maxdist, \orange{0}{v_0}{\safevelo{d_0}}, t \ge 0 \vdash \range{0}{v_0}{\safevelo{d_0}} \\
				&\qquad \wedge v_f \le \maxvelo \wedge x_f + t \cdot v_0 + (\epsilon - \timespan{\globalTime + t}{\historyVar \cdot \comItem{\ch{pos}, m, \globalTime}}) \safevelo{d_0} < m
			\end{aligned}$}

			\RuleNameRight{projIn, valAccessBase}
			\UnaryInf{$\begin{aligned}
				& \invariant{f}, d_0 = \mespos - x_f, \mespos - x_f \le \maxdist, \orange{0}{v_0}{\safevelo{d_0}}, t \ge 0 \\
				&\qquad \vdash \invariant{f} \parameters{v_0, x_f + t \cdot v_f, d_0, \globalTime + t, \historyVar \cdot \comItem{\ch{pos}, m, \globalTime}}
			\end{aligned}$}

			\RuleNameRight{subsR}
			\UnaryInf{$\Gamma_\ch{mes}, \orange{0}{v_0}{\safevelo{d_0}}, t \ge 0 \vdash \invariant{f} \parameters{v_f, x_f + t \cdot v_f, d_0, \globalTime + t, \historyVar_\ch{pos}}$}

			\RuleNameRight{forallR}
			\UnaryInf{$\Gamma_\ch{mes}, \orange{0}{v_0}{\safevelo{d_0}} \vdash \fa{t{\ge}0} \invariant{f} \parameters{v_f, x_f + t \cdot v_f, d_0, \globalTime + t, \historyVar_\ch{pos},}$}

			\RuleNameRight{solution, assign}
			\UnaryInf{$\Gamma_\ch{mes}, \orange{0}{v_0}{\safevelo{d_0}} \vdash [ \evolution{\globalTime' {=} 1, x_f' =  v_0}{non} ] \invariant{f} \parameters{v_f, x_f, d_0, \globalTime, \historyVar_\ch{pos}}$}

			\RuleNameRight{gtime}
			\UnaryInf{$\Gamma_\ch{mes}, \orange{0}{v_0}{\safevelo{d_0}} \vdash [ \evolution{x_f' =  v_0}{non} ] \invariant{f} \parameters{v_f, x_f, d_0, \globalTime, \historyVar_\ch{pos}}$}
	
			\RuleNameRight{test, implR}
			\UnaryInf{$\Gamma_\ch{mes} \vdash [ \test{\orange{0}{v_0}{\safevelo{d_0}}}  ] [ \Plant_f \parameters{v_0} ] \invariant{f} \parameters{v_f, x_f, d_0, \globalTime, \historyVar_\ch{pos}}$}
	
			\RuleNameRight{nondetAssign, forallR}
			\UnaryInf{$\Gamma_\ch{mes} \vdash [ v_f \ceq * ] [ \test{\orange{0}{v_f}{\safevelo{d_0}}}  ] [ \Plant_f ] \invariant{f} \parameters{v_f, x_f, d_0, \globalTime, \historyVar_\ch{pos}}$}
	
			\SideAx$\triangleright \text{ \rref{fig:follower_mes_dist_not_safe}}$
			\RuleNameRight{composition}
			\UnaryInf{$\Gamma_\ch{mes} \vdash [ v_f \ceq * \seq \test{\orange{0}{v_f}{\safevelo{d_0}}}  ] [ \Plant_f ] \invariant{f} \parameters{v_f, x_f, d_0, \globalTime, \historyVar_\ch{pos}}$}
	
			\RuleNameRight{conditional}
			\UnaryInf{$\invariant{f}, H_\ch{pos}, d_0 = \mespos - x_f \vdash [ \progtt{cond} \parameters{d_0} ] [ \Plant_f ] \invariant{f} \parameters{v_f, x_f, d_0, \globalTime, \historyVar_\ch{pos}}$}
	
			\RuleNameRight{assign}
			\UnaryInf{$\invariant{f}, H_\ch{pos}\vdash [ d \ceq \mespos - x_f ] [ \progtt{cond} \parameters{d} ] [ \Plant_f ] \invariant{f}\parameters{v_f, x_f, d, \globalTime, \historyVar_\ch{pos}}$}
	
			\SideAx$\triangleright\text{\RuleName{trueR}}$
			\RuleNameRight{WL, composition}
			\UnaryInf{$\invariant{f}, \fAssume(\historyVar), H_\ch{pos} \vdash [ d \ceq \mespos - x_f \seq \progtt{cond} \parameters{d} ] [ \Plant_f ] \invariant{f} \parameters{v_f, x_f, d, \globalTime, \historyVar_\ch{pos}}$}
	
			\RuleNameRight{acSendRight}
			\UnaryInf{$\invariant{f} \vdash [ \send{\ch{pos}}{\historyVar}{\mespos} ] \acpair{\fAssume(\historyVar), \true} [ d \ceq \mespos - x_f \seq \progtt{cond} \parameters{d} ] [ \Plant_f ] \invariant{f} \parameters{v_f, x_f, d, \globalTime, \historyVar}$}

			\RuleNameRight{nondetAssign, forallR}
			\UnaryInf{$\invariant{f} \vdash [ \mespos \ceq * ] [ \send{\ch{pos}}{\historyVar}{\mespos} ] \acpair{\fAssume(\historyVar), \true} [ d \ceq \mespos - x_f \seq \progtt{cond} \parameters{d} ] [ \Plant_f ] \invariant{f} \parameters{v_f, x_f, d, \globalTime, \historyVar}$}

			\RuleNameRight{comDual}
			\UnaryInf{$\invariant{f} \vdash [ \receive{\ch{pos}}{\historyVar}{\mespos} ] \acpair{\fAssume(\historyVar), \true} [ d \ceq \mespos - x_f \seq \progtt{cond} \parameters{d} ] [ \Plant_f ] \invariant{f} \parameters{v_f, x_f, d, \globalTime, \historyVar}$}

			\RuleNameRight{acNoCom, acWeak}
			\UnaryInf{$\invariant{f} \vdash [ \receive{\ch{pos}}{\historyVar}{\mespos} ] \acpair{\fAssume(\historyVar), \true} [ d \ceq \mespos - x_f \seq \progtt{cond} \parameters{d} ] \acpair{\fAssume(\historyVar), \true} [ \Plant_f ] \invariant{f} \parameters{v_f, x_f, d, \globalTime, \historyVar}$}
	
			\RuleNameRight{acComposition}
			\UnaryInf{$\invariant{f} \vdash [ \ctrlDistance(\historyVar) ] \acpair{\fAssume(\historyVar), \true} [ \Plant_f ] \invariant{f}$}
		\end{prooftree}
	\end{small}
	\vspace*{-1em}
	\caption{
		Continues the proof from \rref{fig:follower_in_dist_safe}.
		The program $\progtt{cond}$ abbreviates 
		$\ifstat{d\le\maxdist}{%
			\{ v_f \ceq * \seq \test{\orange{0}{v_f}{\safevelo{d}}} \}
		}$.
		Formula $H_\ch{pos}$ is short for $\historyVar_\ch{pos} = \historyVar \cdot \comItem{\ch{pos}, m, \globalTime}$.
		Further, $\Gamma_\ch{mes}$ denotes the formula list $\invariant{f}, H_\ch{pos}, d_0 = \mespos - x_f, d_0 \le \maxdist$.
	}
	\label{fig:follower_mes_dist_safe}
\end{figure}

\begin{figure}[ht!]
	\begin{small}
		\begin{prooftree}[shape=justified, JustifiedScoresWidth=.85\textwidth]
			\Axiom{$*$}
	
			\SideAx$\triangleright\text{\RuleName{PA}}$
			\RuleNameRight{real}
			\UnaryInf{$\begin{aligned}
				& \invariant{f}, \mespos - x_f > \maxdist, t \ge 0 \vdash \range{0}{v_f}{\safevelo{d_0}} \wedge 
				v_f \le \maxvelo \\
				&\qquad \wedge
				x_f + t \cdot v_f + (\epsilon - (\globalTime + t - \globalTime)) \safevelo{d_0} < \mespos
			\end{aligned}$}

			\SideAx$\triangleright\text{\RuleName{PA}}$
			\RuleNameRight{projIn, timeAccessBase}			\UnaryInf{$\begin{aligned}
				& \invariant{f}, \mespos - x_f > \maxdist, t \ge 0 \vdash \range{0}{v_f}{\safevelo{d_0}} \wedge 
				v_f \le \maxvelo \\
				&\qquad   \wedge
				x_f + t \cdot v_f + (\epsilon - \timespan{\globalTime + t}{\historyVar \cdot \comItem{\ch{pos}, m, \globalTime}}) \safevelo{d_0} < \mespos
			\end{aligned}$}

			\RuleNameRight{projIn, valAccessBase}
			\UnaryInf{$\begin{aligned}
				& \invariant{f}, \mespos - x_f > \maxdist, t \ge 0  \vdash \invariant{f} \parameters{v_f, x_f + t \cdot v_f, d_0, \globalTime + t, \historyVar \cdot \comItem{\ch{pos}, \mespos, \globalTime}}
			\end{aligned}$}

			\RuleNameRight{subsR}
			\UnaryInf{$\invariant{f}, H_\ch{pos}, \mespos - x_f > \maxdist, t \ge 0 \vdash \invariant{f} \parameters{v_f, x_f + t \cdot v_f, d_0, \globalTime + t, \historyVar_\ch{pos}}$}

			\RuleNameRight{forallR}
			\UnaryInf{$\invariant{f}, H_\ch{pos}, \mespos - x_f > \maxdist \vdash \fa{t{\ge}0} \invariant{f} \parameters{v_f, x_f + t \cdot v_f, \mespos - x_f, \globalTime + t, \historyVar_\ch{pos}}$}

			\RuleNameRight{solution, assign}
			\UnaryInf{$\invariant{f}, H_\ch{pos}, \mespos - x_f > \maxdist \vdash [ \evolution*{\globalTime' {=} 1, x_f' {=} v_f}{non} ] \invariant{f} \parameters{v_f, x_f, \mespos - x_f, \globalTime, \historyVar_\ch{pos}}$}

			\RuleNameRight{gtime}
			\UnaryInf{$\invariant{f}, H_\ch{pos}, \mespos - x_f > \maxdist \vdash [ \Plant_f ] \invariant{f} \parameters{v_f, x_f, \mespos - x_f, \globalTime, \historyVar_\ch{pos}}$}

			\RuleNameRight{subsL, subsR}
			\UnaryInf{$\invariant{f}, H_\ch{pos}, d_0 = \mespos - x_f, d_0 > \maxdist \vdash [ \Plant_f ] \invariant{f} \parameters{v_f, x_f, d_0, \globalTime, \historyVar_\ch{pos}}$}
		\end{prooftree}
	\end{small}
	\vspace*{-1em}
	\caption{
		Continues the proof from \rref{fig:follower_mes_dist_safe}
	}
	\label{fig:follower_mes_dist_not_safe}
\end{figure}

\end{document}